%% file: aoas-article.tex
\documentclass[aoas, preprint]{imsart}

\RequirePackage{amsthm,amsmath,amsfonts,amssymb}
\RequirePackage[authoryear]{natbib}
\RequirePackage[colorlinks,citecolor=blue,urlcolor=blue]{hyperref}
\RequirePackage{graphicx}

\usepackage{dsfont}
\usepackage{natbib}
\usepackage{cases}
\usepackage{color}
\usepackage{xcolor}
\usepackage{url}
\usepackage{bm}
\usepackage{mathtools}
\usepackage{booktabs}
\usepackage{stmaryrd}
\usepackage{xfrac}
\usepackage{algorithm}
\usepackage[noend]{algorithmic}
\usepackage{thm-restate}
\usepackage{mathrsfs}
\usepackage{graphicx}
\usepackage{subcaption}
\usepackage{tikz}
\usetikzlibrary{trees}
\usetikzlibrary{positioning}
\usepackage{siunitx}
\usepackage{multirow}
\usepackage{xcolor}    % pour \color et \definecolor
\usepackage[normalem]{ulem} % pour \sout
\usepackage{mdframed}  % pour l’environnement mdframed

\startlocaldefs
\theoremstyle{plain}

\newtheorem{theorem}{Theorem}[section]

\theoremstyle{definition}

\newtheorem{proposition}[theorem]{Proposition}

\newtheorem{remark}{Remark}

\endlocaldefs

\makeatletter
\renewcommand*{\@seccntformat}[1]{\csname the#1\endcsname.\quad}
\makeatother

\newcommand{\transp}{\top}
\renewcommand{\leq}{\leqslant}
\renewcommand{\geq}{\geqslant}

\newcommand{\balpha}{\boldsymbol{\alpha}}
\newcommand{\bnu}{\boldsymbol{\nu}}

\newcommand{\by}{\boldsymbol{y}}
\newcommand{\bx}{\boldsymbol{x}}
\newcommand{\ba}{\boldsymbol{a}}
\newcommand{\bz}{\boldsymbol{z}}

\newcommand{\bu}{\boldsymbol{u}}
\newcommand{\bv}{\boldsymbol{v}}
\newcommand{\bs}{\boldsymbol{s}}

\newcommand{\R}{\mathbb{R}}
\newcommand{\I}{\mathcal{I}}
\newcommand{\cA}{\mathcal{A}}

\newcommand{\cC}{\mathcal{C}}

\newcommand{\X}{\mathcal{X}}

\newcommand{\C}{\mathcal{C}}
\newcommand{\D}{\mathcal{D}}
\newcommand{\M}{\mathcal{M}}

\newcommand{\cR}{\mathcal{R}}

\newcommand{\Train}{{\operatorname{Train}}}
\newcommand{\Calib}{{\operatorname{Cal}}}
\newcommand{\Test}{{\operatorname{Test}}}
\newcommand{\err}{{\operatorname{error}}}
\newcommand{\tCalib}{{\widetilde{\operatorname{Cal}}}}

\renewcommand{\epsilon}{\varepsilon}

\newcommand{\E}{\mathbb{E}}
\renewcommand{\Im}{\operatorname{Im}}
\renewcommand{\P}{\mathbb{P}}

\definecolor{darkgreen}{RGB}{0,100,0}

\definecolor{darkblue}{RGB}{25,60,170}
\newcommand{\darkblue}[1]{\textcolor{darkblue}{#1}}

\newcommand{\indic}{\mathbf{1}}
\newcommand{\as}{\qquad \mbox{a.s.}}

\newtheorem{assumption}{Assumption}

\newif\ifrevisions
\revisionstrue      % show revisions
\definecolor{Yannig}{HTML}{7570B3} % violet moyen
\definecolor{Yvenn}{HTML}{E7298A} % magenta/framboise
\definecolor{Bachir}{HTML}{A6761D} % brun doré
\definecolor{JeanMichel}{HTML}{D95F02} % orange brûlé

\newenvironment{authorblock}[1]{%
\begin{mdframed}[
  hidealllines=true,
  leftline=true,
  linewidth=3pt,
  linecolor=#1,
  skipabove=6pt,
  skipbelow=6pt,
  innerleftmargin=8pt,
  innerrightmargin=0pt,
  innertopmargin=2pt,
  innerbottommargin=2pt
]
}{%
\end{mdframed}
}

\begin{document}

\begin{frontmatter}
\title{Probabilistic forecasting of Electric Vehicle charging 
in a time-varying hierarchical infrastructure}
\runtitle{Probabilistic hierarchical forecasting for EV charging demand}
\runauthor{Principato et al.}

\begin{aug}
%%%%%%%%%%%%%%%%%%%%%%%%%%%%%%%%%%%%%%%%%%%%%%%
%% Only one address is permitted per author. %%
%% Only division, organization and e-mail is %%
%% included in the address.                  %%
%% Additional information such as            %%
%% identifying the corresponding author must %%
%% be included in in the Acknowledgments     %%
%% section if necessary.                     %%
%% ORCID can be inserted by command:         %%
%% \orcid{0000-0000-0000-0000}               %%
%%%%%%%%%%%%%%%%%%%%%%%%%%%%%%%%%%%%%%%%%%%%%%%
\author[A,B]{\fnms{Guillaume}~\snm{Principato}\ead[label=e1]{guillaume.principato@universite-paris-saclay.fr}}
\author[A,B]{\fnms{Yvenn}~\snm{Amara-Ouali}\ead[label=e2]{yvenn.amara-ouali@edf.fr}}
\author[A,B]{\fnms{Yannig}~\snm{Goude}\ead[label=e3]{yannig.goude@edf.fr}}
\author[B]{\fnms{Bachir}~\snm{Hamrouche}\ead[label=e4]{bachir.hamrouche@edf.fr}}
\author[A,C]{\fnms{Jean-Michel}~\snm{Poggi}\ead[label=e5]{jean-michel.poggi@u-pariscite.fr}}

%%%%%%%%%%%%%%%%%%%%%%%%%%%%%%%%%%%%%%%%%%%%%%
%% Addresses                                %%
%%%%%%%%%%%%%%%%%%%%%%%%%%%%%%%%%%%%%%%%%%%%%%
\address[A]{Laboratoire de Mathematiques d'Orsay, Université Paris-Saclay, CNRS, INRIA \printead[presep={ ,\ }]{e1}}
\address[C]{Université Paris-Cité \printead[presep={ ,\ }]{e5}}
\address[B]{Electricité de France R\&D \printead[presep={,\ }]{e2,e3,e4}}

\end{aug}

\begin{abstract}
Uncertainty quantification for electric vehicle (EV) 
charging demand is essential to turn the increasing pressure 
of EV adoption on power systems into a source of flexibility. 
We study a large-scale dataset of public charging stations 
in Scotland, where demand exhibits a hierarchical structure, 
organized across multiple aggregation levels, 
from individual charging points to regional and national scales.
The charging infrastructure is continuously evolving, 
with new stations deployed and existing ones occasionally removed,
which calls for forecasting and uncertainty quantification methods 
tailored to this evolving structure.
To address these challenges,
we combine hierarchical forecasting with recent advances in
conformal prediction with missing data. 
The conformal procedures considered produce valid prediction intervals 
at every node of the hierarchy and the algorithm
CP-MNR-Nested$^{\star}$ successfully accounts 
for the time-varying hierarchical structure of the data.
\end{abstract}


\end{frontmatter}

%%%%%%%%%%%%%%%%%%%%%%%%%%%%%%%%%%%%%%%%%%%%%%
%% Please use \tableofcontents for articles %%
%% with 50 pages and more                   %%
%%%%%%%%%%%%%%%%%%%%%%%%%%%%%%%%%%%%%%%%%%%%%%
%\tableofcontents

\section{Introduction} \label{sec:intro}
    The introduction begins by motivating 
    the need for uncertainty quantification 
    in electric vehicle (EV) charging demand 
    forecasting and then turns to the specific 
    challenges posed by EV charging infrastructures, 
    with data collected from ChargePlace Scotland,
    the EV charging public network in Scotland,
    serving as a study case.
    It concludes with a brief overview of related work 
    and the organization of the paper.

    \smallskip

    % \subsection{Motivation}
    % --- Motivation EDF ---
    The large-scale deployment of electric vehicles 
    is reshaping charging infrastructures 
    and electricity consumption patterns, 
    both in magnitude and in temporal structure, 
    and is expected to accelerate further in 
    the coming years \citep{IEA_EV_2022}. 
    This directly concerns distribution system 
    operators and charging network operators, 
    who must plan for and operate a load that 
    keeps growing in scale, 
    evolves as new charging points are continuously deployed, 
    and must be tracked across multiple levels 
    of aggregation, from individual charge points 
    to entire networks 
    \citep{muratori2018impact, powell2022charging}. 

    Beyond these challenges, it also creates new opportunities 
    to leverage EV charging as a flexible resource 
    for power system operations. 
    In particular, smart charging has been identified
    as a key mechanism to help balance supply and demand 
    in the presence of intermittent renewable energy sources 
    \citep{stogl2024electric, bogdanov2024role}. 
    Realizing this potential, however, 
    requires reliable forecasts of EV charging demand 
    \citep{alaraj2025machine, wang2025adaptive}, 
    which are subsequently used as inputs to optimization 
    and control algorithms 
    \citep{fescioglu2023electric, diazlondono2024comparison}.

    % --- Uncertainty quantification on top of point forecasts ---
    In this context, point forecasts alone are not sufficient, 
    as they only
    provide a partial description of the energy demand, 
    and fail to capture its variability. 
    Decision-making under uncertainty is ubiquitous in 
    power system operations \citep{gneiting2014probabilistic}, 
    and EV charging demand is no exception: 
    optimization procedures must account for the best 
    and worst case scenarios 
    to avoid risky decisions that could lead to grid instability 
    or inefficient use of resources \citep{nutkani2024impact}.
    Probabilistic forecasting has therefore become 
    a standard tool for EV charging demand forecasting
    \citep{khan2023comparison,zheng2023coherent,ostermannhaug2024probabilistic,amara2025quantifying}, 
    with the goal of producing prediction intervals around point forecasts that
    cover the true values with a specified probability.

    % \subsection{EV charging infrastructure}
        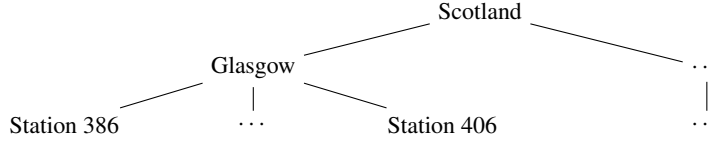
\begin{figure}[t]
            \centering
            \begin{tikzpicture}[
                grow=down,
                edge from parent/.style={draw},
                level 1/.style={sibling distance=6cm,level distance=0.75cm},
                level 2/.style={sibling distance=2.5cm,level distance=0.75cm}
                ]
                \node {Scotland}
                    child {node {Glasgow}
                        child {node {Station 386}}
                        child {node {$\cdots$}}
                        child {node {Station 406}}
                    }
                    % child {node {$\cdots$}
                    %     child {node {$\cdots$}}
                    % }
                    child {node {$\cdots$}
                        child [edge from parent/.style={draw=none}] {node {}}
                        child {node {$\cdots$}}
                        child [edge from parent/.style={draw=none}] {node {}}
                    };
            \end{tikzpicture}
            \caption{Hierarchical structure 
            of EV charging energy demand in Scotland.}
            \label{fig:illustration_hierarchy}
        \end{figure}
    % --- Hierarchical structure of EV charging demand ---
    EV charging demand forecasting is inherently
    a multi-resolution problem, as operational decisions 
    are made across a range of spatial and organizational levels, 
    each tied to a distinct market or grid function. 
    For instance, charge point operators need station-level 
    forecasts to schedule maintenance and manage local 
    queuing and on-site pricing;
    retailers and balance responsible parties 
    need regional aggregates to size day-ahead
    and intraday market positions; 
    and transmission system operators need 
    national aggregates to dimension reserve requirements 
    and maintain system balance. 
    This range, from individual charging stations 
    to regional aggregates and national systems, 
    has been studied at various granularities 
    \citep{kim2021forecasting, yi2022electric, ostermannhaug2024probabilistic}. 
    
    The forecasting task across these granularities 
    is structurally interdependent, 
    as the energy demand at a given level of aggregation 
    is the sum of the energy demand at the lower levels. 
    This construction induces a hierarchical structure, 
    which can be exploited to share information across 
    the different levels of aggregation 
    to improve the accuracy of forecasts at each level 
    of the hierarchy. 
    For instance, \citet{khan2023comparison} compare forecasting 
    accuracy across different spatial aggregation levels, 
    from individual charging piles to their sum, 
    and \citet{zhang2023charging} study a similar hierarchy 
    for the city of Beijing using large-scale 
    real-world driving and charging data, 
    both for point forecasting; 
    \citet{vanetten2024largescale} propose a global model 
    for point forecasting across multiple stations 
    in four cities, while \citet{zheng2023coherent} 
    directly exploit the hierarchical structure 
    by enforcing consistency across a network 
    of charging stations in the case of probabilistic forecasting.

    \begin{figure}[t]
    \centering

    % ===================== GAUCHE =====================
    \begin{minipage}[t]{0.25\textwidth}
    \vspace{0pt}
        \centering

    \includegraphics[
        width=\linewidth,
        trim=90 0 90 0,
        clip
    ]{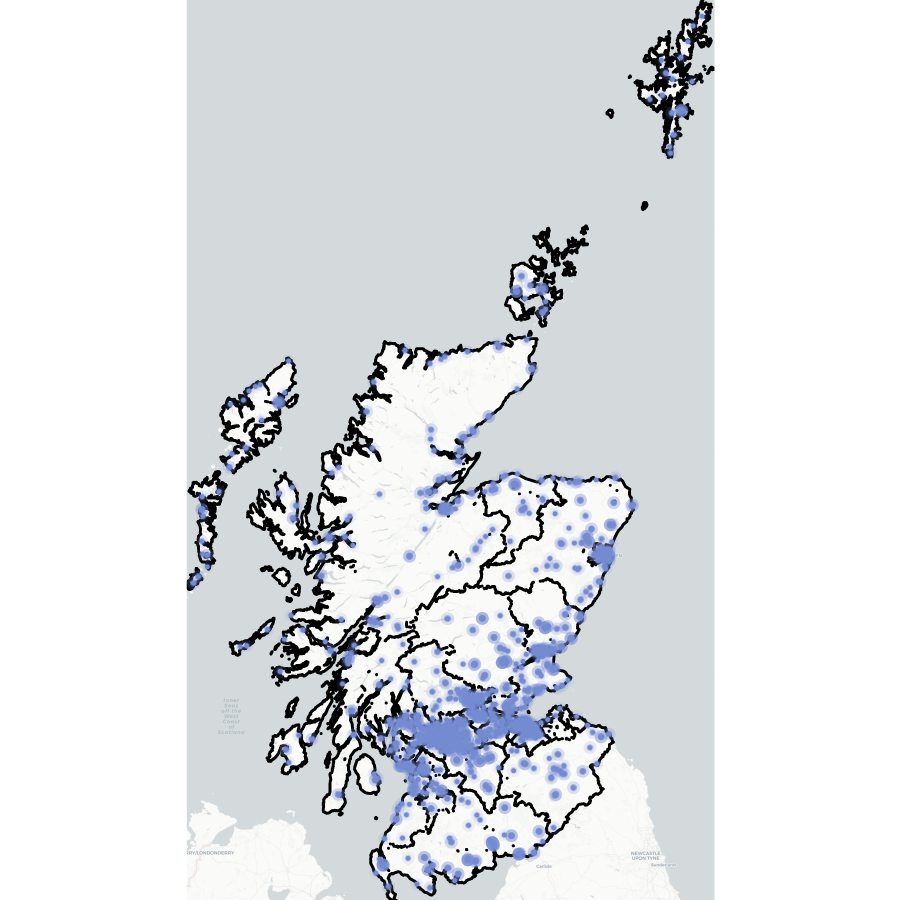}

    \includegraphics[width=\linewidth,
        trim=0 0 0 0,
        clip
    ]{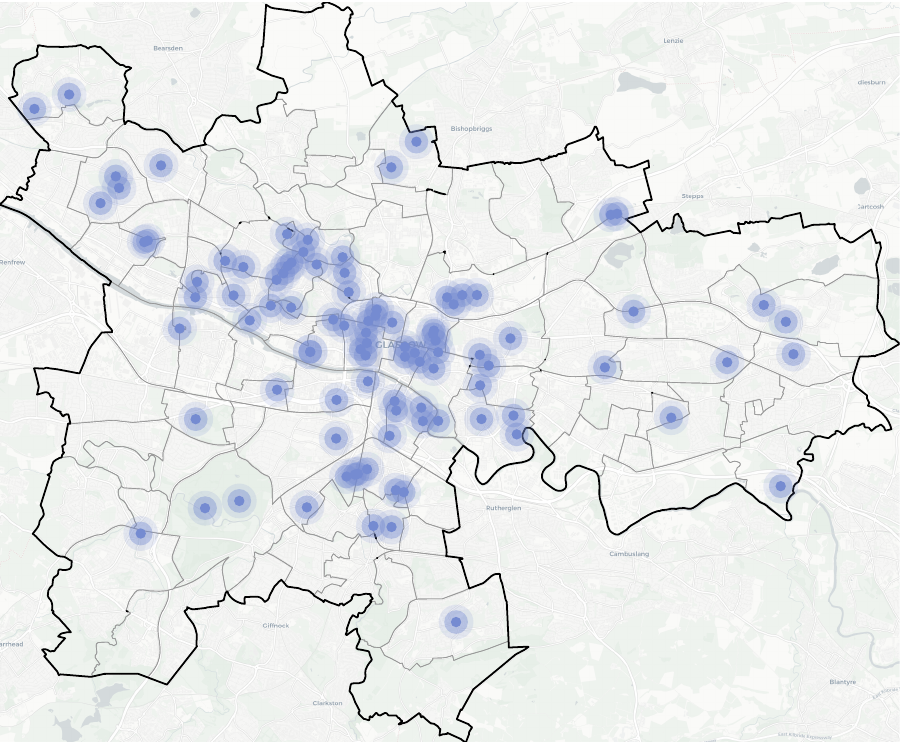}
    
    \end{minipage}
    \hspace{0.25cm}
    % ===================== DROITE =====================
    \begin{minipage}[t]{0.7\textwidth}
    \vspace{0pt}
        \centering

    \includegraphics
    [width=\linewidth,
        trim=5 0 0 0,
        clip
    ]{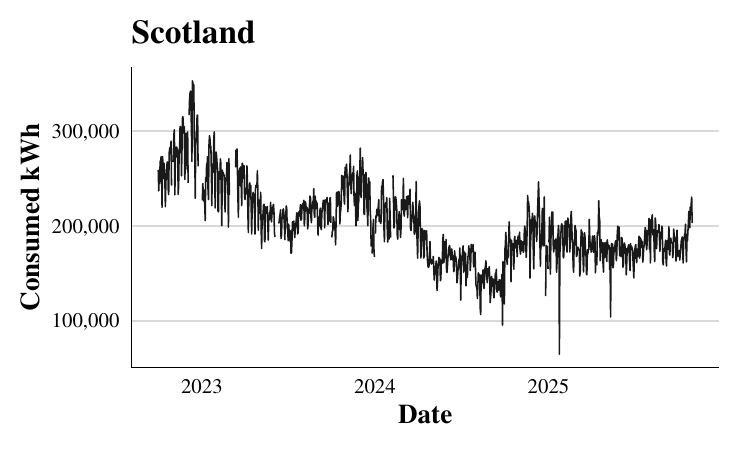}

    \vspace{-0.65cm}

    \includegraphics[width=\linewidth,
        trim=5 0 0 0,
        clip
    ]{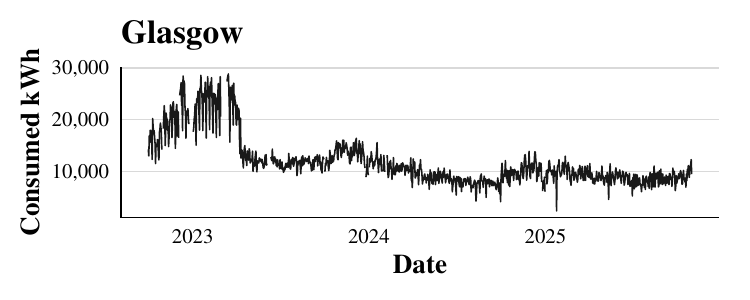}

    \end{minipage}

    \caption{The left panel represents the spatial
    repartition of the charging stations in Scotland (top) 
    and in Glasgow (bottom).
    Stations are represented by blue dots, 
    and the size of the dot is proportional
    to the number of charging points within each station.
    The right panel shows the corresponding time series of daily energy demand.}
    \label{fig:aggregatedseries}    
    \end{figure}

    % --- ChargePlace Scotland dataset is interesting ---
    In this article, we explore a hierarchical structure 
    within the public charging network in Scotland 
    through data from ChargePlace Scotland (see Figure~\ref{fig:illustration_hierarchy}).
    The data, first studied by \citet{davies2024plugging}, 
    exhibit a rich administrative hierarchical structure, 
    with 1398 charging stations distributed 
    across 32 local authorities,
    and a national aggregate at the country level.
    It embodies a central modeling challenge in EV demand forecasting: 
    charging infrastructure evolves over time.
    The adoption of electric vehicles drives 
    the deployment of new charging stations, 
    while existing stations may be temporarily 
    or permanently removed from the network, 
    e.g., due to maintenance or changes in ownership. 
    Since only active stations contribute to aggregate demand,
    it induces a so-called time-varying hierarchical structure.
    These features raise two modeling challenges. 
    The first is to forecast demand for 
    new stations, with limited or no historical data. 
    The second is to account for the current state 
    of the infrastructure, 
    both for point forecasting and 
    for uncertainty quantification.
    The latter being an important source of variability 
    for EV charging demand.

    % --- Additional interesting aspects of the dataset ---
    In addition to these modeling considerations, 
    data from ChargePlace Scotland exhibit 
    several challenging features
    documented by \citet{davies2024plugging}. 
    Despite a relatively high penetration 
    of electric vehicles in Scotland,
    total energy demand at the national level shows 
    a decreasing trend over the period of interest 
    (see the right hand side of Figure~\ref{fig:aggregatedseries}).
    This pattern reflects changes in user 
    behavior driven by external factors,
    including progressive charging tariffs
    introduced as part 
    of policy incentives. 
    It is particularly visible in the city of Glasgow, 
    where the energy demand dropped on April 4th, 2023,
    following the introduction of the new pricing scheme 
    in the public charging network
    (see the right hand side of Figure~\ref{fig:aggregatedseries}).
    The introduction of the tariffs leads 
    to a heterogeneous response across the country, 
    as it depends on the local authority, 
    and even on the ownership 
    of the charging station \citep{davies2024plugging}. 
    These effects add to other sources of non-stationarity 
    inherent to EV charging demand time series, 
    at both aggregated and disaggregated levels.

    % \subsection{Statistical related work}
    % --- Base forecast ---
    Operational forecasting of EV charging demand 
    requires models that not only achieve high predictive accuracy 
    but also satisfy practical constraints 
    such as fast and scalable inference.
    For tabular data and multivariate time series, 
    it is well known that tree ensembles are highly competitive 
    and easy-to-calibrate forecasting models. 
    For instance, frameworks such as Random Forests \citep{breiman2001random}, 
    and XGBoost \citep{chen2015xgboost}, are widely used in machine learning 
    competitions \citep{makridakis2022, bojer2021kaggle} 
    and industrial applications \citep{goude2013local}. 
    Recently, these models have been challenged by deep learning-based approaches.
    In the context of EV charging, \citet{vanetten2024largescale} 
    consider a global model to forecast energy demand across 55 charging stations in Boulder, 
    Palo Alto, London, and Perth, demonstrating the effectiveness of deep learning architectures 
    like Transformers and N-HiTS \citep{challu2023nhits}. 
    Alongside deep learning, pre-trained foundation models have driven significant progress 
    in predictive modeling for both tabular and time series data. 
    Models such as TabPFN \citep{hollmann2025tabpfn} and TabICL \citep{qu2025tabicl} 
    achieve state-of-the-art performance on benchmark tasks like FevBench \citep{shchur2025fev} 
    and TabArena \citep{erickson2026tabarena}. 
    Similar to gradient boosting machines, 
    these foundation models can be readily adapted to time series forecasting 
    by incorporating time-dependent features or lagged variables. 
    In this work, we utilize TabICL due to its open-source availability and high inference efficiency. 
    However, while these automated machine learning frameworks are highly efficient, 
    many applied settings require domain experts to design interpretable models 
    by integrating physical or systemic knowledge directly into the model equations. 
    This ensures a necessary trade-off between predictive performance and interpretability. 
    Generalized Additive Models (\citealp[GAMs]{wood2017generalized}) naturally facilitate this integration, 
    making them highly popular in power system management \citep{pierrot2011short, de2023adaptive, antoniadis2024statistical}. 
    To achieve a complete and up-to-date overview, 
    we propose considering these three types of forecasting models, 
    along with a simple combination of them, at all three levels of the hierarchy (see Section~\ref{sec:predictors}).

    % --- Related work on forecast reconciliation ---
    The aforementioned forecasting methods 
    are tailored for the large number of 
    charging stations in the dataset,
    but they do not account for the 
    hierarchical structure of the data.
    This is a common limitation in the 
    context of hierarchical forecasting,
    which is commonly addressed by 
    producing forecasts that are agnostic 
    to the hierarchical structure,
    and to subsequently adjust them to 
    ensure that they follow the same
    aggregation constraints of the data.
    A well-established framework for this adjustment is 
    forecast reconciliation, introduced by
    \citet{athanasopoulos2009hierarchical,hyndman2011optimal}. 
    The idea is to combine forecasts from different levels
    of aggregation through linear combinations to 
    produce coherent forecasts with 
    respect to the hierarchical structure.
    A natural way to achieve this is through
    a linear projection of the base forecasts 
    onto a so-called coherent subspace \citep{panagiotelis2021forecast,wickramasuriya2019optimal}. 
    Techniques of forecast reconciliation have been applied 
    in the context of EV charging demand forecasting 
    by \citet{zheng2023coherent} to produce probabilistic forecasts
    based on reconciled scenarios. 

    % --- Conformal prediction ---
    Recently, conformal prediction has emerged as a particularly 
    attractive framework 
    to quantify uncertainty around point forecasts,
    as it produces prediction sets with finite-sample 
    and distribution-free coverage guarantees 
    under mild assumptions \citep{vovk2005algorithmic,lei2018distribution,barber2021predictive}.
    Its appeal in the context of EV charging demand forecasting 
    \citep{amara2025quantifying} lies in its flexibility: 
    similar to forecast reconciliation,
    conformal prediction can be applied post-hoc
    to any base predictor, 
    including those obtained from complex machine learning models, 
    without the need for retraining \citep{angelopoulos2023gentle}. 
    This conceptual similarity suggests combining them 
    to produce hierarchy-aware prediction intervals,
    as proposed by \citet{principato2024conformal}.
    In particular, the authors show that forecast reconciliation 
    preserves the validity of conformal prediction intervals,
    while potentially yielding tighter prediction sets,
    thereby reducing uncertainty around the forecasts.
    A recent line of work consists in 
    extending conformal prediction to accommodate 
    non-stationarity in time series settings.
    \citet{gibbs2021adaptive} introduce 
    adaptive conformal inference, that updates
    the conformal sets to account for temporal distribution shifts.
    Adaptive conformal inference has been applied in a variety
    of energy forecasting settings, including
    electricity price forecasting by \citet{zaffran2022adaptive}, 
    and EV charging demand forecasting by \citet{amara2025quantifying}.

    \smallskip

    The rest of the paper is organized as follows. 
    Section~\ref{sec:dataset} provides further details 
    about the dataset and formalizes 
    the forecasting problem.
    The modeling emphasizes the hierarchical structure,
    its temporal evolution, and the specific 
    objectives it induces.
    Section~\ref{sec:preliminaries} is devoted to a
    pedagogical exposition of the main ingredients of the methodology,
    namely forecast reconciliation and conformal prediction
    in the context of hierarchical forecasting.
    Section~\ref{sec:methodology} presents the proposed methodology.
    It describes how to combine forecast reconciliation 
    and conformal prediction in order to account
    for the time-varying hierarchical structure
    when quantifying uncertainty.
    Finally, Section~\ref{sec:experiments} describes the experimental setup and
    summarizes the performance of the different methods on the ChargePlace Scotland data.
        
    \section{EV charging energy demand in Scotland} \label{sec:dataset}
    Data is collected from ChargePlace Scotland, 
    a publicly available record of charging sessions 
    from the largest public EV charging network in Scotland.
    The resulting dataset is referred to as
    the ChargePlace Scotland dataset.
    The code used for the experiments is available on the following~\href{https://github.com/PrincipatoG/Probabilistic-Hierarchical-EV}{GitHub repository}.

    \subsection{The ChargePlace Scotland dataset} \label{sec:data}
    % --- Data collection ---
    The period considered goes from October 2022 to October 2025, 
    extracting session-level information and merging it 
    with spatial and technical characteristics of the charging infrastructure, 
    following \citet{gridcharge2025}.
    The spatial attributes include the geographical coordinates of each charging point,
    the address of the charging station, and the local authority to which it belongs.
    Sessions for which these spatial attributes are missing are excluded from the analysis.
    Energy consumption at the session level is then aggregated to obtain daily energy demand 
    at each location with available geographic information, 
    thereby constructing a panel of daily demand series at the charging-station level.

    % --- Description of the variables ---    
    The charging stations are indexed by $i = 1, \ldots, n$, 
    where $n = 1398$ is the total number of stations.
    For each day $t$, let $y_i^{(t)} \in \R$ be the daily energy demand at the $i$-th station. 
    We consider the natural geographical/administrative partition of the $n$ stations into 
    two levels of aggregation, namely the local authority and the country levels.
    For instance, the local authority of Glasgow City contains 96 charging stations, 
    and its daily energy demand is thus the sum of the daily energy demand of these 96 stations.
    This forms a hierarchical structure, 
    which yields a total of $m = 1431$ nodes (1398 stations, 32 local authorities, and 1 national node).
    The hierarchical structure is illustrated in Figure~\ref{fig:illustration_hierarchy}, 
    with the example of the local authority of Glasgow City.
    This grouping is motivated by 
    the fact that local authorities are the level at which 
    charging infrastructure is planned, funded, 
    and reported by Scottish public bodies 
    (e.g., Transport Scotland, ChargePlace Scotland), 
    making them a natural intermediate level between 
    individual stations and the national system 
    for both operational and policy purposes. 
    This level also coincides with existing administrative 
    and demographic statistics (population, population density, 
    car ownership), which can be 
    used as exogenous covariates. 

    Nevertheless, the proposed approach can be adapted 
    to other contexts (e.g., by considering finer territorial 
    partitions such as districts or quarters).
    Let $\by^{(t)} \in \R^m$ be the vector of daily energy demands 
    over all nodes at day $t$,
    and consider a multivariate regression setting, 
    where the goal is to forecast demand at every level of the hierarchy 
    from a vector of features $\bx^{(t)} \in \X$.
    This can be formalized as a supervised learning problem, 
    where the objective is to learn a function $\widehat{f} : \X \to \R^m$ 
    that is an estimator of the conditional expectation $\E\bigl[\by^{(t)} \mid \bx^{(t)}\bigr]$.
    Vectors of features include weather information, 
    calendar effects, lagged demand, and characteristics of the node of interest 
    (e.g., number of charging points, type of chargers and tariffs indicators),
    that are known at the time of forecasting. 
     
    % The notations introduced in this section are summarized in Table~\ref{tab:notations}.
        \begin{table}[t]
            \centering
            \caption{Basic notations.}
            \label{tab:notations}

            \renewcommand{\arraystretch}{1.25}
            \begin{tabular}{p{0.125\linewidth} p{0.125\linewidth} p{0.60\linewidth}}
            \hline

            \multicolumn{1}{l}{\textbf{Indices}} & \textbf{Domain} & \textbf{Description} \\[0.3em]

            $t$ & Days &time index, on a daily basis \\

            $i$  & $[m]$ & node index, $i=1,\dots,m$: \\
            & & \qquad stations, $i=1,\dots,n$ \\
            & & \qquad local authorities aggregate, $i=n+1,\dots,m-1$ \\
            & & \qquad Scotland aggregate, $i=m$ \\
            
            \hline
            \multicolumn{1}{l}{\textbf{Variables}} & \textbf{Domain} & \textbf{Description} \\[0.3em]
            
            $y_i^{(t)}$  & $\R$ & energy demand at node $i$ and day $t$ \\
            $a_i^{(t)} $ & $\{0,1\}$ & infrastructure activity at node $i$ and day $t$ \\[0.75em]

            $\bx^{(t)} $ & $\X$ & vector of features at day $t$ \\
            $\by^{(t)} $ & $\mathbb{R}^m$ & vector of energy demand at day $t$ \\
            $\ba^{(t)}$ & $\{0,1\}^m$ & vector of infrastructure activity at day $t$ \\[0.75em]

            $H$ & $\{0,1\}^{m \times n}$ & hierarchical structure of the complete data \\
            $H^{(t)}$ & $\{0,1\}^{m \times n}$ & refined hierarchical structure at day $t$, i.e., with respect to $\ba^{(t)}$ \\
            \hline
            % \multicolumn{2}{l}{\textbf{Operators}} \\[0.3em]
            % $\widehat{\cdot}$ & quantities obtained from a base model \\
            % $\widetilde{\cdot}$ & quantities obtained from a reconciled model \\
            % \hline
            % \multicolumn{1}{l}{\textbf{Operators}} & \textbf{Description} \\[0.3em]
            % $[n]$ & set of integers $\{1,\dots,n\}$ \\
            % $|S|$ & cardinality of a set $S$ \\
            % $\indic_A$ & indicator function of an event $A$ \\
            % $\Im(M)$ & image of a matrix $M$ \\
            % $M^{\dagger}$ & Moore-Penrose pseudo-inverse of a matrix $M$ \\
            % \hline
            \end{tabular}

        \end{table}

    % --- Hierarchical structure ---
    To model the hierarchical structure of
    EV charging demand in Scotland, 
    we adopt a standard notation in hierarchical forecasting
    and encode the aggregation constraints 
    by a binary structural matrix $H \in \{0,1\}^{m \times n}$.
    Its entries satisfy 
    $H_{i,j}=1$ if the $j$-th station belongs to the $i$-th node, 
    and $H_{i,j}=0$ otherwise--see the toy-example depicted in~\eqref{eq:example_refinement}. 
    Therefore, for all days $t$, the vector energy demand at day $t$
    belongs to the so-called coherent subspace
    $\Im(H) = \bigl\{H\bu : \bu \in \R^n\bigr\} \subset \R^m$.
    Forecast reconciliation
    exploits this property by projecting 
    unconstrained base forecasts, obtained from arbitrary regression models, 
    onto $\Im(H)$ to enforce coherence
    \citep{athanasopoulos2009hierarchical,hyndman2011optimal}.
    
    \subsection{Modeling the infrastructure activity} \label{sec:hierarchical_structure}

    % --- Time varying hierarchical structure ---
    To capture the temporal evolution of the network, 
    EV charging demand is modeled through 
    a hierarchical structure with a time-varying 
    set of active nodes.
    The energy demand at some locations--especially 
    at the station level--may 
    be zero during a certain amount of time
    due to various foreseeable reasons 
    (see Figure~\ref{fig:glasgow_evolvingstructure}).
    For example, if a station does not exist yet,
    or is temporarily closed due to malfunction, renovation work, act of vandalism, 
    or is definitely closed due to change of ownership or decommissioning.
    These foreseeable events are conceptually 
    distinct from other zeros that may occur 
    due to the absence of demand, which are essentially concentrated
    in stations located in rural areas. 
    In the present work, the goal is to leverage 
    the knowledge of the infrastructure activity
    to improve the accuracy of the forecasts,
    rather than to forecast when a station will be active or not.
    Therefore, is it reasonable to assume that 
    the information of whether a node is active
    (i.e., has non-zero demand) is available at the time of forecasting.
    We thus consider that the infrastructure activity is known at day $t$,
    and focus on forecasting the non-zero coordinates of $\by^{(t)}$.
    This information is represented by binary vectors of infrastructure activity,
    defined for each day $t$ as
    \begin{equation}\label{eq:mask}
        \ba^{(t)} \in \{0,1\}^m \,, \qquad \text{ such that } \qquad 
            \begin{cases}
                a_{i}^{(t)} = 0\,, & \text{ if } y_{i}^{(t)} = 0 \,, \\
                a_{i}^{(t)} = 1\,, & \text{ otherwise} \,.
            \end{cases} 
    \end{equation} 

            \begin{figure}[t]
            \centering

            \begin{subfigure}[b]{0.48\textwidth}
                \centering
                \includegraphics[width=\textwidth]{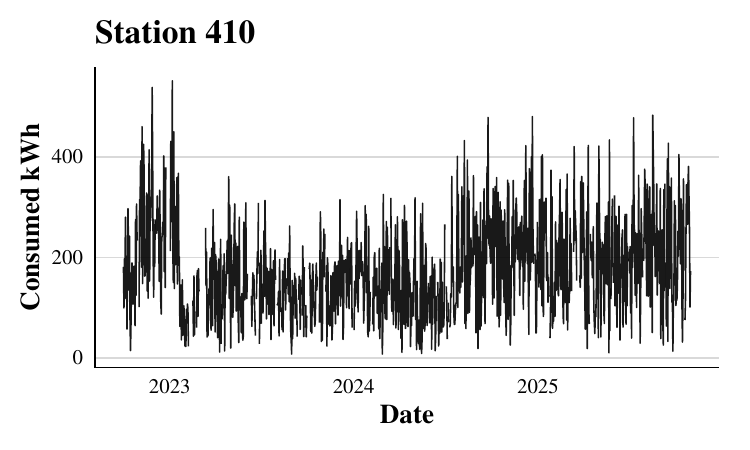}
            \end{subfigure}
            \hfill
            \begin{subfigure}[b]{0.48\textwidth}
                \centering
                \includegraphics[width=\textwidth]{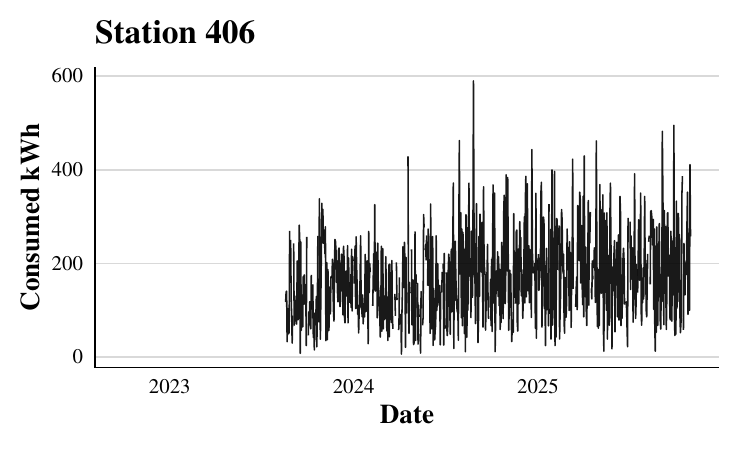}
            \end{subfigure}

            \vspace{-0.35cm}

            \begin{subfigure}[b]{0.48\textwidth}
                \centering
                \includegraphics[width=\textwidth]{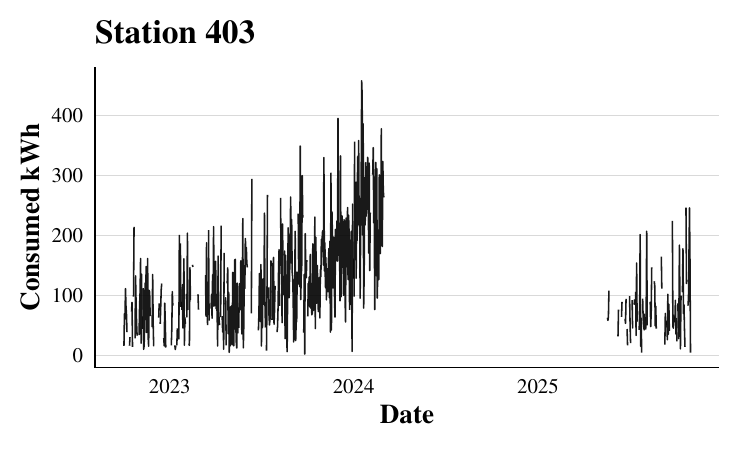}
            \end{subfigure}
            \hfill
            \begin{subfigure}[b]{0.48\textwidth}
                \centering
                \includegraphics[width=\textwidth]{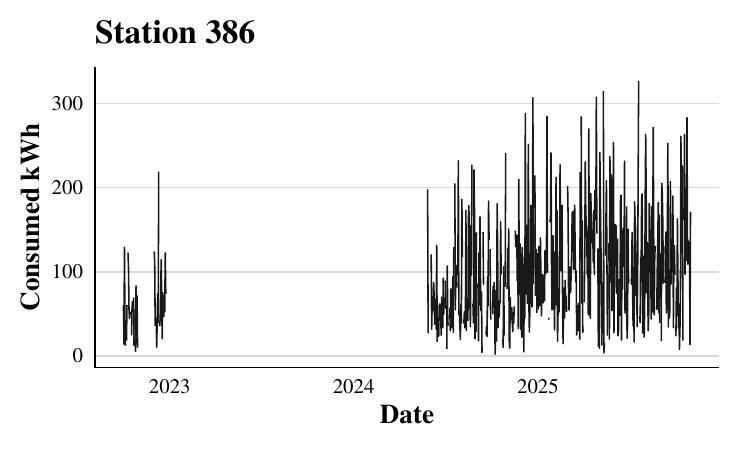}
            \end{subfigure}

            \caption{Time series of daily energy demand at four charging stations 
            in Glasgow City. The figure illustrates three typical node behaviors:
            Station 410 is active throughout the entire period, Station 406 becomes inactive at the end of 2023,
            Stations 403 and 386 both exhibit the effect of predictable temporary closures.
            }
            \label{fig:glasgow_evolvingstructure}
    \end{figure}

    % --- Role of the infrastructure activity ---
    We distinguish the vector of infrastructure activity 
    from the vector of features $\bx^{(t)} \in \X$.
    In particular, infrastructure activity 
    is used only within the learning procedure 
    to exclude inactive nodes 
    from the training data, 
    since these correspond to structural zeros 
    that are not targets of interest and may 
    otherwise introduce bias in the resulting forecasts.
    It is also used post-hoc to refine the hierarchical 
    constraints of the data at day $t$.
    We now formalize this refinement of the hierarchical structure.

    Inactive nodes do not contribute to the aggregate demand,
    which implies that enforcing coherence with respect to the original hierarchy 
    may introduce artificial constraints. 
    This motivates the definition of a
    refined structural matrix and the associated 
    refined coherent subspace for each day $t$.
    The refinement of the coherent subspace removes 
    the inactive components and defines 
    a lower-dimensional coherent subspace 
    that better reflects the operational structure on a given day.
    Precisely, the vector of energy demand
    at day $t$ belong to the potentially 
    smaller coherent subspace $\Im\bigl(H^{(t)}\bigr) \subseteq \Im(H)$, 
    where $H^{(t)}$ is defined as follows:
    \begin{equation}\label{eq:structural_matrix}
        H^{(t)}_{i,j} =     
            \begin{cases}
            0\,, & \text{ if } a_i^{(t)} = 0 \,, \\
            0\,, & \text{ if } \exists i' \in [m] \,, \mbox{ s.t. } \, a_{i'}^{(t)} = 0 \, \mbox{ and } \, H_{i'\!,j} = 1 \,, \\
            H_{i,j}\,, & \text{ otherwise} \,.
            \end{cases}
    \end{equation}

    For example, consider a simplified hierarchy with only three stations and one region.
    Then, if the vector of infrastructure activity at day $t$ is
    $\ba^{(t)} = (1, 0, 1, 1)^{\transp}$, the 
    following refinement of the structural matrix holds.
       \begin{equation} \label{eq:example_refinement}
              H = \begin{bmatrix}
                     1 & 0 & 0 \\
                     0 & 1 & 0 \\
                     0 & 0 & 1 \\
                     1 & 1 & 1
              \end{bmatrix} \qquad \mbox{may be refined into} \qquad 
              H^{(t)} = \begin{bmatrix}
                     1 & 0 & 0 \\
                     0 & 0 & 0 \\
                     0 & 0 & 1 \\
                     1 & 0 & 1
              \end{bmatrix} \,.
       \end{equation} 
    
    \subsection{Coverage objectives} \label{sec:objcov}
    % --- Preamble ---
    Infrastructure activity is a key driver 
    of EV charging demand and must therefore 
    be accounted for when assessing forecast uncertainty. 
    This is addressed by extending the usual marginal 
    coverage requirement to a conditional setting.
    
    % --- Description of UQ ---
    The goal is to produce reliable forecasts of 
    EV charging demand at each node of the hierarchy,
    i.e., for each station, local authority, 
    and the national aggregate,
    on a time period indexed by $t \in \Test$, 
    and to quantify the uncertainty around these forecasts.
    Given a nominal coverage level $1-\alpha = 90\%$, a 
    classic requirement is to produce
    prediction intervals at node $i$ 
    containing the true demand 
    for at least $90\%$ 
    of the days in $\Test$. 
    For instance, if infrastructure activity 
    is properly accounted for (e.g., through forecast reconciliation), 
    there is no uncertainty about the prediction of energy
    demand at an inactive station. 
    Hence, marginal coverage may appear satisfactory 
    simply because inactive periods are trivially covered, 
    while coverage on active periods remains poor.
    Similarly, the uncertainty quantification
    at an active station may depend on the activity 
    of other stations in the hierarchy.
    For example, a temporary closure of neighboring 
    stations may lead to a surge in demand at the station 
    under consideration, resulting in a higher variability 
    of the energy demand. As a consequence, 
    assessing coverage by averaging over all days in $\Test$
    may not be representative of the underlying uncertainty.

    A natural solution is to extend the coverage requirement 
    to a conditional one, where coverage is assessed by averaging 
    over subsets of $\Test$ corresponding to
    each specific configuration of the infrastructure.
    This is infeasible in practice, 
    as the number of possible configurations 
    is large and only few observations are available for each.
    Therefore, we consider two weaker conditional coverage requirements, 
    where coverage is assessed by averaging over subsets of $\Test$
    corresponding to ranges of similar infrastructure activity patterns.
    The first is a local conditioning, 
    where coverage is assessed by averaging 
    over the days for which the node under consideration is active.
    This is the minimal requirement for a meaningful coverage assessment,
    as we are only interested in forecasting
    the energy demand at active nodes. 
    In the rest of the article, 
    this condition is referred to as validity.
    The second is a stronger conditioning,
    where coverage is assessed by averaging over 
    the days for which the node under consideration is active,
    and the number of active nodes in its neighborhood--defined 
    by the local authority--partitions the time
    indices into two subsets: low and high activity regimes.
    The rationale for this additional conditioning 
    is that the activity of the nodes within 
    the same local authority is an important source 
    of variability for the accuracy of the forecasts.
    Note that other conditioning could be considered,
    for instance by changing the definition of the neighborhood,
    or the number of activity regimes.
    This second requirement is referred to as conditional coverage
    or, equivalently, conditional validity.

    \subsection{Efficiency objectives} \label{sec:objeff}
    % --- Efficiency as a additional requirement ---
    Beyond validity, the intervals should also be efficient,
    in the sense that they should be as narrow as possible.
    Indeed, intervals that are too wide are
    valid as they always contain the true value, 
    but will not be representative of the underlying uncertainty
    and will not be useful for decision-making.
    For instance, the prediction interval $[0 \operatorname{kWh}, 10000\operatorname{kWh}]$  
    for the energy demand at a station will
    always contain the true value, but the
    station operator should not prepare 
    for a demand of $10000 \operatorname{kWh}$ that is not realistic.
    Therefore, the efficiency of the prediction intervals 
    is measured by their average length.
    As for validity, the average length is computed 
    only over the days for which the node
    under consideration is active.

    \section{Point and probabilistic forecasting} \label{sec:preliminaries}
    Forecast reconciliation and conformal prediction 
    are two post-hoc frameworks that can be 
    combined to forecast hierarchical data, such as
    EV charging demand and provide the associated prediction intervals.
    We first review the two forecast reconciliation 
    procedures considered in this article, in
    the context of point forecasting,
    and then describe how they can be integrated 
    within conformal prediction.

    \subsection{Two forecast reconciliation procedures}
    \label{sec:forecastreconciliation}
    We use OLS reconciliation 
    (ordinary least squares, \citealp{hyndman2011optimal}) 
    and introduce rOLS (refined OLS),
    its natural extension that accounts 
    for the time-varying hierarchical structure
    modeling introduced in
    Section~\ref{sec:hierarchical_structure}.
    Both procedures perform reconciliation 
    through orthogonal projection:
    OLS projects base forecasts 
    onto the fixed coherent subspace $\Im(H)$,
    whereas rOLS projects them, at day $t$,
    onto the refined coherent subspace 
    $\Im\bigl(H^{(t)}\bigr)$.

    Given a vector of base forecasts
    $\widehat{\by}^{(t)} \in \R^m$, 
    the following closed-form expressions 
    hold for the reconciled forecasts
    obtained from OLS and rOLS, respectively,
    \begin{equation}
        \widetilde{\by}_{\operatorname{OLS}}^{(t)} = H \bigl(H^{\transp} H\bigr)^{-1} H^{\transp} \widehat{\by}^{(t)} \,, 
        \qquad \mbox{and} \qquad
        \widetilde{\by}_{\operatorname{rOLS}}^{(t)} = H^{(t)} \bigl(H^{(t)\transp} H^{(t)}\bigr)^{\dagger} H^{(t)\transp} \widehat{\by}^{(t)} \,,
    \end{equation}
    where $\dagger$ denotes the Moore-Penrose pseudo-inverse.
    
    The impact of using rOLS instead of OLS 
    is illustrated on the
    toy-example, with $n = 3$ stations, $m=4$ nodes,
    the structural matrix given in
    ~\eqref{eq:example_refinement},
    and the infrastructure activity vector 
    $\ba^{(t)} = (1,0,1,1)^{\transp}$ at day $t$.
    The following numerical example 
    of EV charging demand and base forecasts, 
    induces the corresponding OLS and rOLS reconciled forecasts:
    \begin{equation}
        \by^{(t)} = \begin{bmatrix}
                     26 \\
                     0 \\
                     32 \\
                     58
              \end{bmatrix} \,; \qquad 
              \widehat{\by}^{(t)} = \begin{bmatrix}
                     20 \\
                     20 \\
                     30 \\
                     65
              \end{bmatrix} \,; \qquad
              \widetilde{\by}_{\operatorname{OLS}}^{(t)} = \begin{bmatrix}
                     18.75 \\
                     18.75 \\
                     28.75 \\
                     66.25
              \end{bmatrix} \,; \qquad
              \widetilde{\by}_{\operatorname{rOLS}}^{(t)} = \begin{bmatrix}
                     25 \\
                     0 \\
                     35 \\
                     60
              \end{bmatrix} \,.
    \end{equation}

    Apart from the second coordinate, 
    which is zero and thus excluded from the forecasting task, 
    the rOLS reconciled forecasts are closer 
    to the true values than both the base forecasts 
    and the OLS reconciled forecasts.
    Indeed, OLS propagates 
    non-zero value produced 
    by the base forecast at the inactive station.
    It affects the other coordinates and, 
    in particular, leads to an inflated 
    last coordinate relative to the true value.
    By contrast, rOLS constructs a reconciled vector 
    in such a way that none of its components 
    depends on the second coordinate of the base forecast, 
    which is effectively set to zero, 
    and remaining coordinates are adjusted accordingly.

    \subsection{Basics on conformal prediction for hierarchical data}
    \label{sec:hierarchicalconformal}

            \begin{table}[t]
            \centering
            \caption{Advanced notations.}
            \label{tab:notations2}
            \renewcommand{\arraystretch}{1.25}
            \begin{tabular}{p{0.11\linewidth} p{0.44\linewidth} p{0.36\linewidth}}
            \hline
            \multicolumn{1}{l}{\textbf{Functions}} & \textbf{Domain} & \textbf{Description} \\[0.3em]
            
            $\widehat{A}$  & $ \X \times \{0,1\}^m \to \mathbb{R}^m$  & base predictor \\
            $\widetilde{A}$ & $\X \times \{0,1\}^m \to \mathbb{R}^m$ & reconciled predictor \\[0.3em]
            
            $\widehat{R}$ & $\R^m \times \{0,1\}^m \to \R^m$ & reconciliation function \\[0.3em]
            \hline
            \multicolumn{1}{l}{\textbf{Variables}} & \textbf{Formula} & \textbf{Description} \\[0.3em]
            $\widetilde{\by}^{(k)} \in \R^m$ & $\widetilde{A}\bigl(\bx^{(k)}, \ba^{(k)}\bigr) 
            = \widehat{R}\Bigl(\widehat{A}\bigl(\bx^{(k)}, \ba^{(k)}\bigr), \ba^{(k)}\Bigr)$ & reconciled forecasts at day $k$ \\
            $\bs^{(k)} \in \R^m$ & $\bigl|\widetilde{\by}^{\,(k)} - \by^{(k)}\bigr| \,\,\,\, = \,s\Bigl(\bigl(\bx^{(k)},\ba^{(k)}\bigr), \, \by^{(k)}; \widetilde{A}\Bigr) $ & vectorial non-conformity score at day $k$ \\[0.75em]

            $\widetilde{y}_i^{\,(k)} \!\in \R$ & $\Pi_i\bigl(\widetilde{\by}^{\,(k)}\bigr)$ & reconciled forecasts for node $i$ at day $k$ \\[0.15em]
            $s_{i}^{(k)} \in \R$ & $ \Bigl| \widetilde{y}_i^{\,(k)} - y_i^{(k)} \Bigr|$ & non-conformity score for node $i$ at day $k$ \\[0.3em]
            \hline
            \end{tabular}

        \end{table}

    Conformal prediction for hierarchical data 
    \citep{principato2024conformal} is a natural
    framework to quantify uncertainty in the context of
    hierarchical forecasting.
    Originally introduced for fixed hierarchical structures,
    it can easily be adapted 
    to handle time-dependent reconciliation 
    procedures such as rOLS, as described
    in Section~\ref{sec:forecastreconciliation}.
    We recall the main ingredients of the algorithm,
    which we refer to as CP (for Conformal Prediction),
    and introduce the notations that will be used 
    in the subsequent sections.
    The procedure is illustrated in Figure~\ref{fig:diagram}.
    
    CP is a multivariate extension of inductive 
    conformal prediction, also known as split conformal prediction
    \citep{vovk2005algorithmic,lei2018distribution},
    where prediction sets are 
    constructed component-wise and
    the procedure include a reconciliation step.
    The first step of CP consists in 
    splitting the available data into 
    a training and a calibration set,
    indexed by $\Train$
    and $\Calib$, respectively.
    The training set is used to fit a base predictor 
    $\widehat{A} : (\bx, \ba) \to \widehat{\by}$, 
    and a reconciliation function 
    $\widehat{R} : (\widehat{\by}, \ba) \to \widetilde{\by}$, 
    resulting in a reconciled predictor $\widetilde{A}(\bx, \ba) \to \widetilde{\by}$ defined as
    the combination of the base predictor and 
    the reconciliation function:
    \begin{equation}
    \begin{aligned}
    \widetilde{A}:\quad
    &\X \times \{0,1\}^m
    &&\longrightarrow \quad \R^m \,, \\
    &\quad (\bx,\ba)
    &&\longmapsto \quad 
    \widehat{R}\bigl(\widehat{A}(\bx, \ba), \ba\bigr) \,.
    \end{aligned}
    \end{equation}
    The calibration set is then 
    used to compute non-conformity scores,
    that measure how well the reconciled forecasts 
    fit the calibration observations at each node $i \in [m]$.
    We consider the standard non-conformity score 
    function that consists of the
    absolute difference between 
    the reconciled forecasts and the true values,
    i.e.,
    \begin{equation}
    \label{eq:conformity_score_i}
    \begin{aligned}
    \quad \! s_i:\quad
    &\bigl(\X \times \{0,1\}^m\bigr)\times \R^m \times \mathcal{F}
    &&\longrightarrow \quad \R \,, \\
    &\qquad \Bigl((\bx,\ba), \, \by; \, \widetilde{A}\Bigr)
    &&\longmapsto \quad
    \bigl| y_i - \widetilde{y}_i \bigr| =
    \bigl|\Pi_i\bigl(\by-\widetilde{A}(\bx,\ba)\bigr)\bigr| 
    \,,
    \end{aligned}
    \end{equation}
    where $\mathcal{F} = 
    \bigl\{f : \X \times \{0,1\}^m \to \R^m \bigr\}$ is
    the set of predictors, and $\Pi_i$ is
    the operator that extracts 
    the $i$-th component of a vector in $\R^m$.

    Denoting by $s_i^{(k)}$ the non-conformity score
    for node $i$ at day $k \in \Calib$, the
    $(1-\alpha)$--empirical quantile
    of the calibration non-conformity scores is computed,
    up to the usual finite-sample correction, namely,
    \begin{equation}
    \label{eq:quantile_i}
    q_i = \min\Bigl\{ q \in \R : \#\bigl\{ k \in \Calib :
    s_i^{(k)} \leq q \bigr\} \geq (1-\alpha) \bigl(\#\Calib+1\bigr) \Bigr\} \,,
    \end{equation} 
    where the cardinality of a set $S$ is denoted by $\#S$.

    Finally, at each node $i \in [m]$, 
    the resulting conformal set 
    is the prediction interval centered around $\widetilde{y}_i^{\,(t)}$,
    the $i$-th component of the reconciled forecasts 
    for the new observation, and with a radius 
    equal to the empirical quantile $q_i$, i.e.,
    \begin{equation}
    \label{eq:conformal_set_i}
    \forall i \in [m]\,, \qquad
    \C_i\bigl(\bx^{(t)}, \ba^{(t)}\bigr) 
    = \Bigl[\widetilde{y}_i^{\,(t)}\pm q_i\Bigr] \,.
    \end{equation}

    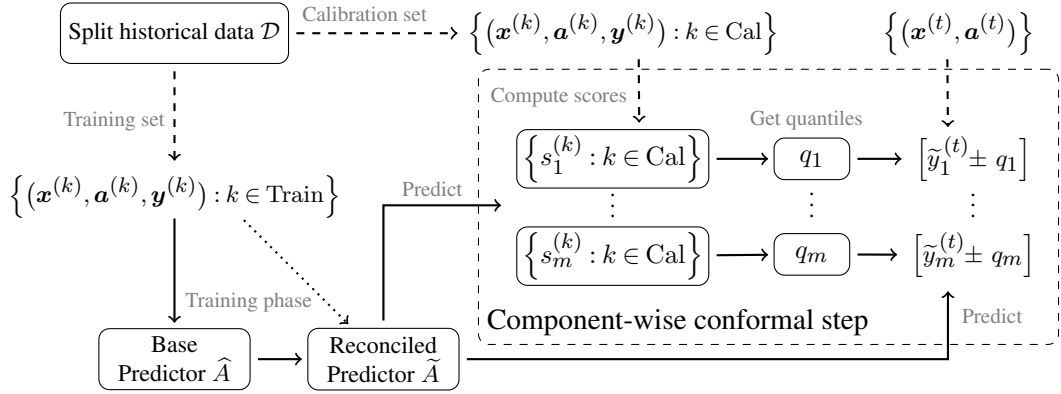
\begin{figure}[t]
        \centering
        \input{aoas-article-tikz.tex}    
        \caption{The CP procedure, adapted from conformal prediction for hierarchical data.
        % Data is split into a training set and a calibration set.
        % A base model is fitted on the training set 
        % to produce a base predictor, which combined with
        % a reconciliation step (that may depend on the training set)
        % yields a reconciled predictor.
        % The latter is used to compute non-conformity scores 
        % on the calibration set, for each node $i \in [N]$ of the hierarchy.
        % Finally, the prediction intervals are constructed at each node $i$
        % as the set of candidate values for which the score 
        % is below an empirical quantile of the calibration scores.
        }
        \label{fig:diagram}
    \end{figure}

    \section{Conformal prediction 
    for time-varying hierarchical structures} \label{sec:methodology}

        % --- Motivation ---
    The conformal procedure described in Section~\ref{sec:hierarchicalconformal}
    encompasses forecast reconciliation methods
    that account for the infrastructure activity,
    but the construction of the prediction intervals 
    is not adapted to the time-varying nature of the hierarchy.
    In particular, the method is designed to produce prediction
    intervals that are--only--valid marginally 
    over all possible configurations 
    of the infrastructure activity for
    the test day $t$.
    This is because the non-conformity scores 
    are computed over the entire calibration set, 
    which includes both active and inactive nodes.
    It is problematic in the case 
    of EV charging demand, 
    as the non-conformity scores
    for inactive nodes stations trivially zero 
    for a reconciled predictor 
    based rOLS, leading to prediction intervals
    that are too small for active nodes that
    have been inactive during the calibration period.

    % --- Outline ---
    This issue is addressed in the subsequent sections
    by proposing three extensions of the CP procedure.
    The first extension, referred to as CP-MNR 
    (for Conformal Prediction 
    with Missing Nodes Reconciliation),
    is designed to produce prediction intervals 
    that are valid conditionally on the node being active,
    a minimal requirement for a meaningful coverage assessment.
    The second extension, CP-MNR-Nested$^\star$,
    is designed to produce prediction intervals 
    that are valid conditionally 
    on the entire configuration of the infrastructure activity,
    a more ambitious goal.
    Finally, the third extension, Adaptive CP-MNR, 
    is designed to produce prediction intervals
    that are valid conditionally on the node being active,
    even when the data exhibit a distribution shift between 
    the calibration and test points.

\subsection{Prediction intervals for active nodes} \label{sec:CP-MNR}
% --- General presentation of the procedure ---
CP-MNR is a simple extension of CP,
designed to produce prediction 
intervals that are valid conditionally on the node being active.

It only differs from the CP procedure in 
that the non-conformity scores are computed
on a subsampled set of calibration indices.
For each node $i \in [m]$,
the subsampled set of calibration indices
is denoted by $\tCalib(i) \subseteq \Calib$
and is defined as the subset of calibration indices
corresponding to observations that share the same 
activity status of the i-th node as the test point, namely,
\begin{equation}
    \forall i \in [m]\,, \qquad \tCalib(i) = \bigl\{k \in \Calib : a_i^{(k)} = a_i^{(t)}\bigr\} \,.
\end{equation}
% --- Construction of the prediction intervals ---
The rest of the procedure remains unchanged 
(i.e., follows the steps described 
in Section~\ref{sec:hierarchicalconformal}).
For each node $i \in [m]$, 
non-conformity scores are 
computed on the restricted calibration set,
empirical quantile at level $1-\alpha$
of these scores is computed, 
and the resulting prediction interval
is defined based on this quantile 
as in~\eqref{eq:conformal_set_i}.

The CP-MNR procedure described above 
is detailed 
in Algorithm~\ref{algo:CP-MNR}.
% --- Mondrian conformal prediction ---
Its construction is reminiscent 
of the Mondrian conformal prediction 
framework \citep{vovk2005algorithmic},
which aims to achieve validity 
conditional on a categorical variable
(here the binary activity of the node 
under consideration).
The theoretical guarantees of 
CP-MNR are therefore similar to those of
Mondrian conformal prediction,
and are stated in the Supplementary 
Material~\ref{app:thCP-MNR}.
The fitting of the predictors 
involved in the first step of Algorithm~\ref{algo:CP-MNR} 
is stated formally in the 
Supplementary Material~\ref{app:implementation:predict-then-reconcile}, 

\begin{algorithm}[t]
	\caption{\label{algo:CP-MNR} CP-MNR at level $1-\alpha$}
	\begin{algorithmic}
    \STATE\textbf{Input:} Dataset $\D = \Bigl\{ \bigl(\bx^{(k)}, \ba^{(k)}, \by^{(k)}\bigr) : k \in \Train \cup \Calib \Bigr\}$; 
        algorithms $\cA$ and $\cR$; 
        test point $\bigl({\bx^{(t)}}, {\ba^{(t)}}\bigr)$.
	\vspace{.25em}
    \STATE 1: Fit the predictors on $\bigl\{\D; \, \Train; \, \cA; \, \cR \bigr\}$ and get $\Bigl\{\widehat{A}(\cdot, \cdot), \, \widehat{R}(\cdot, \cdot), \, \widetilde{A}(\cdot, \cdot) \Bigr\}$  \\
    \vspace{.35em}
    \STATE \textbf{for} $i \in [m]$ \textbf{do} perform Mondrian conformal prediction:
    \vspace{.35em}
    \STATE \qquad 2: Consider a subsampled set of calibration indices
    \quad $\tCalib(i) = \Bigl\{ k \in \Calib : a_i^{(k)} = a_i^{(t)} \Bigr\}$ \\
    \vspace{-.15em}
    \STATE \qquad 3: Compute the non-conformity scores, for $k \in \tCalib(i)$,
    \quad  $s_i^{(k)} = s_i\Bigl(\bigl(\bx^{(k)}, \,\ba^{(k)}\bigr), \, \by^{(k)}; \, \widetilde{A}(\cdot, \cdot)\Bigr)$ \\
    \vspace{.25em}
    \STATE \qquad 4: Compute their $(1-\alpha)$--empirical quantile, defined as \\
    \vspace{.35em}
    \qquad \qquad \qquad $\widetilde{q}_i = \displaystyle{\min\biggl\{ q \in \R : \#\Bigl\{k \in \tCalib(i) : s_i^{(k)} \leq q \Bigr\} \geq (1-\alpha)\bigl(\#\tCalib(i)+1\bigr) \biggr\}}$ \\
    \vspace{-.15em}
    \STATE \qquad 5: Build the prediction interval, \quad $\cC_i\bigl({\bx^{(t)}}, {\ba^{(t)}}\bigr)
    = \Bigl[ \Pi_i\Bigl(\widetilde{A}\bigl(\bx^{(k)}, \ba^{(k)}\bigr)\Bigr) \pm \widetilde{q}_i \Bigr]$ \\
    \vspace{.25em}
    \STATE \textbf{end for} \\
    \STATE \textbf{Output:} Predictions intervals $\Bigl\{\cC_i\bigl({\bx^{(t)}}, {\ba^{(t)}}\bigr): i \in [m] \Bigr\}$.
\end{algorithmic}
\end{algorithm}

\begin{remark} \label{rk:coincides}
    In the case of a node $i \in [m]$ that is always active (e.g., an aggregated node),
    the CP-MNR procedure coincides with the CP procedure of Section~\ref{sec:hierarchicalconformal}.
    In Section~\ref{sec:experiments}, CP is thus omitted
    from the numerical experiments for the aggregated nodes.
\end{remark}

\subsection{Infrastructure-aware prediction intervals}\label{sec:CP-MNR-Nested}
As discussed in Section~\ref{sec:objcov}, 
conditioning on the node being active is 
a minimal but not sufficient requirement
as whole infrastructure activity 
is also an important source of variability
for the accuracy of the forecasts.
In this section, the CP-MNR procedure is 
extended to produce prediction intervals
that are valid conditionally on the entire configuration
of the infrastructure activity.

In principle, infrastructure activity 
vector is also a categorical variable 
so that Mondrian conformal prediction
could be applied.
Nevertheless, this would require to 
partition the calibration set into
$2^m$ subsets, one for each possible 
configuration of the infrastructure activity,
which exceeds the number of available observations 
by many orders of magnitude.
This practical limitation advocates 
for a procedure that accounts for the whole 
infrastructure activity without requiring 
a separate calibration set for each configuration:
CP-MNR-Nested$^\star$.

Before introducing CP-MNR-Nested$^\star$,
for the sake of a more pedagogical exposition,
we first rewrite the conformal step of CP-MNR 
in a equivalent form that emphasizes 
the role of infrastructure activity 
in the construction of the prediction intervals.
For all node $i\in [m]$,
the $i$-th prediction interval output
by Algorithm~\ref{algo:CP-MNR}
can be equivalently written 
as the set of candidate values
whose associated non-conformity score 
is strictly larger to at most a fraction $(1-\alpha)$ 
of the non-conformity scores computed 
on the calibration set, 
up to the usual finite-sample correction, i.e.,
\begin{equation*}
\begin{aligned}
\label{eq:conformal_set_i_equiv}
\C_i\bigl(\bx^{(t)}, \ba^{(t)}\bigr) 
= \Biggl\{ &\Pi_i(\by) \in \R : (1-\alpha )\bigl(\#\tCalib(i)+1\bigr) > \\[-2.ex]
& \,\, \, \#\biggl\{k \in \tCalib(i) : s_i\Bigl(\bigl(\bx^{(k)}, \, \ba^{(k)}\bigr), \, \by^{(k)}; \, \widetilde{A} \Bigr)
< \,s_i\Bigl(\bigl({\bx^{(t)}}, \, {\ba^{(t)}}\bigr), \, \by; \, \widetilde{A}\Bigr) \biggr\} \Biggr\} \,.
\end{aligned}
\end{equation*}
This formulation highlights that
the non-conformity scores computed
on the calibration set only
depend on the infrastructure activity of the test point
through the selection of the calibration indices,
which is a local conditioning on 
the activity of the node under consideration.

% --- Rationale of the methodology ---
The rationale of CP-MNR-Nested$^\star$ is to
enforce a symmetric treatment of calibration 
and test points with respect to the 
infrastructure activity,
by evaluating non-conformity 
scores under a shared activity pattern. 
The resulting procedure is stated in
Algorithm~\ref{algo:CP-MNR-Nested*}, 
and the symmetric treatment of calibration 
and test points is highlighted in blue.
We now describe the three steps of the procedure in more details.

The first step of CP-MNR-Nested$^\star$ is identical to that of CP-MNR, 
where the predictors are fitted on the training set.
In the second step, 
for each calibration index $k \in \Calib$,
a common augmented 
infrastructure activity pattern 
${\widetilde{\ba}^{(k)}} = 
\min \bigl( \ba^{(k)}, {\ba^{(t)}} \bigr)$
is considered,
where the minimum is taken component-wise.
% (for all $i \in [m]$).
It corresponds to
a hypothetical scenario, where the only active nodes 
are those that are active 
in both the calibration and test points.

% --- Nested formulation of the prediction set ---
The third step consists in computing the non-conformity scores
of the augmented calibration points
$\bigl(\bx^{(k)}, {\widetilde{\ba}^{(k)}}, \by^{(k)}\bigr)$
to compare them with those
of the temporary test points 
$\bigl({\bx^{(t)}}, {\widetilde{\ba}^{(k)}}, \by \bigr)$,
for all candidate values $\by \in \R^m$.
This comparison, performed under the same augmented 
infrastructure activity pattern determines 
the cardinality of the following sets: 
\begin{equation}
    \forall \by \in \R^m\,, \qquad \biggl\{ k \in \Calib : s_i\Bigl(\bigl(\bx^{(k)}, \, {\widetilde{\ba}^{(k)}}\bigr), \, \by^{(k)}; \, \widetilde{A} \Bigr) 
    < \, s_i\Bigl(\bigl({\bx^{(t)}}, \, {\widetilde{\ba}^{(k)}}\bigr), \, \by; \, \widetilde{A}\Bigr) \biggr\} \,.
\end{equation}
It is then used to define the prediction intervals,
as in the above reformulation of CP-MNR.

\begin{algorithm}[t]
	\caption{\label{algo:CP-MNR-Nested*} CP--MNR--Nested$^{\star}$ at level $1-\alpha$}
	\begin{algorithmic}
    \STATE\textbf{Input:} 
    Dataset $\D = \Bigl\{ \bigl(\bx^{(k)}, \ba^{(k)}, \by^{(k)}\bigr) : k \in \Train \cup \Calib \Bigr\}$; 
        algorithms $\cA$ and $\cR$; 
        test point $\bigl({\bx^{(t)}}, {\ba^{(t)}}\bigr)$; 
        subsampled set of calibration indices $\widetilde{\Calib} \subseteq \Calib$.
	\STATE 1: Fit the predictors on $\bigl\{\D; \, \Train; \, \cA; \, \cR \bigr\}$ and get $\Bigl\{\widehat{A}(\cdot, \cdot), \, , \widehat{R}(\cdot, \cdot), \, \widetilde{A}(\cdot, \cdot) \Bigr\}$  \\
	\STATE 2: \textbf{for} $k \in \widetilde{\Calib}$ \textbf{do} consider an augmented infrastructure pattern \\
       \vspace{.25em}
        \STATE \qquad \darkblue{${\widetilde{\ba}^{(k)}} = \min \bigl( \ba^{(k)}, {\ba^{(t)}} \bigr)$} \\
       \vspace{.25em}
        \STATE \textbf{end for} \\
    \STATE 3: \textbf{for} $i \in [m]$ \textbf{do} build the prediction set $\cC_i\bigl({\bx^{(t)}}, {\ba^{(t)}}\bigr)$ defined as \\
       \STATE \qquad $\cC_i\bigl({\bx^{(t)}}, {\ba^{(t)}}\bigr) = \Biggl\{ \Pi_i(\by) \in \R : (1-\alpha)\bigl(\#\widetilde{\Calib}+1\bigr) >$ \\
       \vspace{-0.5em}
       \STATE \hspace{10.5em} $\displaystyle{ 
        \#\biggl\{k \in \tCalib : s_i\Bigl(\bigl(\bx^{(k)}, \, \darkblue{\widetilde{\ba}^{(k)}}\bigr), \, \by^{(k)}; \, \widetilde{A}(\cdot, \cdot) \Bigr)
        < \,s_i\Bigl(\bigl({\bx^{(t)}}, \, \darkblue{\widetilde{\ba}^{(k)}}\bigr), \, \by; \, \widetilde{A}(\cdot, \cdot)\Bigr) \biggr\} 
       } \Biggr\}$ \\
        \vspace{.25em}
       \STATE \textbf{end for} \\
    \STATE \textbf{Output:} Predictions sets $\Bigl\{\cC_i\bigl({\bx^{(t)}}, {\ba^{(t)}}\bigr): i \in [m] \Bigr\}$.
\end{algorithmic}
\end{algorithm}

% --- Pub Supplementary Material B ---
%\begin{remark}
This construction does not describe
how to implement the procedure.
The Supplementary Material~\ref{app:implementation:CP-MNR-Nested*} 
is devoted to this purpose:
it explains how to derive the conformal set, 
why a subsampled set of calibration indices
$\widetilde{\Calib} \subseteq \Calib$ 
should be used,
and discusses the nature of the resulting prediction set
as well as the computational cost of the procedure.
%\end{remark}

The construction of CP-MNR-Nested$^\star$ 
is inspired from the missing
data augmentation (MDA) procedure of
\citet{zaffran2023conformal,zaffran2024predictive}
and is actually an extension to our setting
of their algorithm, CP-MDA-Nested$^\star$ 
(\citealp[Algorithm~1]{zaffran2024predictive}).
% --- Pub Supplementary Material A ---
We establish a connection between the two frameworks
in the Supplementary Material~\ref{sec:connection},
and show that the theoretical guarantees of CP-MDA-Nested$^\star$ 
can be extended to CP-MNR-Nested$^\star$
(see the Supplementary Material~\ref{app:thCP-MNR-Nested*}).
Even though some of the assumptions 
that are natural for the former 
are not satisfied in the latter,
the theoretical guarantees of CP-MNR-Nested$^\star$
provide the heuristic behind the construction of the procedure.

\subsection{Online extension of the methodology} \label{sec:adaptive}
% --- Need for a sequential extension ---
The procedures described in Sections~\ref{sec:hierarchicalconformal}, 
\ref{sec:CP-MNR} and~\ref{sec:CP-MNR-Nested} 
(CP, CP-MNR and CP-MNR-Nested$^\star$)
are designed for a single test point.
The underlying assumption behind their construction 
follows the standard conformal prediction framework,
where the main assumption
is that the calibration points and the
test point are drawn independently 
from the same distribution.
In the context of EV charging demand forecasting,
this setting needs to be extended 
to a sequential setting, defined as follows.

% --- Sequential setting ---
At each round $t \in \Test$,
the practitioner observes the covariates and 
the infrastructure activity vector
$\bigl(\bx^{(t)}, \ba^{(t)}\bigr)$,
and output for each node $i$ a prediction set 
$\C_i^{(t)} = \C_i\bigl({\bx^{(t)}}, {\ba^{(t)}}\bigr)$
for the unknown value $y_i^{(t)}$, 
which is then revealed at the end of the round.
In this setting, no stochastic assumption 
is made on the sequence of test points
to account for the temporal dependence and distributional 
shifts that are inherent to EV charging demand data,
as illustrated in Figures~\ref{fig:aggregatedseries} 
and~\ref{fig:glasgow_evolvingstructure}.

% --- Need for an adaptive extension ---
These attributes advocate for an 
online extension of the conformal
procedures through
adaptive conformal inference (ACI, 
\citealp{gibbs2021adaptive}). 
% --- High-level description of the adaptive extension ---
The principle of the adaptive extension is to update the 
conformal intervals sequentially based on the recent miscoverage events.
For each node $i \in [m]$, the fixed nominal level $1-\alpha = 90\%$
used in the conformal step is replaced 
by a node-specific level, updated sequentially over the test period.
If node $i$ is inactive at the current test point, the coverage level is left unchanged.
Otherwise, the coverage level increases if the previous conformal interval 
did not cover the true value, leading to a wider interval, 
and decreases if the previous conformal interval did
cover the true value, leading to a narrower one.

% --- Benefits of the adaptive extension and why only on CP-MNR ---
This strategy can be used to extend each
of the three conformal procedures described 
in the previous sections.
Nevertheless, the coverage level is updated 
only for active nodes, which is consistent 
with the construction of the CP-MNR procedure
that is designed to account for the activity 
of the node under consideration,
but not with the construction of CP and 
of CP-MNR-Nested$^{\star}$.
In particular, this adaptive extension
may hinder the stronger conditional coverage 
guarantee targeted by the 
CP-MNR-Nested$^{\star}$ procedure.
Therefore, only the adaptive extension of CP-MNR is considered,
and is referred to as Adaptive CP-MNR.
Its formal statement and the proof of its theoretical guarantees
are deferred to the Supplementary Material~\ref{app:implementation:Adaptive-CP-MNR}.

\section{Forecasting EV charging energy demand with confidence} \label{sec:experiments}

This section describes the study case of EV charging demand 
forecasting on the ChargePlace Scotland data
and depict how the proposed methodology addresses the 
specific challenges of this application.
We first introduce the experimental setting,
including the rolling-window protocol and the data preprocessing choices.
Then, the most appropriate base predictor and reconciliation strategy 
are selected and fixed for the subsequent experiments.
Finally, the performance of the proposed conformal 
procedures is evaluated and analyzed.

\subsection{Experimental setting}

% --- Sequential treatment ---
The proposed methodology is assessed on a six-month test period, 
denoted by $\Test$, from February 5th to August 18th, 2024.
The period exhibits a significant temporal drift 
in the data distribution (Figure~\ref{fig:aggregatedseries}), 
making it a particularly challenging setting for forecasting. 
As discussed in Section~\ref{sec:adaptive},
the CP, CP-MNR, and CP-MNR-Nested$^{\star}$ procedures 
are designed for a single test point. 
To deploy them sequentially over the entire test period, 
one could in principle retrain the underlying models 
at each time step, at a substantial computational cost, 
or, at the other extreme, keep the models fixed throughout the whole period. 
We adopt an intermediate strategy and update the models every two weeks
to account for the dynamic nature of the data.
Accordingly, the training and calibration periods
are updated using a rolling-window scheme, 
illustrated in Figure~\ref{fig:windows}.

% --- Rolling-window scheme ---
The variability of the performance metrics 
is evaluated by running the experiments 10 times 
with different rolling-window partitions of the historical data.
The resulting mean estimates and standard errors 
are reported throughout this section.
For each experimental run, 
the available data before the first test block 
are split into a training set and a calibration set. 
This partition is then shifted forward by two weeks 
at each subsequent test block.
The calibration set represents 20\% of the historical data 
and is constructed as follows: 
it includes the week preceding the test block 
together with one-week blocks sampled uniformly at random 
from the remaining time indices.
This partitioning affects the prediction intervals 
through a change in the predictors fitted on the training set
and in the calibration points used to compute the non-conformity scores. 

% --- Preprocessing choices ---
Before fitting the predictors, 
the data is preprocessed to remove sessions 
with missing geographical information and, 
across all nodes, dates with anomalous data 
reporting as identified by \citet{davies2024plugging}. 
Subsequently, features are engineered for the regression 
tasks across the three hierarchical levels.
Calendar features encode weekly and yearly cycles 
using day-of-week and day-of-year variables (scaled between 0 and 1). 
Public holidays are recoded and treated as Sundays.
Weather features are collected by requesting data with the \texttt{riem package} \citep{riem}.
Daily maximum temperature and daily average humidity are calculated at each node. 
At the station level, these meteorological features are derived 
using a distance-weighted average from weather stations. 
For regional and national levels, the station level features 
are aggregated using a weighted average based on the number 
of charging points per station.
% To explicitly capture non-linear thermo-sensitivity, 
% a binary indicator for temperatures exceeding 15$^\circ$C 
% is also included in some models that cannot capture it natively. 
With the same weighted average scheme, charger locations 
(longitude and latitude) are considered
both at the station and regional levels.
Equipment features are also included, 
such as the proportion of slow, accelerated, and fast chargers at each node
as well as a dimensioning factor for the number of charging points
and paid charging points,
estimated dynamically using the prior day's number of 
used chargers and the maximum number of used chargers over the previous 30 days.
Charging cost features account for the transition of chargers from 
free to paid described in Section~\ref{sec:intro}.
Historical energy consumption incorporates the prior day's energy, 
the energy consumed on the same day of the previous week, 
and a 1-to-7-day rolling mean.
Days with zero demand are set to NA (i.e., treated as missing values)
so that predictors are fitted exclusively on active nodes. 
To prevent the propagation of these NAs into the lags 
and thus avoid reducing the number of available observations, 
any NA in a lagged variable is imputed using the last available 
observation of the same day type. 
Finally, depending on the model, a shifted logarithmic transformation 
is applied to the target and the lags at the regional and station levels. 
This prevents the model from favoring larger nodes while ensuring 
numerical stability for nodes with near-zero demand.

    \begin{figure}[t]
        \centering
        \includegraphics[width=\textwidth]{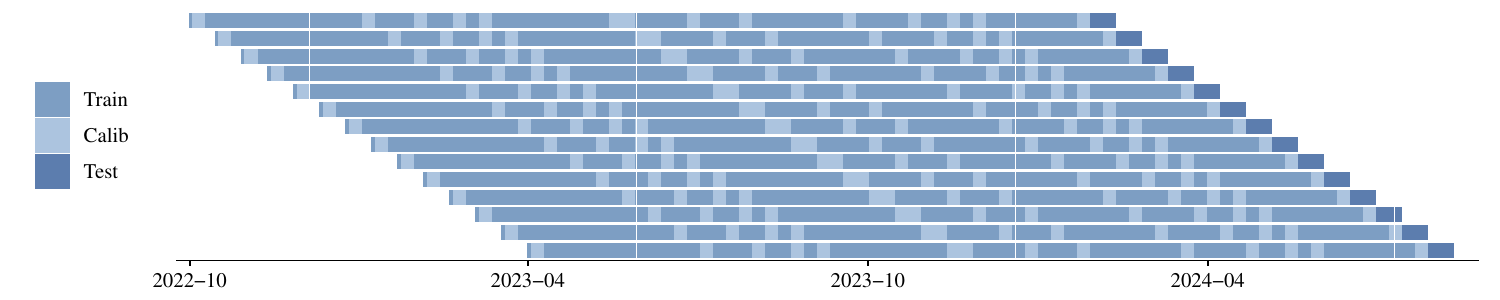}
        \caption{Rolling-window data partitioning. 
        At the first block, the available data is split into a training set $(80\%)$ and a calibration set $(20\%)$. 
        The calibration set of indices is drawn once, combining the week preceding the test period 
        and randomly sampled one-week historical blocks. 
        The same partition is then shifted forward by two weeks across subsequent blocks.}
        \label{fig:windows}
    \end{figure}

\subsection{Choice of the base predictor and reconciliation procedure} 
\label{sec:predictors}
    % --- Base predictor ---
    At all three levels of the hierarchy, 
    we consider a random forest (RF), 
    a generalized additive model (GAM), 
    a (pre-trained) foundation model for time series 
    and tabular data (tabICL), 
    and the combination of these models (Combi)
    that consists of the average of the predictions 
    output by these models. 
    For the regional and station levels, 
    a boosting model (XGB) is added 
    as a base model and included in the average of Combi.
    For the national level, since there is only a single time series, 
    the models exclusively use calendar, weather, and charging cost features, 
    along with lagged energy. 
    Conversely, at the regional and station levels, 
    due to the large number of nodes, a global modeling approach is considered, 
    where a single model is fitted for all nodes within each aggregation level. 
    In this setting, features that enable cross-learning among series, 
    such as infrastructure information and spatial features, are included.
    The GAM is highly efficient for load forecasting, 
    as its additive decomposition and non-linear capabilities 
    are well-suited to model the various effects influencing 
    energy consumption \citep{antoniadis2024statistical}. 
    When applied as a global model, advanced features in GAMs, 
    such as the use of tensor products and combined effects, 
    allow the model to capture the complexity of the data. 
    Furthermore, tree-based models such as RF and XGB are 
    particularly efficient in a global setting, 
    since they naturally group different time series 
    to harness cross-learning.
    TabICL achieves this same global modeling objective 
    through a different approach, relying on cross-attention mechanisms.
    An overview of the related work for the different models 
    is provided in Section~\ref{sec:intro}.

    \begin{remark}
    To keep computation time low and ensure 
    our experiments remain reproducible on a 
    standard laptop for the station level, 
    we restrict the tabICL context to $10^5$ randomly sampled data points 
    and compute a bagging estimate across 5 TabICL models
    of the energy demand by batch of 1024 data points. 
    Note that this approach does not degrade predictive performance; 
    rather, it allows us to effectively manage the memory footprint and computation time.
    \end{remark}

    A base predictor is a multivariate predictor that outputs 
    a vector of forecasts for all nodes in the hierarchy, 
    and can thus be made of any of the above models 
    (e.g., by GAM at the national level and XGB at the regional and station levels).
    For the sake of clarity, the assessment of the 
    probabilistic forecasts are only reported 
    for one base predictor, that consists of Combi
    at each aggregation level,
    and one reconciliation procedure, that is rOLS.
    This choice is motivated by the 
    robustness of this aggregate predictor, 
    and by the modeling rationale of rOLS described 
    in Section~\ref{sec:forecastreconciliation},
    for the reconciliation procedure.
    A posterior analysis of the results 
    in point forecasting confirms this choice,
    as reported in Figure~\ref{fig:rmse}.
    The metric used to assess the point forecast accuracy 
    is the root mean squared error (RMSE),
    computed for each node $i \in [m]$ over the
    active test points. 

    The boxplots report the node-wise RMSEs,
    where the variability of the RMSEs is assessed across 
    the different node of each level in the 10 experimental runs. 
    GAM exhibits the best performance at the national level,
    Combi at the regional level, and tabICL at the station level
    (panel (a) in Figure~\ref{fig:rmse}).
    Overall, the performance between the different models
    and the choice of Combi as the base predictor is 
    essentially driven by its robustness across the hierarchy,
    as it consistently ranks among the top two models at each level.
    Note that the random forest is the worst-performing model at the national level
    by a significant amount. Still, it is included in the Combi predictor,
    which slightly degrades its overall performance.
    
    Given this choice of base predictor, 
    the panels (b) and (c) in Figure~\ref{fig:rmse}
    report the RMSEs of the reconciled forecasts
    respectively at the national, regional, and station levels,
    and from a global perspective, 
    where the RMSE is computed as the square root of the average 
    MSE over all nodes in the hierarchy.
    The different reconciliation procedures 
    exhibit similar performance at the 
    regional and station levels,
    but rOLS is the best-performing procedure
    both at the national level and from a global perspective.
    In particular, from a global perspective, 
    enforcing coherence through OLS degrades 
    the performance of the base predictor,
    while rOLS improves it.

    % --- Posterior validation of the choice ---
\begin{figure}[t]
    \centering
    \subcaptionbox{}{%
        \includegraphics[width=0.15\linewidth,
            trim=11.5 0 4.5 0,clip]{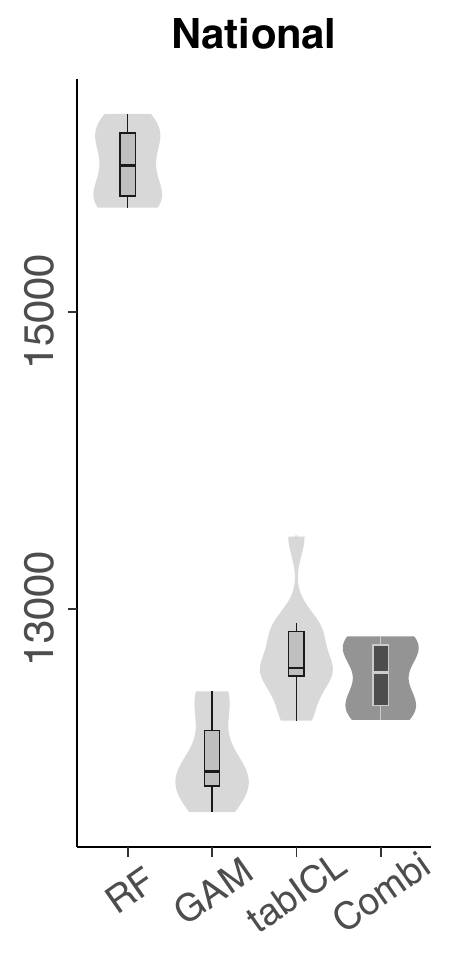}
        \includegraphics[width=0.15\linewidth,
            trim=11.5 0 4.5 0,clip]{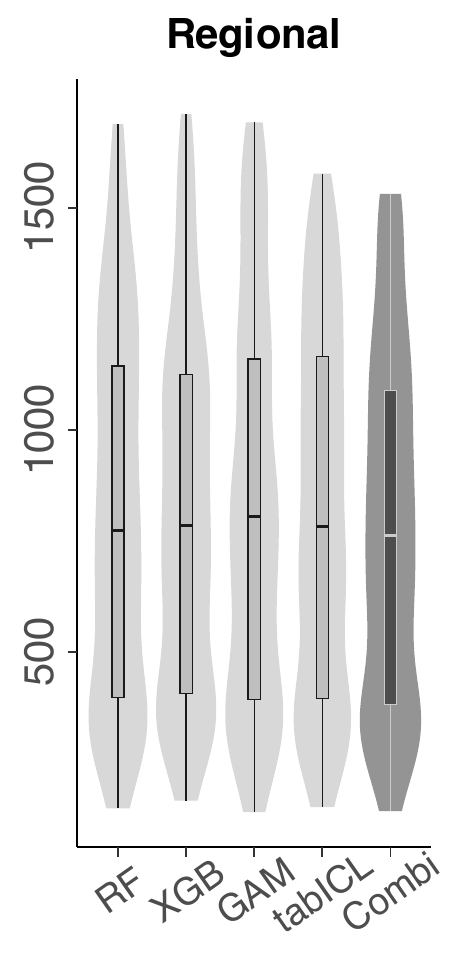}
        \includegraphics[width=0.15\linewidth,
            trim=11.5 0 4.5 0,clip]{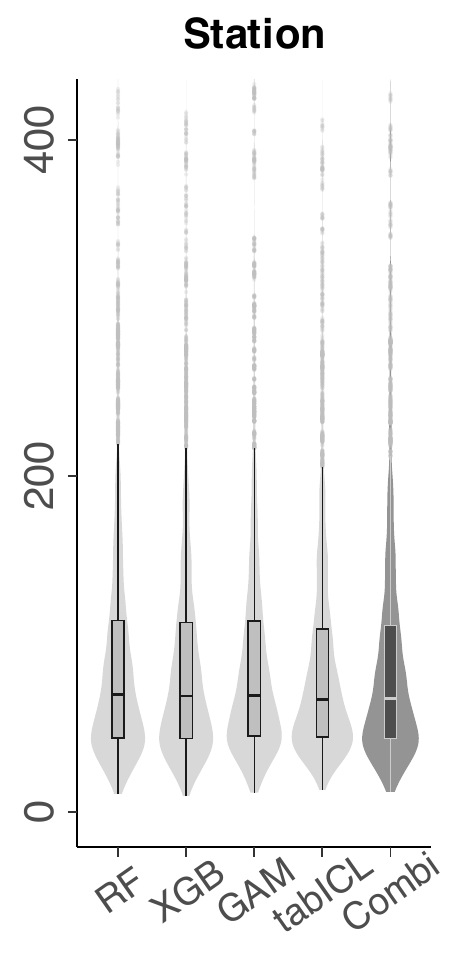}
    }
    \hfill
    \subcaptionbox{}{%
        \includegraphics[width=0.11\linewidth,
            trim=10 0 5 0,clip]{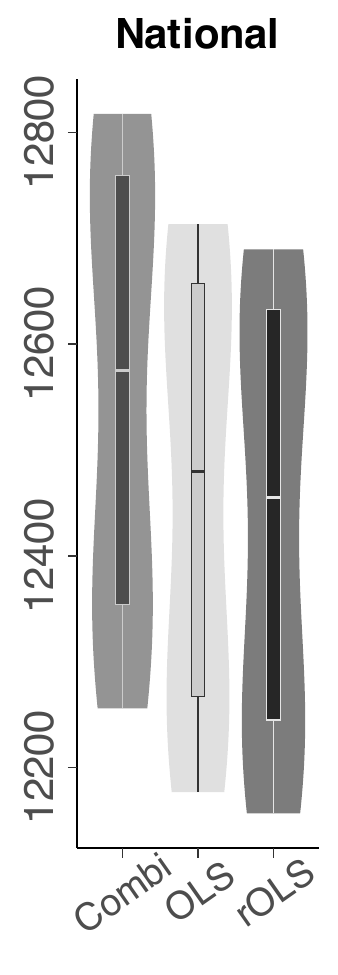}
        \includegraphics[width=0.11\linewidth,
            trim=10 0 5 0,clip]{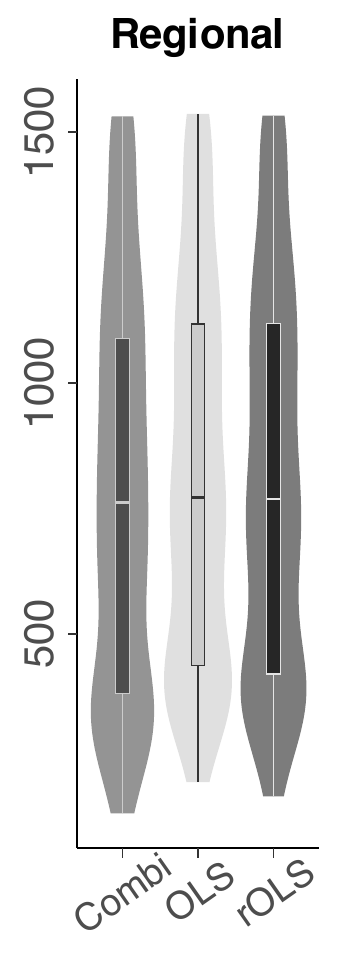}
        \includegraphics[width=0.11\linewidth,
            trim=10 0 5 0,clip]{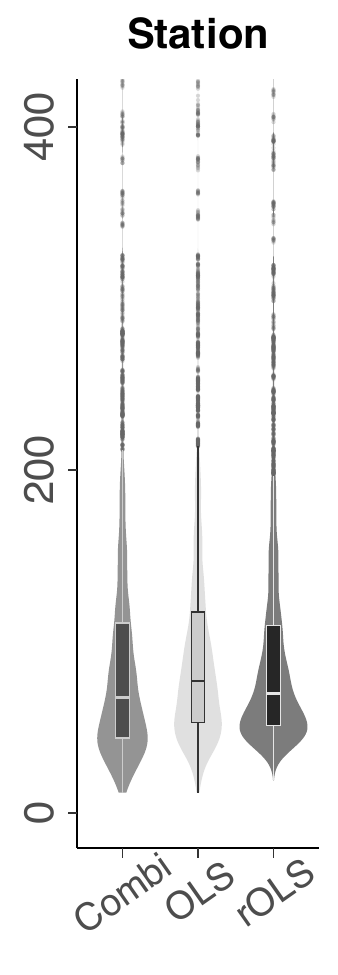}
    }
    \hfill
    \subcaptionbox{}{%
        \includegraphics[width=0.11\linewidth,
            trim=10 0 5 0,clip]{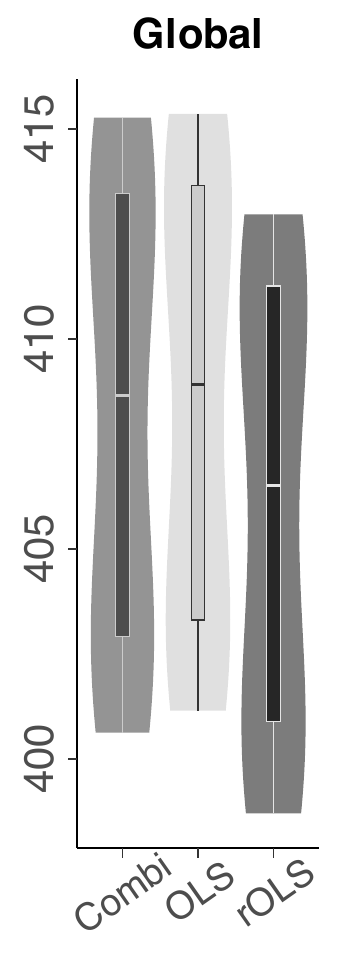}
    }
    \caption{Boxplots of the node-wise RMSEs 
    of the base forecasts (a) 
    and reconciled forecasts (b) 
    across aggregation levels. 
    The Global entry (c) corresponds 
    to the square root of the average MSE over all nodes in the hierarchy.}
    \label{fig:rmse}
\end{figure}

    % --- It benefits in terms of size of the prediction sets ---
    For the sake of completeness, the different reconciliation procedures 
    are also compared in terms of probabilistic forecasts.
    The analysis, deferred to the Supplementary Material~\ref{app:additional_results},
    confirms the choice of rOLS as the reconciliation procedure,
    and shows that it yields smaller prediction intervals than 
    both Direct and OLS, while maintaining the desired coverage level.

    % \smallskip

    % The subsequent analysis is conducted
    % with Combi at every node as a base predictor that is reconciled with rOLS.

    \subsection{Validity rates across all nodes within the hierarchy}
    \label{sec:validity}
    This section reports the evaluation of the procedures 
    from a global perspective through validity rates, 
    defined as proportions of nodes at a given
    level of the hierarchy for which the empirical coverage 
    exceed a specified coverage level.
    The validity rates are computed either conditional on the node being active
    or conditional on the node being active and the neighboring nodes
    being in a certain activity regime, which will be defined in due course.

    \subsubsection{Validity conditional on the node being active}
    For each node $i \in [m]$, empirical coverage estimates 
    are built conditional on the node being active.
    The estimates computed on less than $20$ test points
    are excluded from the validity rate computation.
    Validity rates are thus computed, for each run of the experiments,
    over~$1$ node at the national level, $32$ nodes at the 
    regional level, and $1159$ of the $1398$ nodes at the station level.
    Given a nominal coverage level $1-\alpha = 90\%$,
    the primary objective (depicted in Section~\ref{sec:objcov})
    is to have valid nodes at this coverage level.
    Hence, higher values for the validity rates generally indicate better calibration. 
    However, a validity rate close to $100\%$ is not necessarily desirable, 
    as it may simply reflect systematic overcoverage 
    due to excessively conservative prediction intervals. 
    As a consequence, one expects a well-performing 
    method to achieve a validity rate at the nominal level 
    above $50\%$ without necessarily approaching $100\%$.
    A less a ambitious validity criterion is also considered,
    namely the validity rates at the relaxed level $1-2\alpha = 80\%$.
    The choice to include a relaxed level is motivated by the fact that 
    methods that cover $89\%$ of the test points are not valid at 
    the nominal coverage level $1-\alpha=90\%$, but are more reliable
    than methods that cover only $70\%$ of the test points.
    The validity rates at the relaxed level thus distinguish methods 
    that are close to the nominal level and those that are not.
    In that case, a validity rate close to $100\%$ is desirable.

        \begin{table}[t]
        \centering
        \caption{Validity rates at two coverage levels:
        the nominal level $1-\alpha=90\%$ and a relaxed level $1-2\alpha=80\%$.
        \newline Results are grouped by aggregation level/conformal procedure, 
        and reported as $\operatorname{mean} \pm 1.96 \times \operatorname{standard \, error}$.
        % The best-performing procedure is shown in bold.
        }

        \begin{tabular}{llcccc}
        \toprule
        Node type & Level & CP & CP-MNR & CP-MNR Nested$^{\star}$ & Adaptive CP-MNR \\
        \midrule

        \multirow{2}{*}{National}
        & $1-\alpha$
        & 40.0\% $\pm$ 16.3\%
        & 40.0\% $\pm$ 16.3\%
        & \textbf{90.0\% $\pm$ 10.0\%}
        & 0.0\% $\pm$ 0.0\% \\

        & $1-2\alpha$
        & \textbf{100.0\% $\pm$ 0.0\%}
        & \textbf{100.0\% $\pm$ 0.0\%}
        & \textbf{100.0\% $\pm$ 0.0\%}
        & \textbf{100.0\% $\pm$ 0.0\%} \\

        \addlinespace

        \multirow{2}{*}{Regional}
        & $1-\alpha$
        & 51.9\% $\pm$ 1.2\%
        & 51.9\% $\pm$ 1.2\%
        & \textbf{71.2\% $\pm$ 1.2\%}
        & 47.8\% $\pm$ 0.9\% \\

        & $1-2\alpha$
        & 98.8\% $\pm$ 0.5\%
        & 98.8\% $\pm$ 0.5\%
        & 99.1\% $\pm$ 0.5\%
        & \textbf{100.0\% $\pm$ 0.0\%} \\

        \addlinespace

        \multirow{2}{*}{Station}
        & $1-\alpha$
        & 38.1\% $\pm$ 0.4\%
        & 56.2\% $\pm$ 0.4\%
        & \textbf{73.6\% $\pm$ 0.6\%}
        & 57.6\% $\pm$ 0.6\% \\

        & $1-2\alpha$
        & 67.7\% $\pm$ 0.3\%
        & 93.6\% $\pm$ 0.3\%
        & 90.8\% $\pm$ 0.4\%
        & \textbf{98.0\% $\pm$ 0.1\%} \\

        % \addlinespace

        % \multirow{2}{*}{Mean rank}
        % & $1-\alpha$
        % & 2.67
        % & 2.33
        % & \textbf{1.00}
        % & 3.33 \\

        % & $1-2\alpha$
        % & 3.33
        % & 2.00
        % & 2.00
        % & \textbf{1.00} \\

        \bottomrule
        \end{tabular}

        \label{tab:validity}
    \end{table}

    Table~\ref{tab:validity} reports the validity rates for each
    conformal procedure and at two coverage levels: 
    the nominal level $1-\alpha=90\%$ and a relaxed level $1-2\alpha=80\%$.
    At the nominal level, CP-MNR-Nested$^{\star}$ is the best-performing 
    procedure at each level of the hierarchy by a significant amount, 
    with an average validity rate of $90\%$ at the Scotland node, 
    of $71.2\%$ at the regional level nodes, 
    and of $73.6\%$ at the station level nodes. 
    At the relaxed level, all four procedures are always valid 
    at the national node, and almost always valid at the regional level, 
    with the lowest average validity rate being $98.8\%$ for CP and CP-MNR.
    At the station level, however, Adaptive CP-MNR outperforms the other procedures, 
    with a validity rate of $98\%$ at the relaxed level,
    compared to $90.8\%$ for CP-MNR-Nested$^{\star}$.
    Hence, fewer nodes suffer from a substantial lack of coverage
    with Adaptive CP-MNR,
    which is a desirable robustness property 
    from a practical standpoint.
    A mild undercoverage, say $89\%$ instead of the 
    nominal $90\%$, has limited operational consequences. 
    A severe one, say $70\%$, does not:
    it repeatedly under-estimates the required charging capacity, 
    triggers costly last-minute network reinforcements,
    and erodes the operator's trust in the forecasts. 
    Since the identity of the poorly covered station 
    is not known in advance, avoiding such worst-case 
    failures anywhere in the network matters more 
    in practice than achieving the nominal level everywhere.
    The CP and CP-MNR procedures coincide at the
    national and regional levels, as indicated in Remark~\ref{rk:coincides},
    so that they exhibit the same validity rates at these levels.
    Nevertheless, they differ at the station level, 
    where CP-MNR exhibits a substantially higher validity rate than CP 
    (with $56.2\%$ vs $38.1\%$ at the nominal level,
    and $93.6\%$ vs $67.7\%$ at the relaxed level).
    This confirms the benefit of accounting 
    for the activity of the node under consideration.

    \subsubsection{Validity conditional on the activity of the neighboring nodes}
    % --- Conditional validity of the CP-MNR-Nested$^{\star}$ procedure ---
    In addition to the validity criterion discussed above, 
    that is (only) conditional on the node being active,
    we are interested in validity conditional on the 
    overall infrastructure activity as motivated in Section~\ref{sec:objcov}.
    The station level nodes exhibit
    the importance of conditional validity for
    operational decision making since the underlying uncertainty
    depends on the number of active stations in their neighborhood
    that may affect the accuracy of the forecaster.
    For each active station, the test period is decomposed into two regimes,
    depending on the number of active stations in their local authority: 
    a low-activity regime, when the number of active stations is below 
    the median number of active stations over the test period, 
    and a high-activity regime otherwise. 
    This partitioning provides a simple way to demonstrate 
    the conditional validity of the prediction intervals 
    with respect to the activity of the charging infrastructure. 
    Figure~\ref{fig:conditional} reports the validity rates at 
    level $1-\alpha = 90\%$ respectively marginally and conditional 
    on the low-activity and high-activity regimes.
    Again, only empirical coverages computed on more than $20$ test points
    are included in the validity rates.

        \begin{figure}[t]
    \centering

    % ===================== GAUCHE =====================
    \begin{minipage}[t]{0.75\textwidth}
    \vspace{0pt}
        \centering

    \includegraphics[
        width=\linewidth,
        trim=50 0 50 0,
        clip
    ]{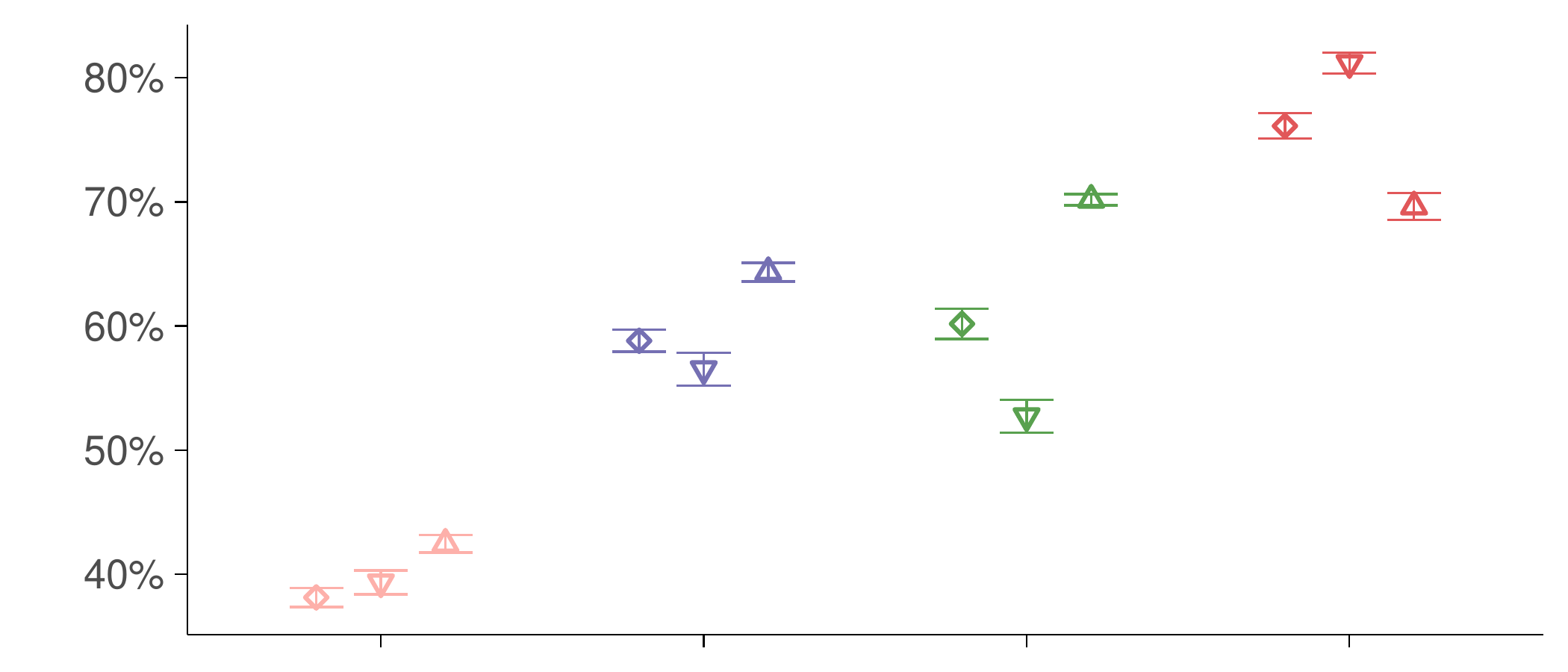}

    % \includegraphics[width=\linewidth,
    %     trim=100 50 200 65,clip]{Figures/conditional_legend_bottom.pdf}
    \end{minipage}
    % ===================== DROITE =====================
    \begin{minipage}[t]{0.15\textwidth}
    \vspace{0pt}
        \centering
    \includegraphics[
        width=\linewidth,
        trim=5 300 10 350,
        clip
    ]{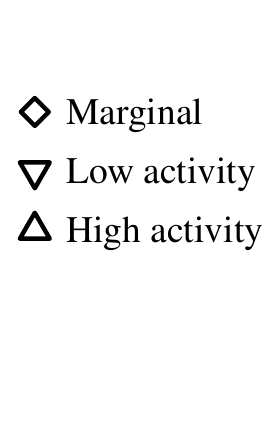}
    \end{minipage}
        \includegraphics[width=0.145\linewidth,
        trim=175 49.5 710 65,clip]{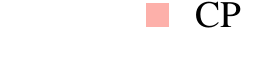}
        \hspace{1.5em}
        \includegraphics[width=0.1275\linewidth,
        trim=300 50 600 65,clip]{Figures/conditional_legend_bottom.pdf}
        \hspace{0.25em}
        \includegraphics[width=0.2385\linewidth,
        trim=410 50 400 65,clip]{Figures/conditional_legend_bottom.pdf}
        \includegraphics[width=0.25\linewidth,
        trim=600 50 200 65,clip]{Figures/conditional_legend_bottom.pdf}
    % \vspace{-0.25cm}

    \caption{Validity rates at the station-level,
    respectively marginally and conditionally 
    on the activity regime (Low or High) 
    within the region of the station under scrutiny.}
    \label{fig:conditional}    
\end{figure}

    The results indicate that the proportion of valid nodes is 
    systematically lower in the low-activity regime 
    than in the high-activity regime
    for CP, CP-MNR, and Adaptive CP-MNR. 
    This suggests that the uncertainty around the forecast is
    higher when the number of neighboring node is low and that
    these three conformal procedures fail to account for 
    the infrastructure activity. On the contrary, CP-MNR-Nested$^{\star}$
    demonstrates a similar validity rate than Adaptive CP-MNR 
    in the high-activity regime (both with $70\%$) 
    but a significantly higher validity rate in the low-activity regime ($81\%$ against $53\%$).
    It emphasizes that the construction of CP-MNR-Nested$^\star$
    successfully account for the whole infrastructure activity
    when producing the prediction intervals.
    The comparison between CP-MNR and Adaptive CP-MNR in
    Figure~\ref{fig:conditional} also illustrates the fact that
    the online extension may hinder conditional coverage 
    (as discussed in Section~\ref{sec:adaptive}). 
    Indeed, the gap in validity rates between
    the low-activity and high-activity regimes is widened 
    by the adaptive extension from $57\%$--$64\%$ to $52\%$--$70\%$.

    % --- Conclure sur ce que cette station dit et ce qu'elle ne dit pas ---
    The numerical results obtained so far indicate that 
    CP-MNR-Nested$^\star$ provides the best overall performance 
    in terms of the two validity measures.
    We now examine whether these empirical guarantees translate into 
    prediction intervals that are also meaningful from a practical standpoint.
    To this end, prediction intervals are displayed for a selection 
    of representative nodes from the hierarchy. 
    Besides illustrating the methods, these examples provide 
    a node-level assessment of the coverage-length trade-off 
    and reveal how each method
    adapts to different infrastructure activity patterns.

\subsection{Focus on Scotland, Glasgow, and two stations in Glasgow}

        \begin{figure}[b]
        \centering
            \includegraphics[width=0.396\linewidth,
                trim=8 22 15 5,clip]{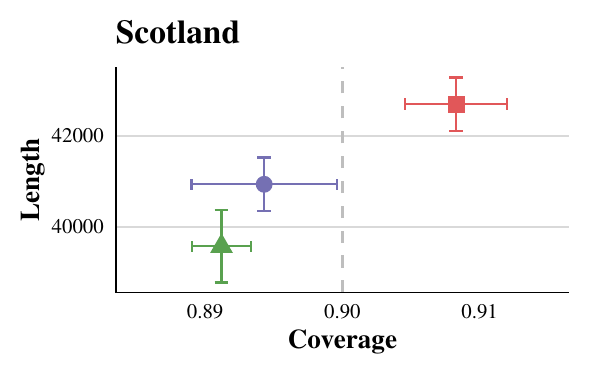}
            % \vspace{-0.65em}
            \includegraphics[width=0.372\linewidth,
                trim=22 22 15 5,clip]{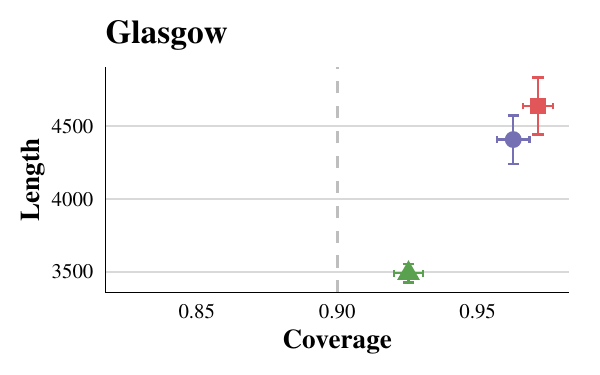}
            \includegraphics[width=0.412\linewidth,
                trim=0 0 15 5,clip]{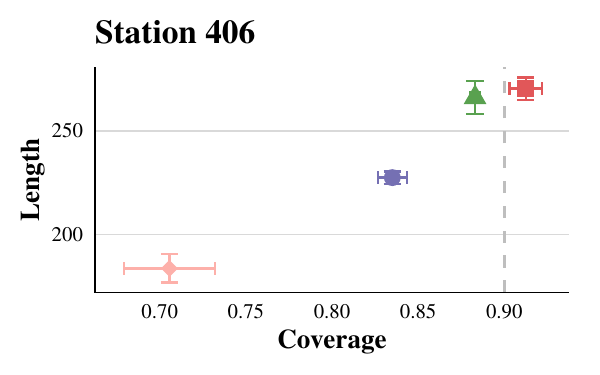}
            \vspace{-0.65em}
            \includegraphics[width=0.380\linewidth,
                trim=22 0 15 5,clip]{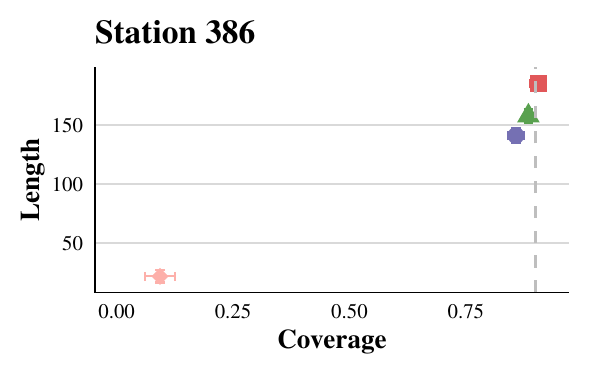}
            \hspace{0.5em}
            \includegraphics[width=0.65\linewidth,
                trim=80 20 80 10,clip]{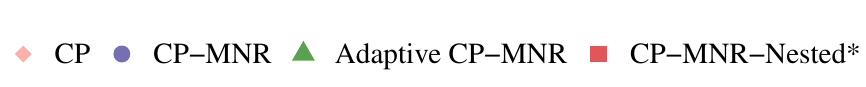}
        \caption{Coverage and average length of the prediction intervals for 
        Scotland (top left), Glasgow (top right),
        Station 406 (bottom left), and Station 386 (bottom right).
        The results are reported as $\operatorname{mean} \pm 1.96 \times \operatorname{standard \, error}$.}
        \label{fig:coverage_length}

    \end{figure}

    % --- Presentation of the results through 1 node per level ---
    The Scotland node, depicted in Figure~\ref{fig:aggregatedseries}, 
    corresponds to the highest aggregation level of the hierarchy. 
    Its demand profile exhibits a pronounced downward trend over the study 
    period, mainly driven by the gradual introduction of charging 
    tariffs across Scotland. 
    As an aggregated node, it remains only weakly affected 
    by changes in the activity of individual charging stations. 
    The Glasgow node, also depicted in Figure~\ref{fig:aggregatedseries}, 
    corresponds to a large urban area with a mature yet continuously 
    evolving charging infrastructure. 
    At the station level, Stations 406 and 386 are two representative nodes in Glasgow,
    displayed in the right-hand panel of Figure~\ref{fig:glasgow_evolvingstructure}. 
    Station 406 became operational only a few months before the beginning of the test period, 
    whereas Station 386 reappeared halfway through the test period  
    after a prolonged period of inactivity following 
    a brief operational phase in 2022.
    These two station-level nodes are expected to constitute more challenging 
    forecasting tasks, especially Station 386, 
    since the calibration set may contain only a small number of observations 
    for which the stations are active.

    \subsubsection{Node-wise coverages and lengths}
    Figure~\ref{fig:coverage_length} reports the empirical coverage 
    and average length for the four nodes under scrutiny.
    A procedure that is valid at the nominal level $1-\alpha=90\%$
    is on the right of the dashed line in Figure~\ref{fig:coverage_length},
    and the efficiency of a procedure is granted by a shorter average length 
    of the prediction intervals, i.e., located at the bottom of the plots.
    The prediction intervals produced by the CP-MNR and Adaptive CP-MNR procedures 
    exhibit a satisfying coverage at all nodes, with a minimal empirical 
    coverage of $84\%$ for CP-MNR and $88\%$ for Adaptive CP-MNR,
    both for Station 406. It confirms the findings of Table~\ref{tab:validity} 
    obtained with validity rates at the relaxed coverage level:
    the above procedures are closely aligned with the nominal 
    coverage level of $90\%$, even if not necessarily valid.
    The only procedure that is valid 
    for all four nodes is CP-MNR-Nested$^{\star}$,
    with a coverage of $91\%$ for Scotland, of $97\%$ for Glasgow,
    and of $91\%$ both for Station 406 and Station 386. 
    Again, it corroborates the fact that CP-MNR-Nested$^{\star}$ 
    provides the best validity guarantees at all levels of the hierarchy.
    The CP procedure, on the other hand, suffers from a significant 
    lack of coverage at the station level, with an empirical coverage of $70\%$ 
    for Station 406 and of (only) $9\%$ for Station 386. 
    This highlights the importance of accounting for the activity 
    of the node under consideration when constructing the prediction intervals,
    as done by the CP-MNR, Adaptive CP-MNR, and CP-MNR-Nested$^{\star}$.
    In terms of efficiency, the Adaptive CP-MNR procedure yields significantly 
    smaller prediction intervals in the case of Glasgow 
    (a relative improvement of at least $20\%$), 
    while being valid, suggesting that Adaptive CP-MNR is 
    the best compromise between validity and efficiency.
    By contrast, CP-MNR-Nested$^{\star}$ produces the largest prediction intervals
    on average, for all four nodes, which emphasizes the trade-off between validity and efficiency.

    % The results reported in Figure~\ref{fig:coverage_length}
    % are consistent with the validity rates reported in Table~\ref{tab:validity},
    % and provide a more detailed view of the coverage-length trade-off
    % at the node level, confirming that CP-MNR-Nested$^{\star}$ 
    % is the only procedure that is valid at all nodes, 
    % but at the cost of larger prediction intervals.

    \subsubsection{Visual assessment of the prediction intervals}    
        \begin{figure}[t]
        \centering
        \begin{subfigure}[b]{0.37\textwidth}
            \centering
            \includegraphics[
            width=\linewidth,
            trim=8 35 20 0,
            clip]{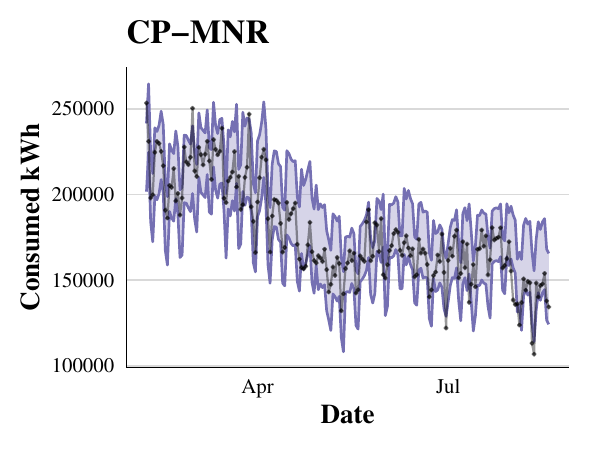}
        \end{subfigure}
        \begin{subfigure}[b]{0.3025\textwidth}
            \centering
            \includegraphics[
            width=\linewidth,
            trim= 55 35 20 0,
            clip]{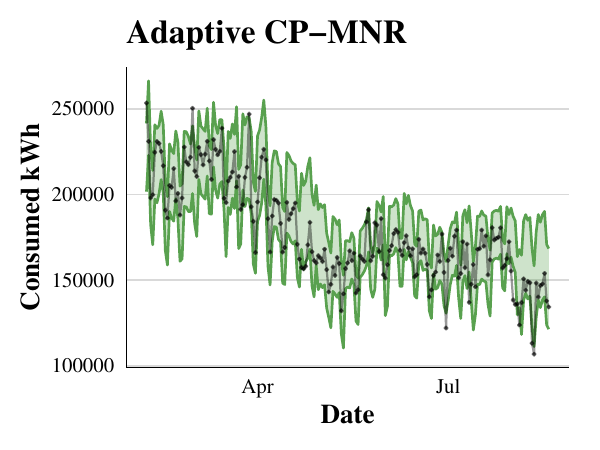}
        \end{subfigure}
        \begin{subfigure}[b]{0.3025\textwidth}
            \centering
            \includegraphics[
            width=\linewidth,
            trim= 55 35 20 0,
            clip]{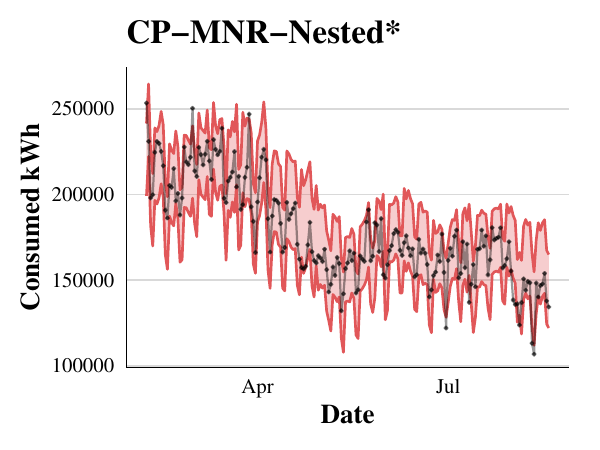}
        \end{subfigure}

        \vspace{0.05cm}

        \begin{subfigure}[b]{0.38\textwidth}
            \centering
            \includegraphics[
            width=\linewidth,
            trim=0 0 20 25,
            clip]{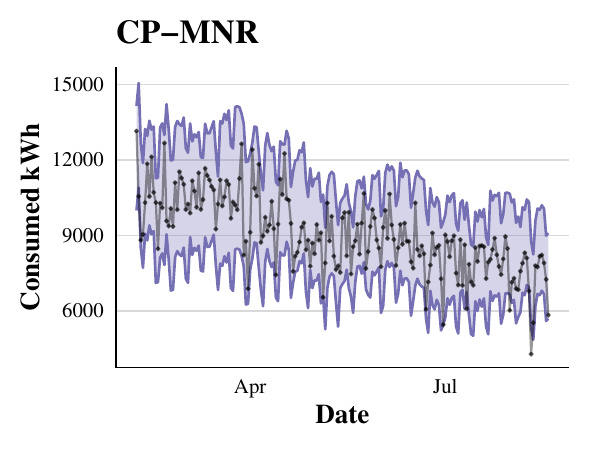}
        \end{subfigure}
        \begin{subfigure}[b]{0.3025\textwidth}
            \centering
            \includegraphics[
            width=\linewidth,
            trim=55 0 20 25,
            clip]{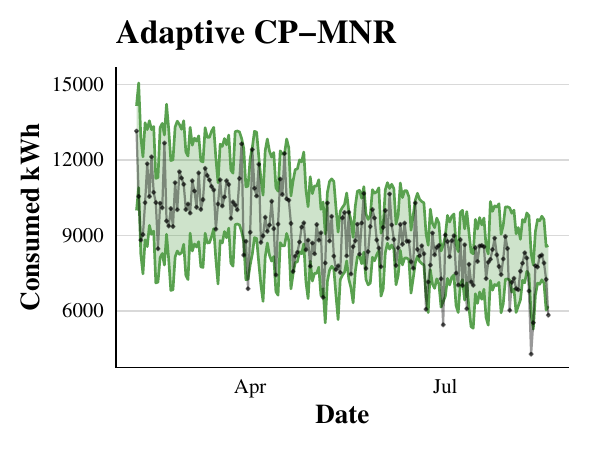}
        \end{subfigure}
        \begin{subfigure}[b]{0.3025\textwidth}
            \centering
            \includegraphics[
            width=\linewidth,
            trim=55 0 20 25,
            clip]{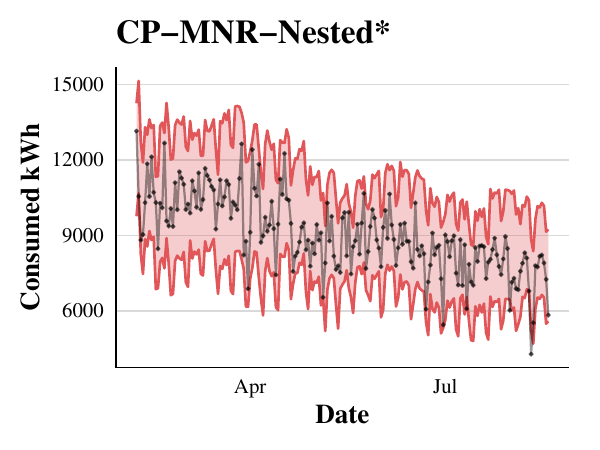}
        \end{subfigure}
        \caption{Prediction intervals for Scotland (top) and Glasgow (bottom).
        }
        \label{fig:intervals}
    \end{figure}

    To complement the numerical results, the performances 
    are visually assessed by reporting in 
    Figures~\ref{fig:intervals}--\ref{fig:intervals_station406}--\ref{fig:intervals_station386} 
    the prediction intervals output by all procedures for one run of the experiment.
    
    At the aggregated levels, namely in Figure~\ref{fig:intervals},
    the visual inspection of the prediction intervals indicates that 
    the procedures are able to handle the temporal evolution of the data. 
    Adaptive CP-MNR is designed to adapt to the temporal evolution of the data
    but the fact that CP-MNR and CP-MNR-Nested$^{\star}$ also exhibit this behavior
    suggests that updating the underlying reconciled predictors
    every two weeks is sufficient to capture this evolution.
    The prediction intervals produced by CP-MNR and 
    CP-MNR-Nested$^{\star}$ seem overly conservative 
    in the case of Glasgow, especially at the end of the test period, 
    whereas Adaptive CP-MNR adjusts to this accurately predicted 
    period, and produces smaller prediction intervals, while remaining valid.
    This illustrates the benefit of the adaptive extension,
    that progressively reduces the size of the prediction intervals if 
    the coverage level is above the nominal level. 

        \begin{figure}[t]
        \centering
        \begin{subfigure}[b]{0.38\textwidth}
            \centering
            \includegraphics[
            width=\linewidth,
            trim=0 35 20 0,
            clip]{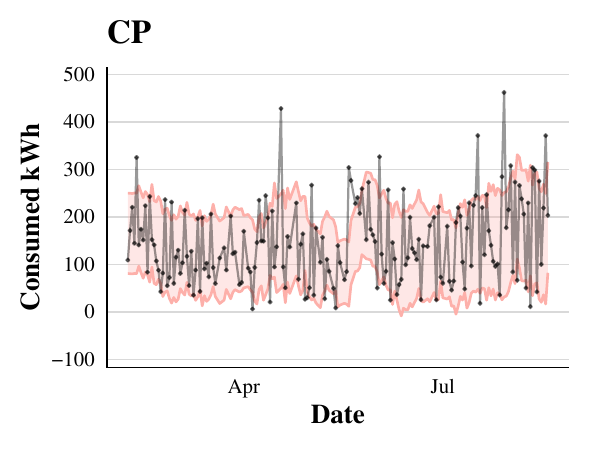}
        \end{subfigure}
        \hspace{2em}
        \begin{subfigure}[b]{0.3125\textwidth}
            \centering
            \includegraphics[
            width=\linewidth,
            trim=48 35 20 0,
            clip]{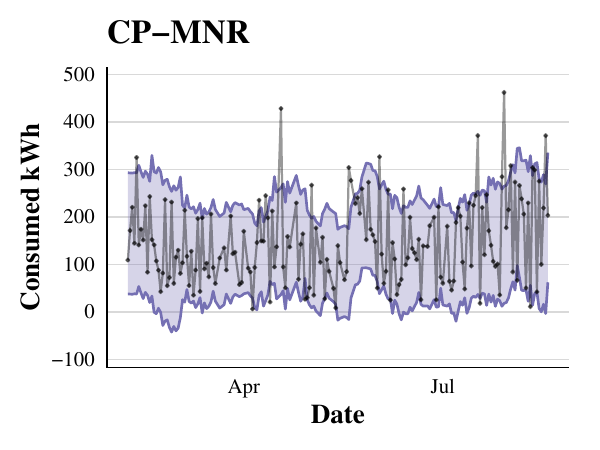}
        \end{subfigure}

        \vspace{0.05cm}

        \begin{subfigure}[b]{0.38\textwidth}
                \centering
                \includegraphics[
                width=\linewidth,
                trim=0 0 20 0,
                clip]{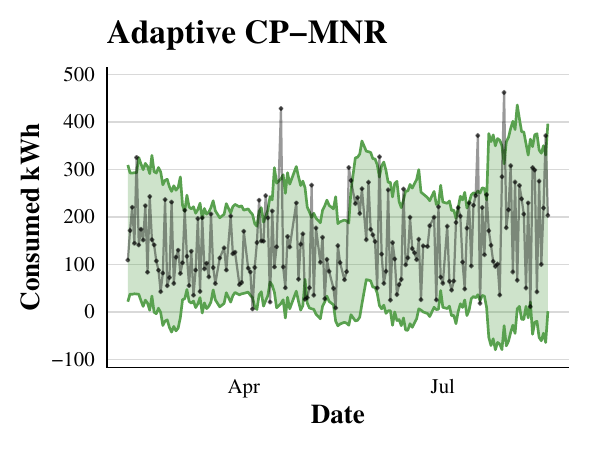}
        \end{subfigure}
            \hspace{2em}
        \begin{subfigure}[b]{0.3125\textwidth}
                \centering
                \includegraphics[
                width=\linewidth,
                trim=48 0 20 0,
                clip]{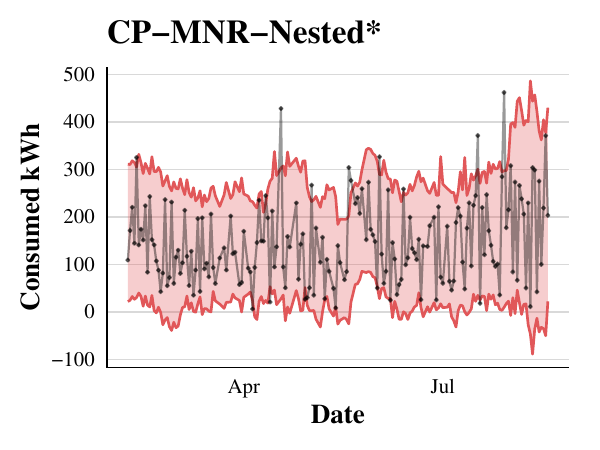}
            \end{subfigure}
        \caption{Prediction intervals for Station 406.
        }
        \label{fig:intervals_station406}
    \end{figure}

        \begin{figure}[b]
        \centering
            \begin{subfigure}[b]{0.375\textwidth}
            \centering
            \includegraphics[
            width=\linewidth,
            trim=8 35 20 0,
            clip]{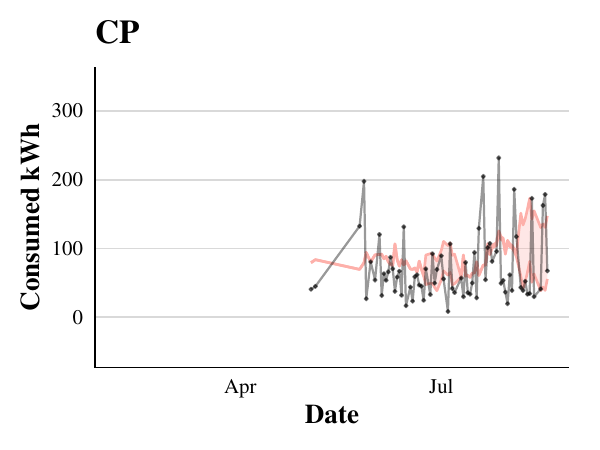}
        \end{subfigure}
                \hspace{2.5em}
        \begin{subfigure}[b]{0.3125\textwidth}
            \centering
            \includegraphics[
            width=\linewidth,
            trim=45 35 20 0,
            clip]{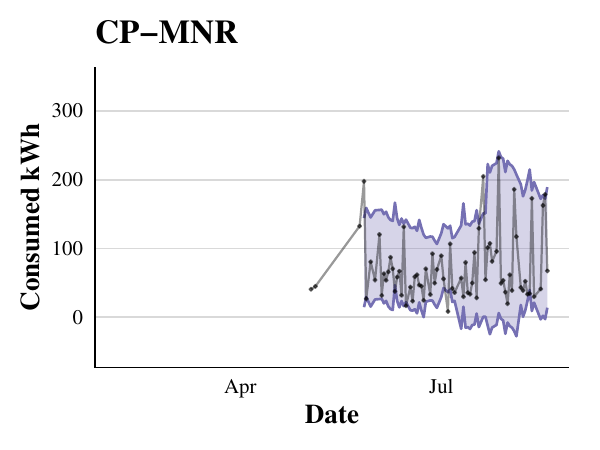}
        \end{subfigure}

        \vspace{0.05cm}

        \begin{subfigure}[b]{0.37\textwidth}
                \centering
                \includegraphics[
                width=\linewidth,
                trim=8 0 20 0,
                clip]{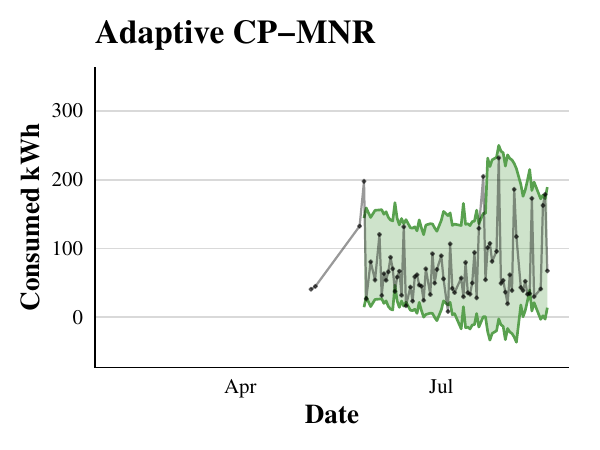}
        \end{subfigure}
                \hspace{2.5em}
        \begin{subfigure}[b]{0.3125\textwidth}
                \centering
                \includegraphics[
                width=\linewidth,
                trim=45 0 20 0,
                clip]{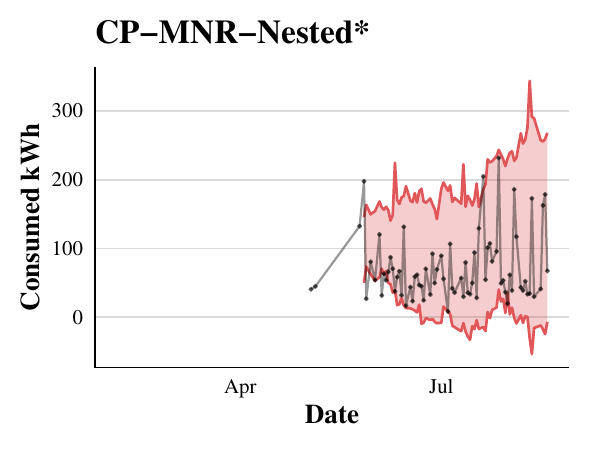}
        \end{subfigure}
    \caption{Prediction intervals for Station 386.}
    \label{fig:intervals_station386}
    \vspace{-1em}
    \end{figure}
    % --- Focus on the station level and the impact of missing node ---
    At the station level, Station 406 (Figure~\ref{fig:intervals_station406}) and
    Station 386 (Figure~\ref{fig:intervals_station386}) both 
    illustrate two typical traits of the prediction intervals.
    The first attribute is that the CP procedure produces prediction intervals 
    that are systematically too narrow, leading to a significant 
    undercoverage of the test points,
    whereas the other three procedures produce prediction intervals 
    that are close to the nominal coverage level in terms of empirical coverage.
    This is particularly visible in Figure~\ref{fig:intervals_station386},
    where the energy demand is zero for most of the calibration points.
    As a consequence, the CP-MNR, Adaptive CP-MNR, and CP-MNR-Nested$^{\star}$ 
    procedures does not produce prediction intervals until 
    some non-zero calibration points are observed 
    (this is why only CP output prediction intervals for the first three 
    non-zero observations of the test period). 
    The same phenomenon is observed for Station 406, 
    but to a smaller extent: the proportion of
    zero energy demands is lower than for Station 386,
    so that they have a smaller impact on the prediction 
    intervals produced by CP.

    The second attribute arises at the end of the test period and
    is related to the activity of neighboring stations.
    On August 10th and 11th, the number of active stations in Glasgow 
    is significantly lower than usual, 
    namely $45$ and $46$ compared to $67$ active 
    stations on average over the test period. At these two dates, 
    CP-MNR-Nested$^{\star}$ produces significantly 
    larger prediction intervals than it does for the rest of the test period,
    which is not the case for the other procedures.    
    This observation suggests that the CP-MNR-Nested$^{\star}$ 
    procedure accounts for the charging infrastructure activity 
    in a more sensitive way than the other procedures, since it 
    detects that the uncertainty is higher 
    when the number of active stations is lower than usual.
    
    The above phenomenon are consistent with the findings of Figure~\ref{fig:conditional},
    where every other procedure exhibits a lower validity rate 
    in the low-activity regime than in the high-activity regime.
    Note that the Adaptive CP-MNR procedure also produces 
    larger prediction intervals at the end of the test period,
    but this behavior is driven by the fact that it fails to
    cover two consecutive at the end of July,
    rather than by some adaptive consideration of the 
    activity of the charging infrastructure.

\section{Discussion and future work}
We successfully combined forecast reconciliation
and conformal prediction in the time-varying hierarchical structure
of the EV charging network from the ChargePlace Scotland data.
The resulting forecasts and prediction intervals are valuable
assets (e.g., as entries of an optimization algorithm) 
that accurately reflect the uncertainty of the forecasts
for the current state of the charging infrastructure.
The numerical experiment demonstrates that CP-MNR-Nested$^{\star}$ 
produces valid prediction intervals for different
configurations of the charging infrastructure,
whereas CP-MNR and Adaptive CP-MNR, yield prediction intervals
that are valid conditionally on the activity of the node under consideration,
but that are less sensitive to the activity of the whole charging infrastructure.

% --- Discussion of the other source of adaptivity ---
A natural direction for future work is to increase
the adaptivity of the methodology through 
online extensions of the learning and reconciliation procedures.
In this direction, one can build on the methodology of
\citet{bregere2022online}, that combines aggregation of experts with OLS reconciliation,
and of \citet{di2021online,hihat2024multivariate} that
propose online multivariate reconciliation procedures
based respectively on a least-squares and 
a regularized least-squares online estimation.

%\section*{APPENDICES MUST BE MOVED TO A SUPPLEMENT FILE}
%%%%%%%%%%%%%%%%%%%%%%%%%%%%%%%%%%%%%%%%%%%%%%
%% Support information, if any,             %%
%% should be provided in the                %%
%% Acknowledgements section.                %%
%%%%%%%%%%%%%%%%%%%%%%%%%%%%%%%%%%%%%%%%%%%%%%
\begin{acks}[Acknowledgments]
The authors would like to thank Margaux Zaffran for the fruitful discussions on her work on conformal prediction for missing data.
\end{acks}

%%%%%%%%%%%%%%%%%%%%%%%%%%%%%%%%%%%%%%%%%%%%%%
%% Funding information, if any,             %%
%% should be provided in the                %%
%% funding section.                         %%
%%%%%%%%%%%%%%%%%%%%%%%%%%%%%%%%%%%%%%%%%%%%%%
\begin{funding}
This work was conducted as part of the Inria-EDF challenge.
\end{funding}

%%%%%%%%%%%%%%%%%%%%%%%%%%%%%%%%%%%%%%%%%%%%%%
%% Supplementary Material, including data   %%
%% sets and code, should be provided in     %%
%% {supplement} environment with title      %%
%% and short description. It cannot be      %%
%% available exclusively as external link.  %%
%% All Supplementary Material must be       %%
%% available to the reader on Project       %%
%% Euclid with the published article.       %%
%%%%%%%%%%%%%%%%%%%%%%%%%%%%%%%%%%%%%%%%%%%%%%
\bigskip

\begin{supplement}
\stitle{Connection with CP-MDA and Theoretical Results} 
\sdescription{This supplement gives mathematical insights that
motivate the construction of
CP-MNR and CP-MNR-Nested$^{\star}$. 
It includes the conceptual connection between our
conformal prediction with missing node reconciliation (MNR) setting
and the conformal prediction with missing data (MDA) framework 
developed by \citet{zaffran2024predictive} that guided our approach.
}
\end{supplement}
\begin{supplement}
\stitle{Implementation of the Algorithms}
\sdescription{This supplement provides 
further information on the implementation of the proposed algorithms.
It includes (i) the formal description of the learning
and reconciliation algorithms that are used in
Algorithms~\ref{algo:CP-MNR} and~\ref{algo:CP-MNR-Nested*};
(ii) the implementation of the conformal step in Algorithm~\ref{algo:CP-MNR-Nested*};
(iii) the implementation and theoretical analysis 
of the Adaptive CP-MNR procedure introduced in Section~\ref{sec:adaptive}.}
\end{supplement}
\begin{supplement}
\stitle{Additional Numerical Results} 
\sdescription{This supplement focuses on the comparison
of the different forecast reconciliation procedures 
in terms of the prediction interval coverage and length.
It advocates for the use of rOLS, which is the focus of the main text.
}
\end{supplement}

%%%%%%%%%%%%%%%%%%%%%%%%%%%%%%%%%%%%%%%%%%%%%%%%%%%%%%%%%%%%%
%%                  The Bibliography                       %%
%%                                                         %%
%%  imsart-nameyear.bst  will be used to                   %%
%%  create a .BBL file for submission.                     %%
%%                                                         %%
%%  Note that the displayed Bibliography will not          %%
%%  necessarily be rendered by Latex exactly as specified  %%
%%  in the online Instructions for Authors.                %%
%%                                                         %%
%%  MR numbers will be added by VTeX.                      %%
%%                                                         %%
%%  Use \cite{...} to cite references in text.             %%
%%                                                         %%
%%%%%%%%%%%%%%%%%%%%%%%%%%%%%%%%%%%%%%%%%%%%%%%%%%%%%%%%%%%%%

%% if your bibliography is in bibtex format, uncomment commands:
\bibliographystyle{imsart-nameyear} % Style BST file
\bibliography{aoas-article}       % Bibliography file (usually '*.bib')

%%%%%%%%%%%%%%%%%%%%%%%%%%%%%%%%%%%%%%%%%%%%%%
%% Temporary Appendix section is placed     %%
%% here. Its content should be moved into   %%
%% another file.                            %%     
%%%%%%%%%%%%%%%%%%%%%%%%%%%%%%%%%%%%%%%%%%%%%%
\clearpage
\appendix

\section{Connection with CP-MDA and Theoretical Results} 
\label{app:missing}
This Supplementary Material presents
the heuristic theoretical coverage guarantees 
of the CP-MNR and CP-MNR-Nested$^{\star}$ procedures.
The analysis is based on the connection between the CP-MNR setting
and the conformal prediction with missing data framework 
of \citet{zaffran2023conformal,zaffran2024predictive}, which guided the approach.
We first establish this connection.

\subsection{Connection with the conformal prediction with missing data framework}
\label{sec:connection}
% --- Connection to the missing-data framework ---
The incentive for the CP-MNR and 
CP-MNR-Nested$^{\star}$ procedures
is to construct prediction intervals 
that are valid conditionally 
on the state of the infrastructure, 
represented by a binary activity vector. 
This is closely related to the objective of the
missing data augmentation (MDA) framework of
\citet{zaffran2023conformal,zaffran2024predictive},
where validity is conditioned on a binary mask 
indicating which features are missing.
Nevertheless, the binary vector of infrastructure 
activity does not encode missingness
but instead identifies inactive nodes 
in the hierarchy, whose zero-valued 
constraints directly affect 
the forecast of every node of the hierarchy, 
through the reconciliation procedure.

The two settings are thus conceptually different,
but the following observation establishes 
a connection between them:

\smallskip

``Base forecasts are covariates
in the reconciliation step,
and enforcing the infrastructure activity amounts 
to setting some base forecasts to zero,
which corresponds to a specific imputation strategy
in the framework of conformal prediction with missing data.''

\smallskip

% --- Details on the connection with CP-MDA ---
This observation explicit is made explicit
by building on the interpretation 
from~\citet{hollyman2021understanding} of
forecast reconciliation as forecast combination.
This interpretation is illustrated on a simple example, 
omitting algebraic details for brevity.

Consider the toy-example used throughout the main text 
in Sections~\ref{sec:hierarchical_structure} 
and~\ref{sec:forecastreconciliation},
with $n = 3$ stations, $m=4$ nodes, 
infrastructure activity vector 
$\ba = (1,0,1,1)^{\top}$, 
and structural matrix $H$ 
defined in~\eqref{eq:example_refinement}.
Let $\widehat{\by}$ denote a base forecasts vector
and denote by $\widetilde{\by}$ 
the associated reconciled forecasts vector
obtained by applying a projection $P$ 
onto the image of $H$ to the base forecasts vector.
Note that $P$ can be the orthogonal projection 
as in the OLS reconciliation procedure,
but it can also be any other projection onto the image of $H$.

Following \citet{hollyman2021understanding}, 
the reconciled forecasts $\widetilde{\by}$ can be expressed
as a constrained linear combination of direct forecast of the node under
scrutiny, and an indirect one, constructed from the remaining nodes:
\begin{align*}
        \widetilde{y}_4 &= (3 - \lambda_1 - \lambda_2 - \lambda_3) \, \widehat{y}_4 + 
    (\lambda_1 + \lambda_2 + \lambda_3 - 2) \bigl(\widehat{y}_1 + \widehat{y}_2 + \widehat{y}_3\bigr) \,; \\
    \widetilde{y}_1 &= \lambda_1 \, \widehat{y}_1 + 
    (1 - \lambda_1) \bigl(\widehat{y}_4 - \widehat{y}_2 - \widehat{y}_3\bigr) \,; 
    \quad \widetilde{y}_2 = \lambda_2 \, \widehat{y}_2 + 
    (1 - \lambda_2) \bigl(\widehat{y}_4 - \widehat{y}_1 - \widehat{y}_3\bigr) \,; \\
    \widetilde{y}_3 &= \lambda_3 \, \widehat{y}_3 + 
    (1 - \lambda_3) \bigl(\widehat{y}_4 - \widehat{y}_1 - \widehat{y}_2\bigr) \,, 
    \quad \mbox{where} \quad \lambda_1, \lambda_2, \lambda_3 \quad \mbox{depend on the projection } P \,.
\end{align*}
Now, the reconciled forecasts 
$\widetilde{\by}^{\,\prime}$ produced by rOLS
can be expressed in a similar way:
\begin{align*} 
    \widetilde{y}^{\,\prime}_4 &= (2 - \lambda'_1 - \lambda'_3) \, \widehat{y}_4 + 
    (\lambda'_1 + \lambda'_3 - 1) \bigl(\widehat{y}_1 + 0 + \widehat{y}_3\bigr) \,; \\
    \widetilde{y}^{\,\prime}_1 &= \lambda'_1 \, \widehat{y}_1 + 
    (1 - \lambda'_1) \bigl(\widehat{y}_4 - 0 - \widehat{y}_3\bigr) \,; 
    \quad
    \widetilde{y}^{\,\prime}_2 = 0 \,; \quad
    \widetilde{y}^{\,\prime}_3 = \lambda'_3 \, \widehat{y}_3 + 
    (1 - \lambda'_3) \bigl(\widehat{y}_4 - \widehat{y}_1 - 0 \bigr) \,,
\end{align*}
with $\lambda'_1=\lambda'_3=2/3$ in this case. 

This formulation makes it clear that 
the base forecasts of inactive nodes 
no longer contribute to the forecast combination: 
$\widehat{y}_2$ is systematically replaced by zero. 
This mechanism can be interpreted 
as a particular imputation strategy, 
where the base forecasts of the inactive nodes are deterministically 
imputed by zero before constructing the indirect forecasts.

This perspective reveals a close connection 
with the conformal prediction with missing data 
framework of \citet{zaffran2024predictive}, 
allowing the theoretical analysis to follow 
essentially the same arguments.

\subsection{Theoretical Result for CP-MNR}
\label{app:thCP-MNR}
In this section, we present the theoretical coverage guarantee for CP-MNR,
which is conditional on the activity of the node under scrutiny.

\smallskip

Following the procedure described 
in Sections~\ref{sec:hierarchicalconformal} 
and~\ref{sec:methodology}, 
we consider a calibration set 
of observations indexed by $\Calib$ 
with cardinality $\#\Calib$, and a test point 
$(\bx^{(t)}, \ba^{(t)})$ with unknown target value $\by^{(t)}$.
Recall that the non-conformity score function
introduced in~\eqref{eq:conformity_score_i} and 
considered throughout the theoretical analysis 
is defined as
\begin{equation}
    \label{eq:scoreapp}
    \forall i \in [m]\,, \qquad s_i : \Bigl((\bx, \ba), \by; \widetilde{A} \Bigr) 
    \longmapsto \Bigl|\Pi_i\Bigl(\by - \widetilde{A}(\bx, \ba)\Bigr)\Bigr|\,.
\end{equation}
As in classic conformal prediction, 
the theoretical framework is built on 
an exchangeability assumption.
The latter may not hold in practice, 
but provides a heuristic justification 
for the validity of the proposed procedures.
\begin{assumption}\label{ass:exch}
       The random variables $\bigl\{\bigl(\bx^{(k)}, \ba^{(k)}, \by^{(k)}\bigr)\bigr\}_{k \in \Calib \cup \{t\}}$ are exchangeable, i.e., their joint distribution is invariant under permutations.
\end{assumption}

The following theorem provides a theoretical coverage guarantee for CP-MNR, 
that is conditional on the activity of the node under scrutiny.

\begin{theorem}\label{th:weak-conditional}
    Consider the nominal coverage level $1-\alpha \in [0,1]$.
    Under Assumption~\ref{ass:exch}, the CP-MNR procedure of Algorithm~\ref{algo:CP-MNR}, 
    run with any learning and reconciliation algorithms $\cA$, $\cR$, satisfies
    \begin{equation} \label{eq:weak-conditional}
        \forall i \in [m] \,,  \qquad \P\Bigl(y_i^{(t)} \in \, \cC_i\bigl(\bx^{(t)}, \ba^{(t)}\bigr) \Big| \, a^{(t)}_i=1 \Bigr) \geq 1-\alpha \,.
    \end{equation}
\end{theorem}
\begin{proof}
    The result follows from the fact that Algorithm~\ref{algo:CP-MNR} 
    outputs infrastructure activity-conditional 
    Mondrian conformal predictors \citep{vovk2005algorithmic}.
    Precisely, our formulation corresponds to the
    group-conditional framework of \citet{romano2020malice}.
    For the sake of completeness, we report their arguments.

    \smallskip
    Fix $i \in [m]$. 
    Conditional on $a_i^{(t)} = 1$, 
    the grouping of calibration indices 
    in Algorithm~\ref{algo:CP-MNR} is of the form
    $\I_1 = \bigl\{k \in \Calib : a_i^{(k)} = 1 \bigr\}$. 
    Thus, the exchangeability property of Assumption~\ref{ass:exch}
    implies that 
    $\bigl\{\bigl(\bx^{(k)}, \ba^{(k)}, \by^{(k)}\bigr)\bigr\}_{k \in \I_1 \cup \{t\}}$ 
    is exchangeable conditional on $a_i^{(t)} = 1$.  

    \smallskip

    The reconciled predictor $\widetilde{A}$ 
    only depends on data from the training set, 
    so that the non-conformity scores 
    $\bigl\{s_i^{(k)}\bigr\}_{k \in \I_1 \cup \{t\}}$
    is exchangeable conditional on $a_i^{(t)} = 1$.

    \smallskip

    Finally, resorting to \citealp[Lemma~1]{romano2020malice}, which is essentially the 
    same argument as in the classical conformal prediction setting adapted to conditional exchangeability,
    we obtain the infrastructure activity-conditional guarantee in~\eqref{eq:weak-conditional}, as announced.
\end{proof}

\begin{remark}
% --- Intermediate conditioning is meaningful here ---
The CP-MNR procedure of Algorithm~\ref{algo:CP-MNR}
has no direct counterpart in \citet{zaffran2024predictive}.
For instance, in MDA, there is no particular 
reason to assess validity conditional on missingness 
of one specific covariate rather than another.
By contrast, in MNR, the activity of the node under consideration 
is a particularly meaningful conditioning variable, 
as it directly shapes the demand behavior of the node.
\end{remark}

\subsection{Theoretical Results for CP-MNR-Nested$^{\star}$}
\label{app:thCP-MNR-Nested*}
We show two theoretical coverage guarantees 
for the CP-MNR-Nested$^{\star}$ procedure of Algorithm~\ref{algo:CP-MNR-Nested*}.
The first one is a marginal guarantee, obtained
under the mild exchangeability assumption of Assumption~\ref{ass:exch}.
The second one is a conditional guarantee,
obtained under a stronger assumption, as described next.

\smallskip

The first result is a weaker result 
than the one of Theorem~\ref{th:weak-conditional} in two ways: 
it is marginal rather than locally conditional, 
and the coverage level is $1-2\alpha$ rather than $1-\alpha$.

\begin{theorem} \label{th:marginalnested}
    Consider the nominal coverage level $1-\alpha \in [0,1]$.
    Under Assumption~\ref{ass:exch}, 
    the CP-MNR-Nested* procedure of Algorithm~\ref{algo:CP-MNR-Nested*}, 
    run with $\tCalib = \Calib$, any learning 
    and reconciliation algorithms $\cA$, $\cR$, 
    satisfies
\begin{equation} \label{eq:marginalvalidity}
    \forall i \in [m] \,,  \qquad \P\Bigl(y_i^{(t)} \in \, \cC_i\bigl(\bx^{(t)}, \ba^{(t)}\bigr) \Bigr) \geq 1-2\alpha \,.
\end{equation}
\end{theorem}
\begin{proof}
    The proof follows the lines and arguments 
    from the proof of \citealp[Theorem 4.2]{zaffran2024predictive},
    that is itself inspired from the jackknife+ conformal predictors 
    proof from~\citealp[Theorem~1]{barber2021predictive}.
    For the sake of completeness, we report their proofs to our setting and notations.

    \smallskip

    Fix $i \in [m]$, and denote by $\C_i\bigl(\bx^{(t)}, \ba^{(t)}\bigr)$ 
    the corresponding conformal prediction set output
    by Algorithm~\ref{algo:CP-MNR-Nested*}, when run with $\tCalib = \Calib$.
    By construction, we have that

    \begin{equation} \label{eq:loo}
        \Bigl\{
        y_i^{(t)} \notin \C_i\bigl(\bx^{(t)}, \ba^{(t)}\bigr) \Bigr\} =
        \Bigl\{
        (1 - \alpha)\bigl(\#\Calib +1 \bigr) \leq \sum_{k \in \Calib} \indic_{S_i^{(k),t} < S_i^{(t),k}} \Bigr\} \,,
    \end{equation}
    where $S_i^{(p),q}$ denotes the non-conformity score of a calibration point $p$ 
    when the infrastructure activity pattern of the point $q$ is applied on top of its own 
    infrastructure activity pattern, for all $p,q \in \Calib \cup \{t\}$.

    \smallskip

    Following \citet{barber2021predictive}, 
    we define a comparison matrix $M_i$, 
    indexed by $\Calib \cup \{t\}$, with entries 
    $M_i(p,q) = \indic_{S_i^{(p),q} > S_i^{(q),p}}$ for $p,q \in \Calib \cup \{t\}$,
    and consider the set
    \begin{equation} \label{eq:WMi}
    W(M_i) = \biggl\{p \in \Calib \cup  \{t\} : \!\!\! \sum_{q \in \Calib \, \cup \{t\}} \!\!\! M_i(p,q) \geq (1-\alpha)\bigl(\#\Calib+1\bigr) \biggr\} \,.
    \end{equation}
    Thereof, we can rewrite~\eqref{eq:loo} as
    $\bigl\{y_i^{(t)} \notin \C_i\bigl(\bx^{(t)}, \ba^{(t)}\bigr) \bigr\} 
    = \bigl\{t \in W(M_i) \bigr\}$, so that it suffices 
    to upper bound the probability of the event $\bigl\{t \in W(M_i)\bigr\}$.

    \smallskip

    Denote by $\M_i$ the function that maps 
    $\bigl\{\bigl(\bx^{(k)}, \ba^{(k)}, \by^{(k)}\bigr)\bigr\}_{k \in \Calib \cup \{t\}}$
    to the comparison matrix $M_i$ defined above, and fix
    $p, \,q \in \Calib \cup \{t\}$.
    For any permutation $\sigma$ on $\Calib \cup \{t\}$,
    the symmetrical construction of the reconciled forecast, 
    and thus of the scores, ensures that
    \begin{align}\label{eq:symmetry}
    M_i\bigl(\sigma(p),\sigma(q)\bigr) &= \M_i\Bigl(\bigl(\bx^{(k)}, \ba^{(k)}, \by^{(k)}\bigr)_{k \in \Calib \cup \{t\}}\Bigr)\bigl(\sigma(p),\sigma(q)\bigr) \nonumber \\
    &= \M_i\Bigl(\bigl(\bx^{(\sigma(k))}, \ba^{(\sigma(k))}, \by^{(\sigma(k))}\bigr)_{k \in \Calib \cup \{t\}}\Bigr)(p,q) \,.
    \end{align}
    Therefore, we combine~\eqref{eq:symmetry} with Assumption~\ref{ass:exch}, 
    and get that $M_i^{\sigma}\overset{d}{=} M_i$, where
    $M_i^{\sigma}$ is the permuted matrix $\bigl(M_i\bigl(\sigma(p),\sigma(q)\bigr)\bigr)_{p,q \in \Calib \cup \{t\}}$.
    As a consequence, $\P\bigl(p \in W(M_i)\bigr)$ 
    does not depend on the specific value of $p \in \Calib \cup \{t\}$, which implies that
    \begin{align*} 
    \P\bigl(t \in W(M_i)\bigr) = 
    \frac{1}{\#\Calib + 1} \sum_{p \in \Calib \, \cup \{t\}} \!\!\! \P\bigl(p \in W(M_i)\bigr) 
    = \frac{1}{\#\Calib + 1} \E\Bigl[ \bigl| W(M_i) \bigr| \Bigr] \,.
    \end{align*}
    The proof concludes by using 
    a deterministic upper bound 
    on the cardinality of $W(M_i)$,
    obtained by Landau's theorem 
    for tournaments \citep{landau1953dominance}, 
    namely,
    \[
    \#W(M_i) \leq 2\alpha \bigl(\#\Calib + 1\bigr)\,.
    \]\end{proof}

We now turn to the stronger conditional coverage 
guarantee of CP-MNR-Nested$^{\star}$, obtained
under the following stronger assumptions.
\begin{assumption}\label{ass:iid}
       The random variables $\bigl\{\bigl(\bx^{(k)}, \ba^{(k)}, \by^{(k)}\bigr)\bigr\}_{k \in \Calib \cup \{t\}}$ are drawn i.i.d.
\end{assumption}

\begin{assumption}\label{ass:mask}
    We assume that the infrastructure activity pattern is independent of the observations, i.e., 
    for all $k \in \Calib \cup \{t\}$, 
    we have that $\ba^{(k)} \perp \bigl(\bx^{(k)},\by^{(k)}\bigr)$.
\end{assumption}

Assumption~\ref{ass:iid} is a slightly 
stronger version of Assumption~\ref{ass:exch},
but the main difference is that Assumption~\ref{ass:mask}
is a very restrictive assumption, 
which does not hold in our setting.
Indeed, the infrastructure activity vector is deterministically 
induced by the EV charging energy demand vector,
which breaks the independence assumption of Assumption~\ref{ass:mask}.
Nevertheless, the encouraging simulation evidence from \citealp[Section 5.2 and 5.3]{zaffran2024predictive},
together with our empirical findings in Section~\ref{sec:experiments},
indicate that the above Theorem~\ref{th:strong-conditional} and the associated
conditional guarantee of CP-MNR-Nested$^{\star}$ 
may still hold in practice, even when Assumption~\ref{ass:mask} is violated.

\begin{theorem}\label{th:strong-conditional}
    Consider the nominal coverage level $1-\alpha \in [0,1]$.
    Under Assumption~\ref{ass:iid} and~\ref{ass:mask}, 
    the CP-MNR-Nested$^{\star}$ procedure of Algorithm~\ref{algo:CP-MNR-Nested*}, 
    run with $\tCalib = \Calib$, any learning 
    and reconciliation algorithms $\cA$, $\cR$, 
    satisfies
    \begin{equation} \label{eq:strong-conditional}
        \forall \ba \in \{0,1\}^m \!, \, i \in [m] \,, 
        \qquad \P\Bigl(y_i^{(t)} \in \, \C_i\bigl(\bx^{(t)}, \ba^{(t)}\bigr) \Big| \, \ba^{(t)} = \ba \Bigr) \geq 1-2\alpha \,.
    \end{equation}
\end{theorem}
\begin{proof}
    As for the proof of Theorem~\ref{th:marginalnested}, 
    the proof essentially follows the lines and arguments of \citet{zaffran2024predictive}, 
    namely of Theorem~4.3. 
    We only report our own formulation of their proof with our setting and notations.

    \medskip

    \emph{Step 1: conditional exchangeability.}
    By Assumptions~\ref{ass:iid}~and~\ref{ass:mask}, the test infrastructure activity pattern 
    is independent from $\bigl(\bx^{(k)}, \by^{(k)}\bigr)$ for all
    $k \in \Calib \cup \{t\}$ so that the set 
    of observations with substituted infrastructure activity
    $\bigl\{\bigl(\bx^{(k)}, \ba^{(t)}, \by^{(k)} \bigr)_{k \in \Calib}, \bigl(\bx^{(t)}, \ba^{(t)}, \by^{(t)} \bigr)\bigr\}$
    is exchangeable conditional on $\ba^{(t)}$.
        
    \medskip
    
    \emph{Step 2: from Algorithm~\ref{algo:CP-MNR-Nested*} to Algorithm~\ref{algo:randomCP-MNR-Nested*}.}
    We recast CP-MNR-Nested$^{\star}$ as a procedure that performs the conformal step 
    on the observations with substituted infrastructure activity patterns, for which we have conditional exchangeability by Step~1.
    To do so, we introduce Algorithm~\ref{algo:randomCP-MNR-Nested*}, and show that the CP-MNR-Nested* procedure
    corresponds to a particular instance of Algorithm~\ref{algo:randomCP-MNR-Nested*} run with calibration inputs:
    \[
        \Bigl\{ {\bigl(\bu^{(k)}, \bv^{(k)}, \bz^{(k)} \bigr)}\Bigr\}_{k \in \Calib}
        = \Bigl\{ \bigl(\bx^{(k)}, {\ba^{(t)}}, \by^{(k)} \bigr)\Bigr\}_{k \in \Calib} \,.
    \] 

    First, we use that 
    $\widetilde{A}(\cdot,\cdot)=\widehat{R}\bigl(\widehat{A}(\cdot, \cdot), \cdot\bigr)$
    so that, for any observation $(\bu, \bv, \bz)$, 
    its non-conformity score at node $i \in [m]$ 
    with augmented infrastructure activity 
    $\widetilde{\ba} = \min \bigl( \bv, \ba^{(t)} \bigr)$
    of Algorithm~\ref{algo:CP-MNR-Nested*} rewrites as
    \[
    s_i\Bigl((\bu, \widetilde{\ba}), \bz; \widehat{R}\bigl(\widehat{A}(\cdot, \cdot), \cdot\bigr)\Bigr) 
    = s_i\Bigl((\bu, \ba^{(t)}), \bz; \widehat{R}^{\star}\bigl(\widehat{A}^{\star}(\cdot, \cdot; \bv), \cdot; \bv\bigr)\Bigr) \,,
    \]
    where $\widehat{A}^\star(\cdot, \cdot; \bv) : (\bx, \ba) \mapsto \widehat{A}\bigl(\bx, \min(\ba, \bv)\bigr)$ is a modified base predictor,
    and similarly, $\widehat{R}^\star\bigl(\widehat{A}^\star(\cdot, \cdot; \bv), \cdot; \bv\bigr) : (\bx, \ba) \mapsto \widehat{R}\bigl(\widehat{A}^\star(\bx, \ba; \bv), \min(\ba, \bv)\bigr)$
    is its reconciled version.

    Second, instead of considering the modified 
    reconciled predictor evaluated
    with infrastructure activity $\bv^{(k)}$,
    we introduce a reconciled predictor 
    with random infrastructure activity 
    \[
    \widetilde{g}^{(k)}(\cdot, \cdot) = \mathcal{G}\bigl((\bu^{(k)}, \bv^{(k)}, \bz^{(k)})_{k \in \Calib}, (\bx^{(t)}, \ba^{(t)}, \by^{(t)}); \bnu^{(k)}\bigr) = \widehat{R}^{\star}(\widehat{A}^{\star}(\cdot, \cdot ; \, {\bnu^{(k)}}), \cdot ; \, {\bnu^{(k)}}) \,,
    \]
    where $\bnu^{(k)}$ is drawn independently from 
    $\bigl\{(\bu^{(k)}, \bv^{(k)}, \bz^{(k)})_{k \in \Calib}, (\bx^{(t)}, \ba^{(t)}, \by^{(t)}) \bigr\}$,
    and $\mathcal{G}$ is a randomized algorithm\footnote{The randomization requires to adapt the argument
    from the proof of Theorem~\ref{th:marginalnested} to the case of randomized predictors, 
    which we omit for the sake of simplicity (see \citealp[Appendix~C]{zaffran2024predictive} for more details).} 
    that outputs a pre-determined function that is independent from the input data.
    Hence, setting $\bnu^{(k)} = \ba^{(k)}$ for each $k \in \Calib$, 
    which is a valid choice since $\ba^{(k)}$ is independent 
    from $\bigl\{(\bx^{(k)}, \, {\ba^{(t)}},\, \by^{(k)}) \bigr\}_{k \in \Calib \cup \{t\}}$     
    by Assumption~\ref{ass:iid}~and~\ref{ass:mask}, we can rewrite 
    the term involved in the third step of Algorithm~\ref{algo:CP-MNR-Nested*} as
    \begin{equation} \label{eq:rewriting}
    \#\biggl\{k \in \Calib : 
    s_i\Bigl(\bigl({\bu^{(k)}}, {\ba^{(t)}}\bigr), {\bz^{(k)}}; \, \widetilde{g}^{(k)}(\cdot, \cdot) \Bigr)
    < \, s_i\Bigl(\bigl({\bx^{(t)}}, {\ba^{(t)}}\bigr), \by; \, \widetilde{g}^{(k)}(\cdot, \cdot) \Bigr) \biggr\} \,.
    \end{equation}
    As a consequence, this specific instance of Algorithm~\ref{algo:randomCP-MNR-Nested*} 
    is equivalent to Algorithm~\ref{algo:CP-MNR-Nested*}.

        \begin{algorithm}[t]
	\caption{\label{algo:randomCP-MNR-Nested*} CP--MNR-Nested$^\star$ at level $1-\alpha$ with random infrastructure activity patterns}
	\begin{algorithmic}
    \STATE\textbf{Input:} Predictor $\widehat{R}^{\star}\bigl(\widehat{A}^{\star}(\cdot, \cdot; \cdot), \cdot; \cdot\bigr)$;
    dataset $\Bigl\{ {\bigl(\bu^{(k)}, \, \bv^{(k)}, \, \bz^{(k)}\bigr)} : k \in \Calib \Bigr\}$; 
    test point $\bigl({\bx^{(t)}}, \, {\ba^{(t)}}\bigr)$.
    \vspace{0.25em}
	\STATE 1: \textbf{for} $k \in \Calib$ \textbf{do} consider a {random} augmented nested infrastructure activity pattern \\
    \vspace{0.25em}
    \STATE \qquad {Randomly draw  $\ba^{(k)}$ independently from 
        $\bigl\{\bigl({\bu^{(k)}}, {\bv^{(k)}}, {\bz^{(k)}}\bigr)_{k \in \Calib}, \bigl({\bx^{(t)}}, {\ba^{(t)}}, {\by^{(t)}}\bigr) \bigr\}$}\\
    \vspace{0.15em}
        \STATE \qquad {Define accordingly a modified reconciled predictor, \quad
       $\widetilde{g}^{(k)}(\cdot, \cdot) = \widehat{R}^{\star}\Bigl(\widehat{A}^{\star}\bigl(\cdot, \cdot; \ba^{(k)}\bigr), \cdot; \ba^{(k)}\Bigr)$}
       \STATE \textbf{end for} \\
    \vspace{0.25em}
       \STATE 3: \textbf{for} $i \in [m]$ \textbf{do} build the prediction set $\C_i\bigl({\bx^{(t)}}, {\ba^{(t)}}\bigr)$ defined as \\
       \STATE \qquad $\cC_i\bigl({\bx^{(t)}}, {\ba^{(t)}}\bigr) = \Biggl\{ \Pi_i(\by) \in \R : (1-\alpha)\bigl(\#\Calib+1\bigr) >$ \\
       \vspace{-0.5em}
       \STATE \hspace{9.3em} $\displaystyle{
        \#\biggl\{k \in \Calib : s_i\Bigl(\bigl({\bu^{(k)}}, {\ba^{(t)}}\bigr), {\bz^{(k)}}; \, \widetilde{g}^{(k)}(\cdot, \cdot) \Bigr)
        < \, s_i\Bigl(\bigl({\bx^{(t)}}, {\ba^{(t)}}\bigr), \by; \, \widetilde{g}^{(k)}(\cdot, \cdot) \Bigr) \biggr\}
       } \Biggr\}$ \\
        \vspace{.25em}
       \STATE \textbf{end for} \\
    \STATE \textbf{Output:} Predictions sets $\Bigl\{\cC_i\bigl({\bx^{(t)}}, {\ba^{(t)}}\bigr): i \in [m] \Bigr\}$.
\end{algorithmic}
\end{algorithm}

    \medskip

    \emph{Step 3: Conditional Coverage.}
    The same argument and notations as 
    in the proof of Theorem~\ref{th:marginalnested} conclude. 
    Indeed, since $\mathcal{G}$ treats symmetrically 
    all the observations within 
    the set of observations with
    substituted infrastructure activity, 
    we obtain that 
    \begin{align} \label{eq:symmetry2}
    \nonumber M^{\star}_i\bigl(\sigma(p),\sigma(q)\bigr) &= 
    \M^{\star}_i\bigl((\bx^{(k)}, {\ba^{(t)}}, \by^{(k)})_{k \in \Calib \cup \{t\}}; (\ba^{(k)})_{k \in \Calib} \bigr)(\sigma(p),\sigma(q)) \\
    &= \M^{\star}_i\bigl((\bx^{(\sigma(k))}, {\ba^{(t)}}, \by^{(\sigma(k))})_{k \in \Calib \cup \{t\}}; (\ba^{\sigma(k)})_{k \in \Calib} \bigr)(p,q) \,,
    \end{align}
    where $\M^{\star}_i$ the function that maps the observations with substituted infrastructure activity 
    and the true random infrastructure activity patterns to a comparison matrix $M^{\star}_i$, 
    based on the equivalence obtained in~\eqref{eq:rewriting}, 
    in a similar way as in the proof of Theorem~\ref{th:marginalnested}.

    Therefore, the independence of the true infrastructure activity patterns 
    $\bigl\{\ba^{(k)}\bigr\}_{k \in \Calib}$ from the observations with substituted infrastructure activity
    $\bigl\{(\bx^{(k)}, \, {\ba^{(t)}},\, \by^{(k)}) \bigr\}_{k \in \Calib \cup \{t\}}$,
    the exchangeability of the latter conditional on $\ba^{(t)}$, and 
    the expression obtained in~\eqref{eq:symmetry2},
    entails that, conditional on $\ba^{(t)}$, we have
    \[
    \bigl(M^{\star}_i\bigl(\sigma(p),\sigma(q)\bigr)\bigr)_{p,q \in \Calib \cup \{t\}} 
    \overset{d}{=} \M^{\star}_i\bigl((\bx^{(k)}, {\ba^{(t)}}, \by^{(k)})_{k \in \Calib \cup \{t\}}; (\ba^{(k)})_{k \in \Calib} \bigr)
    = M^{\star}_i \,.
    \]
    Finally, we conclude by using the same argument as in the proof of Theorem~\ref{th:marginalnested},
    and get the infrastructure activity-conditional guarantee in~\eqref{eq:strong-conditional}, as announced.
\end{proof}

\begin{remark} \label{rk:assertion}
    The conditional coverage guarantee of Theorem~\ref{th:strong-conditional} 
    holds for any reconciliation Algorithm $\cR$, 
    including the trivial one that does not perform any reconciliation step.
    Therefore, the assertion of Theorem~\ref{th:strong-conditional} is
    not that conditional coverage is guaranteed 
    thanks to the infrastructure-aware reconciliation step,
    but rather that conditional coverage is a 
    consequence of the assumptions of Theorem~\ref{th:strong-conditional}, 
    and is preserved even
    when incorporating such a reconciliation step.
\end{remark}
% \begin{remark}
% % --- Imputation is mandatory, not forecast reconciliation ---
% In MDA, imputation is mandatory, 
% as standard learning algorithms cannot operate directly on incomplete inputs.
% By contrast, in MNR, infrastructure-aware reconciliation is optional:
% it can improve forecast accuracy, but it is not required to produce forecasts.
% In particular, the conditional coverage guarantee of CP-MNR-Nested$^{\star}$ 
% in Theorem~\ref{th:strong-conditional} remains valid even 
% in the absence of any reconciliation step.
% While this may appear surprising, it follows directly from Assumptions~\ref{ass:iid} and~\ref{ass:mask}, 
% which imply conditional validity for forecasts that do not depend on the infrastructure activity pattern.
% The significance of Theorem~\ref{th:strong-conditional} is therefore that 
% conditional validity is preserved even when incorporating an infrastructure-aware reconciliation step.
% \end{remark}

\section{Implementation of the Algorithms} \label{app:implementation}
This Supplementary Material provides a detailed description 
of the practical implementation of the proposed algorithms.
First, Section~\ref{app:implementation:predict-then-reconcile} describes
the construction of the base and reconciled predictors, 
involved in the first step of both Algorithms~\ref{algo:CP-MNR} and~\ref{algo:CP-MNR-Nested*},
from Algorithm~\ref{algo:predict-then-reconcile}, 
which is summarized in Algorithm~\ref{algo:predict-then-reconcile},
Then, Section~\ref{app:implementation:CP-MNR-Nested*}
presents the implementation of the CP-MNR-Nested$^{\star}$ procedure.
Finally, Section~\ref{app:implementation:Adaptive-CP-MNR}
is devoted to the implementation of the Adaptive CP-MNR procedure.

\subsection{Construction of the base and reconciled predictors}
\label{app:implementation:predict-then-reconcile}
The construction of the base and reconciled predictors
is summarized in Algorithm~\ref{algo:predict-then-reconcile}.
The formulation encompasses 
the three predictors introduced 
in Section~\ref{sec:forecastreconciliation}
and considered in the experiments of Section~\ref{sec:experiments},
namely the Direct (no reconciliation), OLS and rOLS predictors.

In Algorithm~\ref{algo:predict-then-reconcile}, 
the learning algorithm $\cA$ (e.g., independent
instances of a regression algorithm for each node) 
is fitted on the subset indexed by $\Train$ 
of the available data $\D$, and outputs 
a predictor $\widehat{A} \in \mathcal{F}$,
where $\mathcal{F} = \{f : \X \times \{0,1\}^m \to \R^m\}$.
% --- Forecast reconciliation ---
Given the base predictor $\widehat{A}$,  
the reconciliation algorithm $\cR$ (e.g., rOLS)
is fitted on the same subset of $\D$ 
indexed by $\Train$, and outputs 
a reconciliation function $\widehat{R} : \R^m \times \{0,1\}^m \to \R^m$.
The combination of a base predictor and a reconciliation function
yields a reconciled predictor $\widetilde{A} \in \mathcal{F}$, defined as
$\widetilde{A} : (\bx, \ba) \longmapsto \widehat{R}\bigl(\widehat{A}(\bx, \ba), \ba\bigr)$.

The reconciliation algorithms considered in this article
are constant, in the sense that they 
do not depend on the training data. 
The above formalism thus covers these three cases,
but also allows for more general reconciliation algorithms 
that may depend on the training data.
For instance the MinT reconciliation procedure \citep{wickramasuriya2019optimal} 
consists of a projection matrix that is computed from the data by
computing the empirical covariance matrix of the forecast error.
The explicit form of the reconciliation function 
$\widehat{R}$ for the three considered strategies is as follows.

% --- Concrete reconciliation strategies ---
The Direct strategy leaves base forecasts unchanged. 
In this case, the reconciliation algorithm $\cR$ is constant and outputs
\begin{equation} \label{eq:Direct}
    \widehat{R} : \bigl(\widehat{\by}, \ba\bigr) \longmapsto \widehat{\by} \,. \tag{Direct}
\end{equation}
The OLS strategy reconciles forecasts by projecting the base forecasts onto the coherent subspace $\Im(H)$.
In this case, let $P = H\bigl(H^{\transp}H\bigr)^{-1}H^{\transp}$ be the orthogonal projection 
onto $\Im(H)$, and consider the constant reconciliation algorithm $\cR$ that outputs 
\begin{equation} \label{eq:OLS}
    \widehat{R} : \bigl(\widehat{\by}, \ba\bigr) \longmapsto P \, \widehat{\by} \,. \tag{OLS}
\end{equation}
The rOLS strategy reconciles forecasts by projecting the base forecasts onto 
the refined coherent subspace $\Im\bigl(H^{(\ba)}\bigr)$.
In this case, the reconciliation algorithm $\cR$ is constant and outputs the reconciliation function
\begin{equation} \label{eq:rOLS}
    \widehat{R} : \bigl(\widehat{\by}, \ba\bigr) \longmapsto P^{(\ba)} \, \widehat{\by} \,, \tag{rOLS} 
\end{equation}
where  $P^{(\ba)} = H^{(\ba)}\bigl((H^{(\ba)})^{\transp}H^{(\ba)}\bigr)^{\dagger}(H^{(\ba)})^{\transp}$ 
with $(\cdot)^{\dagger}$ being the Moore--Penrose pseudo-inverse.

\begin{algorithm}[t]
	\caption{\label{algo:predict-then-reconcile} Construction of the predictors}
	\begin{algorithmic}
		\STATE\textbf{Input:} Structural matrix $H$; dataset $\D = \bigl\{ \bigl(\bx^{(k)}, \ba^{(k)}, \by^{(k)}\bigr) : k \in \Train \cup \Calib \bigr\}$;
        set of indices $\Train$ for training; learning algorithm $\cA$; reconciliation algorithm $\cR$.
		\STATE 1: Fit $\cA$: $\widehat{A}(\cdot, \cdot) \leftarrow \cA\Bigl(\bigl\{ \bigl(\bx^{(k)}, \ba^{(k)}, \by^{(k)}\bigr) : k \in \Train \bigr\}\Bigr) $ \\
		\STATE 2: Fit $\cR$: $\widehat{R}(\cdot, \cdot) \leftarrow \cR\Bigl(\Bigl\{ \bigl(\widehat{A}\bigl(\bx^{(k)}, \ba^{(k)}\bigr), \ba^{(k)}, \by^{(k)}\bigr) : k \in \Train\Bigr\}, \, H\Bigr)$ \\
              \textbf{Output:} Fitted predictor $\widehat{A}$, fitted reconciliation function $\widehat{R}$, and fitted reconciled predictor $\widetilde{A} = \widehat{R}(\widehat{A}(\cdot, \cdot), \cdot)$.
	\end{algorithmic}
\end{algorithm}

\subsection{Implementation of CP-MNR-Nested$^{\star}$}
\label{app:implementation:CP-MNR-Nested*}
First, the computation of the 
conformal prediction set in
Algorithm~\ref{algo:CP-MNR-Nested*}
is made explicit from a practical standpoint. 
Then, several practical aspects of the implementation are discussed,
including the nature of the resulting prediction set
and the choice of the calibration indices.

This implementation follows the efficient implementation
of cross-conformal prediction sets (\citealp[Algorithm 1]{gupta2022nested}).
Without loss of generality, the conformal prediction sets 
produced by CP-MNR-Nested$^{\star}$ are of the form
\begin{equation} \label{eq:general_form}
    \C = \Bigl\{ c \in \R : \#\bigl\{ k : r_k < |c_k - c| \bigr\} < (1-\alpha)(N + 1) \Bigr\} \,, 
\end{equation}
for some $N \in \mathbb{N}$, some non-negative scalars $r_1, \ldots, r_N$, and scalars $c_1, \ldots, c_N$.

The above expression does not yield 
a direct implementation,
but can be reformulated in terms of a family of overlapping intervals.
For each $k \in [N]$, let $I_k = [c_k - r_k, c_k + r_k]$,
and note that $\{r_k < |c_k - c|\} = \{c \notin I_k\}$ for all $c \in \R$.
Thus, it follows that
\begin{equation} \label{eq:sweep_line_form}
    \C = \Bigl\{ c \in \R : \#\{k : c \in I_k\}  > \alpha(N+1) - 1 \Bigr\} \,.
\end{equation}
From~\eqref{eq:sweep_line_form}, 
the computation of the conformal set
reduces to identifying the regions of $\R$ 
where the function $c \mapsto \sum_{k=1}^N \indic_{c \in I_k}$ 
exceeds a fixed threshold. 
This function is piecewise constant with changes only at the interval endpoints 
$\{c_k - r_k,, c_k + r_k\}_{k=1}^N$. 
Consequently, sorting these $2N$ breakpoints 
and performing a single pass that updates the active interval count, 
thereby yield all connected components of the set $\mathcal{C}$.

\begin{remark}
    The computation of the prediction sets itself
    in CP-MNR-Nested$^{\star}$ is not computationally intensive.
    On the contrary, the computation of the reconciled forecasts
    for each calibration point and each test point
    can be computationally expensive, as it requires 
    to compute a projection matrix for each calibration point 
    and each test point.
\end{remark}

% --- Discussion of the calibration indices ---
\begin{remark}
In Algorithm~\ref{algo:CP-MNR-Nested*}, 
the practitioner can pick any subset 
of calibration indices $\tCalib \subseteq \Calib$
to mitigate the issue that would arise from 
the potential sparsity of the 
augmented infrastructure activity patterns.
\end{remark}

\begin{remark}
In the experiments, the subsampling is performed by selecting
the calibration observations with infrastructure activity vector
having at most $600$ different components with respect to the test one.
\end{remark}

\begin{remark}
The CP-MNR-Nested$^{\star}$ procedure does not 
necessarily yield intervals.
For simplicity, the convex hull of the resulting prediction set
is considered in Section~\ref{sec:experiments}.
\end{remark}

\subsection{Implementation of Adaptive-CP-MNR}
\label{app:implementation:Adaptive-CP-MNR}
\begin{algorithm}[t]
	\caption{\label{algo:AdaptiveCP-MNR} Adaptive CP-MNR at level $1-\alpha$ with step size $\gamma$}
	\begin{algorithmic}
    \STATE\textbf{Input:} Dataset $\D = \Bigl\{ \bigl(\bx^{(k)}, \ba^{(k)}, \by^{(k)}\bigr) : k \in \Train \cup \Calib \Bigr\}$; 
        algorithms $\cA$ and $\cR$; 
        set of indices $\Test$.
	\STATE \textbf{Initialization:} Set $\tau = 1$; $\D^{(\tau)} = \D$; $\Train^{(\tau)} = \Train$; $\Calib^{(\tau)} = \Calib$; $\balpha^{(\tau)} = (\alpha, \ldots, \alpha)^{\top} \in [0,1]^m$. \\
    \vspace{0.5em}
    \STATE \textbf{for} $t \in \Test$ \textbf{do}  \\
    \STATE \qquad 1: \textbf{for} $i \in [m]$ \textbf{do} run CP-MNR at level $1-\alpha_i^{(\tau)}$ on $\D^{(\tau)}$ and get $\cC_i^{(t)} = \C_i\bigl(\bx^{(t)}, \ba^{(t)}\bigr)$ \\
    \vspace{0.15em}
    \STATE \qquad \qquad 1.1: Observe $\by^{(t)}$ and compute \quad \qquad $\err_i^{(\tau)} = \indic_{y_i^{\!(t)} \notin \cC_i^{(t)}} \times \indic_{a^{\!(t)}_i= \, 1}$
    \vspace{0.05em}
    \STATE \qquad \qquad 1.2: Update miscoverage level  \qquad $\alpha_i^{(\tau+1)} = \alpha_i^{(\tau)} + \gamma \times \err_i^{(\tau)}$
    \vspace{0.15em}
    \STATE \qquad \textbf{end for} \\
    \vspace{0.15em}
    \STATE \qquad 2: Update dataset $\D^{(\tau+1)} = \D^{(\tau)} \cup \bigl\{ \bigl(\bx^{(t)}, \ba^{(t)}, \by^{(t)}\bigr) \bigr\}$
    \STATE \qquad 3: Split  $\D^{(\tau+1)}$ into $\Train^{(\tau+1)}$ and $\Calib^{(\tau+1)}$ \\
    \STATE \qquad 4: Update counter $\tau = \tau + 1$.   
    \STATE \textbf{end for} \\
    \STATE \textbf{Output:} Predictions sets $\Bigl\{\cC_i^{(t)} : i \in [m], \, t \in \Test \Bigr\}$.
\end{algorithmic}
\end{algorithm}
% --- Formalization of the adaptive extension ---
This section describes Adaptive-CP-MNR,
the online extension of CP-MNR
introduced in Section~\ref{sec:adaptive}.
The procedure, summarized in Algorithm~\ref{algo:AdaptiveCP-MNR},
is a straightforward combination of CP-MNR 
from Algorithm~\ref{algo:CP-MNR} 
with the adaptive conformal inference framework of \citet{gibbs2021adaptive}.

In Algorithm~\ref{algo:AdaptiveCP-MNR}, a 
loop is performed over the test points $t \in \Test$ in chronological order
and a counter $\tau \geq 1$ is used to denote the current round of the loop.
The first step of the procedure consists in running 
CP-MNR with node-specific coverage levels that are updated sequentially.
The sequence is initialized at each node by the nominal level $\alpha \in [0, 1]$.
For each node $i \in [m]$, we perform CP-MNR at the current level $1-\alpha_i^{(\tau)}$,
instead of the fixed nominal level $1-\alpha$, 
which yields the prediction set $\C_i^{(t)}$.
When the true value $y_i^{(t)}$ is observed, 
the miscoverage level is updated with the following rule:
\begin{equation}\label{eq:ACI_update}
\forall \tau \geq 1 \,, \qquad  \alpha_i^{(\tau+1)} = \alpha_i^{(\tau)} + \gamma \Bigl(\alpha - \indic_{y_i^{(t)} \notin \cC_i^{(t)}}\Bigr) \, \indic_{a^{(t)}_i = \,1}  \,,
\end{equation}
where $\indic_{(\cdot)}$ is the indicator function and 
$\gamma > 0$ is a step size parameter.
The step size $\gamma$ controls the adaptation speed, 
and is set to $\gamma = 0.02$ in our experiments.
A small value of $\gamma$ yield a slow adaptation to temporal drifts, 
while a large value of $\gamma$ may lead to oscillations 
outside of the $[0, 1]$ interval \citep{gibbs2021adaptive}.
The second, third and fourth steps of the procedure 
consist in updating the dataset, 
the training and calibration sets of indices, 
and the counter $\tau$ for the next round of the loop.
\begin{remark}
% --- Comments on the update rule ---
By convention, $\alpha_i^{(\tau)}\leq0$ yields the full space and $\alpha_i^{(\tau)}\geq1$ the empty set.
\end{remark}

% --- Theoretical guarantee ---
Adaptive CP-MNR inherits 
the classic asymptotic coverage guarantee 
of adaptive conformal inference (\citealp[Proposition 4.1]{gibbs2021adaptive}).
Precisely, at each node $i \in [m]$, 
the only difference with the original adaptive conformal inference framework 
is that the miscoverage level is only updated 
if the node is active, i.e., if $a_i^{(t)} = 1$. 
Therefore, the following proposition is a direct application of 
the classic asymptotic coverage guarantee of ACI to the Adaptive CP-MNR procedure.

\begin{proposition}
    The Adaptive CP-MNR procedure of Algorithm~\ref{algo:AdaptiveCP-MNR} satisfies
    \begin{equation} \label{eq:propositionACI}
        \forall i \in [m]\,, \qquad \biggl|\frac{1}{T_i} \sum_{t \in \Test} 
        \indic_{y_i^{(t)} \notin \cC_i^{(t)} } \indic_{a_i^{\!(t)} = \, 1}  - \alpha \biggr| \leq \frac{\max\{\alpha, 1-\alpha\} + \gamma}{\gamma \, T_i} \as,
    \end{equation}
    where ${T_i = \sum_{t \in \Test} \indic_{a^{\!(t)}_i = \, 1}}$ 
    denotes the number of active test points for node $i$.
    In particular, if $T_i$ goes to infinity as $\#\Test$ goes to infinity, 
    then \eqref{eq:propositionACI} entails that
    \begin{equation}
        \lim_{T_i \to \infty} \frac{1}{T_i} \sum_{t \in \Test} \indic_{y_i^{(t)} \in \cC_i^{(t)}} \indic_{a_i^{\!(t)} = \, 1} = 1-\alpha \as
    \end{equation}
\end{proposition}
\begin{proof}
    The proof is a direct application of \citealp[Proposition 4.1]{gibbs2021adaptive}.

    We first get by contradiction that $\alpha_i^{(\tau)} \in [-\gamma, 1+\gamma]$ 
    almost surely for all $\tau \geq 1$. Then use a telescopic sum argument to 
    obtain the bound in~\eqref{eq:propositionACI}.
\end{proof}

\section{Additional Numerical Results} \label{app:additional_results}

    \subsection{Comparison of the reconciliation procedures 
    in terms of efficiency}
    The performance of the different reconciliation procedures 
    in terms of prediction interval length are reported in Table~\ref{tab:winrate_efficiency},
    and confirm the choice of rOLS as the reconciliation procedure.
    The average length of the prediction intervals is computed
    conditionally on the node being active for each nodes, 
    and the win rate of each reconciliation procedure 
    is computed with respect to the Direct baseline.
    Hence, a win rate significantly above $50\%$ indicates 
    that the reconciliation procedure yields shorter prediction intervals than the Direct baseline,
    advocating for its use in practice if the validity rates are similar.
    The results are gathered by node type to exhibit differences in performance across the hierarchy.

    OLS yields, at best, a slight improvement over the Direct baseline: 
    it shortens intervals at the national and regional levels, 
    but it is not consistently beneficial at the station level.
    By contrast, rOLS is uniformly better, 
    with a clear efficiency gain across the hierarchy;
    for instance with a win rate of $60\%$, $61.6\%$, 
    and $59.4\%$ at the national, regional, 
    and station levels, respectively, for CP-MNR.
    Note that CP-MNR-Nested$^{\star}$ is not included in the comparison, 
    since its construction is specific to the rOLS setting
    (the infrastructure activity vector is not taken in 
    consideration in Direct and OLS).
    At the station level, the highest efficiency gain is obtained for 
    the CP procedure, with a win rate of $81\%$. 
    This efficiency gain is actually due to the fact
    that CP does not achieve the nominal coverage level, 
    as discussed next.

    \subsection{Comparison of the reconciliation procedures 
    in terms of validity}

    % --- It does not come at the cost of coverage (except for CP) ---
    Any improvement in terms of efficiency is only meaningful 
    if it does not come at the cost of coverage.
    In this regard, Tables~\ref{tab:validity_all} and~\ref{tab:validity_relaxed_all} 
    reports the validity rates (as defined in Section~\ref{sec:validity}) 
    at the nominal level $1-\alpha = 90\%$ 
    and at the relaxed level $1-2\alpha = 80\%$, respectively.
    Both tables demonstrate that the rOLS reconciliation procedure
    does not compromise the validity of the prediction intervals.
    It stems from improved point forecast accuracy,
    both on the calibration and test sets, 
    recovering the theoretical insights from~\citet{principato2024conformal}.
    In essence, they show that, under the assumption of i.i.d.\ and 
    elliptical non-conformity scores, forecast reconciliation improves 
    the efficiency of conformal prediction sets without compromising their validity.
    The only exception is the CP procedure, 
    which does not subsample the calibration data according to node activity.
    Therefore, the rOLS forecast errors are set to zero 
    in the calibration set for inactive nodes, 
    leading to a lack of coverage conditional 
    on the corresponding node being active in the test set.
    
    \clearpage
    \begin{table}[t]
        \centering
        \caption{Win rate of OLS and rOLS with respect to Direct
        in terms of prediction interval average lengths.
        Results are reported as $\operatorname{mean} \pm 1.96 \times \operatorname{standard \, error}$.
        % The best-performing reconciliation procedure is shown in bold.
        }

        \begin{tabular}{llccc}
        \toprule
        Node type & Reconciliation & CP & CP-MNR & Adaptive CP-MNR \\
        \midrule

        \multirow{2}{*}{National}
        & OLS
        & 50.0\% $\pm$ 16.7\%
        & 50.0\% $\pm$ 16.7\%
        & 80.0\% $\pm$ 13.3\% \\

        & rOLS
        & \textbf{60.0\% $\pm$ 16.3\%}
        & \textbf{60.0\% $\pm$ 16.3\%}
        & \textbf{90.0\% $\pm$ 10.0\%} \\

        \addlinespace

        \multirow{2}{*}{Regional}
        & OLS
        & 55.9\% $\pm$ 2.2\%
        & 55.9\% $\pm$ 2.2\%
        & 46.9\% $\pm$ 1.5\% \\

        & rOLS
        & \textbf{61.6\% $\pm$ 2.4\%}
        & \textbf{61.6\% $\pm$ 2.4\%}
        & \textbf{47.5\% $\pm$ 1.5\%} \\

        \addlinespace

        \multirow{2}{*}{Station}
        & OLS
        & 48.0\% $\pm$ 0.4\%
        & 42.4\% $\pm$ 0.3\%
            & 37.7\% $\pm$ 0.3\% \\

            & rOLS
            & \textbf{81.0\% $\pm$ 0.2\%}
            & \textbf{59.4\% $\pm$ 0.3\%}
            & \textbf{57.4\% $\pm$ 0.2\%} \\

            \bottomrule
            \end{tabular}

            \label{tab:winrate_efficiency}
        \end{table}

    \begin{table}[t]
    \centering
    \caption{Validity rates at level $1-\alpha = 90\%$ 
    for the different reconciliation procedures.
    Results are reported as $\operatorname{mean} \pm 1.96 \times \operatorname{standard \, error}$.
    The best-performing reconciliation procedure is shown in bold.}
    \label{tab:validity_all}
    \begin{tabular}{llccc}
    \toprule
    Node type & Reconciliation & CP & CP-MNR & Adaptive CP-MNR \\
    \midrule

    \multirow{3}{*}{National}
    & Direct
    & 10.0\% $\pm$ 10.0\%
    & 10.0\% $\pm$ 10.0\%
    & \textbf{0.0\% $\pm$ 0.0\%} \\

    & OLS
    & 30.0\% $\pm$ 15.3\%
    & 30.0\% $\pm$ 15.3\%
    & \textbf{0.0\% $\pm$ 0.0\%} \\

    & rOLS
    & \textbf{40.0\% $\pm$ 16.3\%}
    & \textbf{40.0\% $\pm$ 16.3\%}
    & \textbf{0.0\% $\pm$ 0.0\%} \\

    \addlinespace

    \multirow{3}{*}{Regional}
    & Direct
    & \textbf{67.2\% $\pm$ 0.8\%}
    & \textbf{67.2\% $\pm$ 0.8\%}
    & \textbf{62.8\% $\pm$ 1.4\%} \\

    & OLS
    & 53.8\% $\pm$ 1.7\%
    & 53.8\% $\pm$ 1.7\%
    & 49.1\% $\pm$ 1.6\% \\

    & rOLS
    & 51.9\% $\pm$ 1.2\%
    & 51.9\% $\pm$ 1.2\%
    & 47.8\% $\pm$ 0.9\% \\

    \addlinespace

    \multirow{3}{*}{Station}
    & Direct
    & \textbf{55.6\% $\pm$ 0.3\%}
    & \textbf{56.9\% $\pm$ 0.2\%}
    & \textbf{61.8\% $\pm$ 0.3\%} \\

    & OLS
    & 51.6\% $\pm$ 0.2\%
    & 55.8\% $\pm$ 0.1\%
    & 60.4\% $\pm$ 0.3\% \\

    & rOLS
    & 38.1\% $\pm$ 0.4\%
    & 56.2\% $\pm$ 0.4\%
    & 57.6\% $\pm$ 0.6\% \\
    \bottomrule
    \end{tabular}
    \end{table}

\begin{table}[t]
\centering
\caption{Validity rates at the relaxed level $1-2\alpha = 80\%$ 
    for the different reconciliation procedures.
    Results are reported as $\operatorname{mean} \pm 1.96 \times \operatorname{standard \, error}$.
    The best-performing reconciliation procedure is shown in bold.}
\label{tab:validity_relaxed_all}
\begin{tabular}{llccc}
\toprule
Node type & Reconciliation & CP & CP-MNR & Adaptive CP-MNR \\
\midrule

\multirow{3}{*}{National}
& Direct
& \textbf{100.0\% $\pm$ 0.0\%}
& \textbf{100.0\% $\pm$ 0.0\%}
& \textbf{100.0\% $\pm$ 0.0\%} \\

& OLS
& \textbf{100.0\% $\pm$ 0.0\%}
& \textbf{100.0\% $\pm$ 0.0\%}
& \textbf{100.0\% $\pm$ 0.0\%} \\

& rOLS
& \textbf{100.0\% $\pm$ 0.0\%}
& \textbf{100.0\% $\pm$ 0.0\%}
& \textbf{100.0\% $\pm$ 0.0\%} \\

\addlinespace

\multirow{3}{*}{Regional}
& Direct
& 97.2\% $\pm$ 0.7\%
& 97.2\% $\pm$ 0.7\%
& \textbf{100.0\% $\pm$ 0.0\%} \\

& OLS
& \textbf{98.8\% $\pm$ 0.5\%}
& \textbf{98.8\% $\pm$ 0.5\%}
& \textbf{100.0\% $\pm$ 0.0\%} \\

& rOLS
& \textbf{98.8\% $\pm$ 0.5\%}
& \textbf{98.8\% $\pm$ 0.5\%}
& \textbf{100.0\% $\pm$ 0.0\%} \\

\addlinespace

\multirow{3}{*}{Station}
& Direct
& \textbf{88.7\% $\pm$ 0.4\%}
& 94.2\% $\pm$ 0.3\%
& 98.3\% $\pm$ 0.1\% \\

& OLS
& 84.4\% $\pm$ 0.4\%
& \textbf{94.6\% $\pm$ 0.3\%}
& \textbf{98.5\% $\pm$ 0.1\%} \\

& rOLS
& 67.7\% $\pm$ 0.3\%
& 93.6\% $\pm$ 0.3\%
& 98.0\% $\pm$ 0.1\% \\

\bottomrule
\end{tabular}
\end{table}

\vspace*{\fill}

\clearpage

\end{document}

%% file: aoas-article-tikz.tex
\usetikzlibrary{calc}
\usetikzlibrary{positioning, fit}
\tikzset{
    softbox/.style={
        draw,
        rounded corners,
        inner sep=2pt,
        font=\small, 
        minimum width=1.cm,
        minimum height=0.55cm,
    }
}
\begin{tikzpicture}[
    node distance=1.6cm and 2.2cm,
    box/.style={draw, rectangle, rounded corners, align=center, minimum width=2.cm, minimum height=0.8cm},
    smallbox/.style={draw, rectangle, align=center, minimum width=3.2cm, minimum height=0.9cm},
    arrow/.style={->, thick, shorten >=2pt, shorten <=2.75pt},
    dashedbox/.style={draw, dashed, rounded corners, inner sep=0.5cm}
]

% ---------------- Left part ----------------
\node[box] (split) {Split historical data $\D$};

\node[right=2.3cm of split] (calset) {\!$\Bigl\{\bigl(\bx^{(k)}, \ba^{(k)}, \by^{(k)}\bigr) : k \in \Calib\Bigr\}$};

\node[below=1.35cm of split] (trainset) {$\Bigl\{\bigl(\bx^{(k)}, \ba^{(k)}, \by^{(k)}\bigr) : k \in \Train\Bigr\}$};

\node[box, below= 1.35cm of trainset] (base) {Base\\Predictor $\widehat{A}$};

\node[box, right = 0.75cm of base] (recon) {Reconciled\\Predictor $\widetilde{A}$};

% Arrows left part
\draw[arrow, dashed] (split) -- (calset);
\draw[arrow, dashed] (split) -- (trainset);
\draw[arrow] ($(trainset)+(0cm,-0.2cm)$) -- (base);
\draw[arrow] (base) -- (recon);
\draw[arrow, dotted] ($(trainset)+(0.85cm,-0.25cm)$) -- (recon);

\node[font=\scriptsize, gray] at (split.south) [xshift=-0.8cm, yshift=-0.75cm] {Training set};
\node[font=\scriptsize, gray] at (split.east) [xshift=1.0cm, yshift=+0.25cm] {Calibration set};

\node[font=\scriptsize, gray] at (base.east) [xshift=-0.025cm, yshift=+0.75cm] {Training phase};
\node[font=\scriptsize, gray] at (trainset.east) [xshift=1.1cm, yshift=+0.1cm] {Predict};

% ---------------- Right top ----------------
\node[right=1.15cm of calset] (htop) {$\Bigl\{\bigl(\bx^{(t)}, \ba^{(t)}\bigr)\Bigr\}$};

% ---------------- Inside conformal step ----------------
\node[softbox, right=0.75cm of recon, yshift=2.65cm] (S1b) {$\Bigl\{s_1^{(k)} : k \in \Calib \Bigr\}$};
\node[softbox, below=0.5cm of S1b] (SNb) {$\Bigl\{s_m^{(k)} : k \in \Calib \Bigr\}$};

\node[softbox, right=0.85cm of S1b, yshift=0cm] (S1h) {$q_1$};
\node[softbox, below=0.7cm of S1h] (SNh) {$q_m$};

% vertical dots
\node (vdots) at ($(S1b)!0.4!(SNb)$) {$\vdots$};
\node at ($(S1h)!0.4!(SNh)$) {$\vdots$};

\coordinate[left=1.25cm of vdots, yshift=-0.1cm] (recmid);
\draw[arrow] (recon.north) -- ++(0,0.5) |- (recmid);

\coordinate[below=-1.3cm of vdots] (reccal);
\draw[arrow, dashed] ($(calset.south)+(0.18cm,0.15cm)$) -- ($(reccal)+(0+0.35cm,0cm)$);

% \node[font=\scriptsize, gray] at (calset.south) [xshift=-0.215cm, yshift=-0.4cm] {Calibration \, scores};
\node[font=\scriptsize, gray] at (calset.south) [xshift=-0.875cm, yshift=-0.45cm] {Compute scores};
\node[font=\scriptsize, gray] at (calset.south) [xshift=2.4cm, yshift=-0.75cm] {Get quantiles};

% \node[font=\tiny, gray] at (calset.south) [xshift=4.55cm, yshift=-0.05cm] {(for all candidate values $\by$)};

% outputs
\node[right=0.75cm of S1h] (C1) {$\Bigl[\widetilde{y}_1^{\,(t)} \!\!\pm \, q_1\Bigr]$};
\node[below=0.45cm of C1] (CN) {$\Bigl[\widetilde{y}_{m}^{\,(t)} \!\!\pm \, q_m\Bigr]$};
\node (vdots) at ($(C1)!0.4!(CN)$) {$\vdots$};

\draw[arrow, dashed] ($(htop.south)+(0-0.15cm,0.15cm)$) -- ($(C1)+(-0.335,0.35cm)$);

\node[inner sep=0pt] (padbottom) at ($(CN.south)+(0,-0.25cm)$) {};
\node[inner sep=0pt] (padtop) at ($(C1.north)+(0,0.25cm)$) {};

\draw[arrow] (S1b) -- (S1h);
\draw[arrow] (SNb) -- (SNh);

\draw[arrow] (S1h) -- (C1);
\draw[arrow] (SNh) -- (CN);

% \draw[gray, opacity=0.5, very thick]
%     ($(calset.east)!0.45!(htop.west)+(0,0.4cm)$) --
%     ($(calset.east)!0.45!(htop.west)+(0,-3.5cm)$);
% \node[gray, font=\scriptsize] at (S1b.north) [yshift=0.25cm] {Calibration scores};

\node[font=\scriptsize, gray] at (CN.south) [xshift=0.25cm, yshift=-0.4cm] {Predict};
\draw[arrow] ($(recon.east)+(0cm, -0.cm)$) -- ++(5cm,0.) -| ($(CN.south)+(-0.335cm,0cm)$);

\coordinate (rightfit) at ($(C1.east)+(-0.25cm,0)$);

% ---------------- Conformal block ----------------
\node[dashedbox, fit=(S1b)(SNh)(rightfit)(padbottom)(padtop)] (conf) {};
\node[font=\normalsize] at (conf.south) [xshift = -1.2cm, yshift=0.3cm] {Component-wise conformal step};

\end{tikzpicture}